\documentclass[11pt]{iopart}

\expandafter\let\csname equation*\endcsname=\relax 
\expandafter\let\csname endequation*\endcsname=\relax 

\usepackage[english]{babel}
\usepackage[utf8x]{inputenc}
\usepackage[T1]{fontenc}
\usepackage{setspace}
\usepackage{manfnt}
\usepackage{dsfont}
\usepackage[mathscr]{eucal}

\usepackage[a4paper,top=2.5cm,bottom=2.5cm,left=2cm,right=2cm,marginparwidth=1.5cm]{geometry}

\usepackage{amsthm}
\usepackage{amssymb}
\usepackage{amsmath}
\usepackage{braket}

\usepackage{soul}
\usepackage{bm}

\usepackage{graphicx}
\usepackage{caption}
\usepackage{subcaption}

\usepackage{marvosym}

\usepackage[colorlinks=true, allcolors=blue]{hyperref}

\usepackage{comment}

\usepackage[perpage]{footmisc} 

\DeclareMathOperator\supp{supp}
\DeclareMathOperator\gs{\ket{\Omega}}
\DeclareMathOperator\gsb{\bra{\Omega}}
\newcommand{\norm}[1]{\left\lVert#1\right\rVert}

\newcommand{\spctm}[1]{\mathbf{#1}}
\newcommand{\metric}{\mathrm{g}}

\newcommand{\alg}[1]{\mathscr{#1}}

\DeclareMathOperator{\BorR}{\text{Bor}(\mathbb{R})}
\newcommand{\bor}[1]{\mathtt{#1}}

\newcommand{\fwt}[1]{\mathfrak{#1}}

\newtheorem{claim}{Claim}[section]
\theoremstyle{definition}
\newtheorem{definition}{Definition}[section]
\newtheorem{example}{Example}[section]
\newtheorem{properties}{Properties}[section]
\newtheorem{remark}{Remark}[section]

\begin{document}

\title{Are Ideal Measurements of Real Scalar Fields Causal?}
\author{Emma Albertini \href{https://orcid.org/0000-0003-3099-642X}{\includegraphics[scale=0.066]{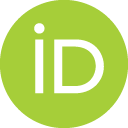}},
Ian Jubb 
\href{https://orcid.org/0000-0001-7339-2058}{\includegraphics[scale=0.066]{./ORCID.png}}
}

\vspace{10pt}
\address{Theoretical Physics Group, The Blackett Laboratory, Imperial College, Prince Consort Rd., London, SW7 2BZ, United Kingdom}
\address{School of Theoretical Physics, Dublin Institute for Advanced Studies, 10 Burlington Road, Dublin 4, D04 C932, Ireland}
\eads{\mailto{emma.albertini17@imperial.ac.uk}\\\mailto{ijubb@stp.dias.ie}}
\vspace{10pt}
\begin{indented}
\item \today
\end{indented}

\begin{abstract}
Half a century ago a local and (seemingly) causally consistent implementation of the projection postulate was formulated for local projectors in Quantum Field Theory (QFT) by utilising the basic property that spacelike local observables commute. This was not the end of the story for whether projective, or ideal measurements in QFT respect causality. In particular, the causal consistency of ideal measurements was brought into question by Sorkin 20 years later using a scenario previously overlooked. Sorkin's example, however, involved a non-local operator, and thus the question remained whether ideal measurements of local operators are causally consistent, and hence whether they are physically realisable. Considering both continuum and discrete spacetimes such as causal sets, we focus on the basic local observables of real scalar field theory --- smeared field operators --- and show that the corresponding ideal measurements violate causality, and are thus impossible to realise in practice. We show this using a causality condition derived for a general class of update maps for smeared fields that includes unitary kicks, ideal measurements, and approximations to them such as weak measurements. We discuss the various assumptions that go into our result. Of note is an assumption that Sorkin's scenario can actually be constructed in the given spacetime setup. This assumption can be evaded in certain special cases in the continuum, and in a particularly natural way in Causal Set Theory. In such cases one can then freely use the projection postulate in a causally consistent manner. In light of the generic acausality of ideal measurements, we also present examples of local update maps that offer causality-respecting alternatives to the projection postulate as an operationalist description of measurement in QFT.
\end{abstract}

\tableofcontents

\section{Introduction}

The nature of measurements in quantum theory has long been the subject of debate~\cite{AharonovRelativisticQM,AharonovRelativisticQM,Popescu_1994}. This debate becomes particularly fraught when the principles of relativity are introduced in Quantum Field Theory (QFT), as the concept of measurement, or wavefunction collapse, sits uneasily alongside relativistic causality~\cite{AharonovQFT,Jonsson_2014,EduardoOptics}. That said, most would agree that quantum theory must respect relativistic causality, for example, in discussions of Bell experiments where no signal can be sent between spacelike, or causally disconnected, agents.

The textbook description of measurement in quantum theory usually involves the \textit{projection postulate}, where the state is projected onto some eigenspace associated to the measured observable (this process is also called an \textit{ideal measurement}). For more general operations on quantum systems, regardless of the details of what is done physically to the system, we can describe the effect of our actions on the system through a completely positive map on the state (also called a \textit{quantum channel}). This operationalist theory of measurements and other update maps more generally --- where we only care to describe the effect of our actions on the state, and not the details of our experimental apparatus, etc. --- has been a challenge to copy over to relativistic QFT. In particular, in~\cite{Sorkin_impossible,Benincasa_2014,Jubb_2022} it was shown that not all completely positive maps respect causality.

Nonetheless, we have long discussed measurements in QFT in the form of scattering amplitudes, and validated the resulting predictions to extremely high accuracy. This description of measurements, however, is only an approximation, since scatterings involve preparing the state in the infinite past, as well as making a \textit{single} measurement of the outgoing state in the infinite future. What we seek here is a better description of what is going on. More specifically, we want an operationalist description of measurements and other update maps in QFT that can describe (in a causally consistent way) \textit{multiple} measurements happening in local and finite regions of spacetime, and thus a description that accords with our experience of, say, making multiple measurements of finite duration at the Large Hadron Collider.

Hellwig and Kraus laid the foundations of such a causally consistent description in~\cite{Hellwig_Kraus}, with their focus on operators local to regions of spacetime, similarly local projectors, and local operations that do not affect the expectation values of spacelike separated observables~\footnote{Here we have assumed we are in the Heisenberg picture, where field operators carry the dynamics and depend on the spacetime coordinates. Thus, it makes sense to talk about spacelike separated observables.}. With their work, it seemed that the projection postulate could be used in QFT in a local manner consistent with relativity. This was not the end of the story for relativistic causality in QFT, however, as they missed a key scenario (discussed in Section~\ref{sec:sorkinscenario}) which highlights a further subtlety of causality and measurements in QFT.

We refer to this as the \textit{Sorkin scenario}, following its use by Sorkin in~\cite{Sorkin_impossible}. There he shows that if some agent, \textit{Charlie}, were able to perform a particular ideal measurement in scalar QFT, then another agent, \textit{Alice}, would be able to send a faster-than-light, or superluminal, signal to a third agent, \textit{Bob}. Such a violation of relativistic causality implies that Charlie's measurement \textit{cannot} be physically realisable by \textit{any} physical laboratory process. Sorkin's example consisted of measuring whether the system is in a particular wavepacket state or not. While the corresponding projection operator for this ideal measurement is not a local operator (in the sense of  Hellwig and Kraus~\cite{Hellwig_Kraus}), Sorkin's example at least illustrates the fact that not all ideal measurements of self-adjoint operators in QFT are physically realisable. An immediate question is then: which self-adjoint operators in QFT have ideal measurements which respect causality, and are therefore physically realisable, at least in principle?

Of particular interest are the local operators in the theory, as the superluminal signal in Sorkin's example may be due to the non-locality of the operator he used. One of the most basic local operators that comes to mind is the field operator for a real scalar field. Can we, at least in principle, make ideal measurements of the field without violating causality~\footnote{Note, we are agnostic about the physical process which realises this ideal measurement of the field. All we care about here is whether \textit{any} process can give rise to an update of the state corresponding to that given by an ideal measurement of the field operator.}? Technically speaking, the field operator is actually an operator-valued distribution, and thus one has to integrate, or \textit{smear}, it against some compactly supported function to get a well-defined local operator on the Fock space --- a \textit{smeared field operator}. The main question we want to answer is then:

\begin{center}
    \textit{\textbf{Do ideal measurements of smeared fields respect causality?}}
\end{center}
\vspace{1mm}

\noindent Since smeared field operators generate all other local operators in the theory, it is clearly important to answer this question.

In~\cite{Sorkin_impossible,Benincasa_2014} it was conjectured that such a measurement would respect causality, while in~\cite{Jubb_2022} the contrary was argued, though no definitive argument was given due to the technical nature of the problem. Here we provide an in-depth technical analysis of this question. We verify that such ideal measurements are acausal, and hence are not realisable by \textit{any} physical laboratory process. For example, by coupling a real scalar quantum field to \textit{any} other probe quantum/classical system in \textit{any} way (c.f. von Neumann measurement models and Unruh-deWitt detectors), we will never be able to realise the textbook projective measurement of the field operator. On the contrary, such a realisation, as we show below, would lead to a violation of relativistic causality, and hence cannot be possible in practice.

In deriving this result we make a few basic assumptions (see Section~\ref{sec:Assumptions}). One of which is the assumption that a Sorkin scenario can be engineered in the given spacetime setup. There are some interesting cases where this is not possible, and we explore these in Section~\ref{sec:Transitive loophole and discrete spacetimes}. Of note is the case of a discrete spacetime, such as a causal set~\cite{Surya_2019}. The discreteness of the spacetime opens up a somewhat natural loophole to avoid Sorkin's scenario, and thus a way to incorporate projective, or ideal, measurements of the field without violating causality. Thus, this work also unveils an interesting connection between the projection postulate and the continuum, or discrete, nature of spacetime.

\subsection{Further Background}\label{sec:Further background}

Aside from ideal measurements, one can ask whether other update maps are causal, and hence physically realisable~\cite{Jubb_2022}. In other words, what subset of the completely positive maps are causal in a generic QFT? Closely related to this question is previous work on the causality and locality of various operations on multipartite systems~\cite{Beckman_2001}, and there are analogies one can draw between that work and the QFT case in~\cite{Jubb_2022}. Namely, the properties of locality and causality can be separated in both cases, in the sense that one can find examples of update maps that are i) both local and causal, ii) local and not causal, iii) causal and not local, iv) and neither.

On the QFT side, there have been recent proposals for measurement models in the spirit of von Neumann~\cite{Bostelmann_2021,Polo_G_mez_2022,Fewster2023cfq}, where one couples a probe system to the main system of interest. One then measures the probe and interprets the result as a measurement of some observable of the main system. These models are particularly useful for studying the propagation of information from one system to another. Two examples for the probe system are a 2-state probe (such as in an Unruh-deWitt detector~\cite{Polo_G_mez_2022,EduardoOptics}), or a probe quantum field~\cite{Bostelmann_2021,Fewster2023cfq}. Both types of probes have their merits and drawbacks. In the former case, one must take account of the acausalities that emerge from coupling the non-relativistic detector system to the relativistic quantum field~\cite{EduardoOptics,EduardoBrokenCovariance}. In the latter one must use a distinct probe quantum field for each measurement, and discard it after use. Thus, given the finite number of fields we have in the standard model, this approach can only be applied within the realm of effective field theories, where we are more comfortable introducing new probe quantum fields and discarding them after each measurement. Clearly, however, this approach cannot be the full description at a more fundamental level. Both approaches also suffer from the usual problem of regression with von Neumann measurement models, namely, that one is left wondering how to describe a measurement on the probe, and thus must introduce a probe for the probe, and so on.

There may be other models of measurement in QFT that we are yet to formulate, but in the absence of such a complete description, it is still important to understand the space of physically allowed (with respect to causality) update in a generic QFT. This is the starting point for~\cite{Jubb_2022}, where a condition is derived for any map to respect causality in real scalar QFT. It should be noted that the derivation in~\cite{Jubb_2022} uses some assumptions that we will not require in the calculations below (see Section~\ref{sec:Causality of smeared field operations} for further discussion).

\subsection{Layout}\label{sec:Layout}

We first introduce some necessary concepts from continuum spacetime geometry in Section~\ref{sec:Continuum Spacetime geometry}, and, since we also consider discrete spacetimes such as causal sets, we introduce Causal Set Theory in Section~\ref{sec:Causal Set Theory}.

Before we get to QFT, we first cover some important concepts in classical real scalar field theory in Section~\ref{sec:Classical Real Scalar Field Theory}, as these will be relevant to later calculations. When reviewing QFT in Section~\ref{sec:Real Scalar Quantum Field Theory}, we follow an approach inspired by the more mathematically rigorous formulation of Algebraic (A)QFT~\cite{aqft_fewster_rejzner}. The basic local operators of interest to us are the smeared field operators, and thus we also devote some time to their properties and interpretation. 

The technical nature of our main result --- that ideal measurements of smeared fields are acausal --- necessitates some level of understanding of functional calculus, and so we introduce the relevant concepts in Section~\ref{sec:Functional calculus}. For the sake of brevity in this section, and with some following sections, we delegate the derivations of more technical results to the appendices in Section~\ref{sec:Appendices}.

In Section~\ref{sec:Ideal measurements} we review the textbook description of the projection postulate and ideal measurements, and introduce a key concept we call \textit{resolution}, facilitated by the machinery of functional calculus in Section~\ref{sec:Functional calculus}. We then essentially follow Hellwig and Kraus~\cite{Hellwig_Kraus} in Section~\ref{sec:Ideal measurements in QFT} when introducing the notion of ideal measurements of local operators to the relativistic setting of QFT. We further discuss more general update maps in Section~\ref{sec:Local and causal operations in QFT}, as well as defining what we mean for an update map to be local to a region of spacetime. To capture a wide variety of local update maps in a single framework, we introduce the general concept of a \textit{Kraus update map} in Section~\ref{sec:Kraus update maps}. This includes the case of ideal measurements as well as many other useful maps in QFT, e.g. unitary kicks. One can also consider `less' ideal, or approximations to ideal measurements in this framework. In Section~\ref{sec:sorkinscenario} we review Sorkin's scenario and some prior results on causality in~\cite{Jubb_2022}. We also note here why one of the assumptions from~\cite{Jubb_2022} should potentially be dropped. This point is further addressed in the remaining sections.

Section~\ref{sec:Causality of smeared field operations} then comprises our main result. We discuss our assumptions in Section~\ref{sec:Assumptions}, one of which is the existence of Sorkin's scenario in the given spacetime setup. We discuss some special cases where this is not possible in Section~\ref{sec:Transitive loophole and discrete spacetimes}, including the case of a discrete spacetime. In Section~\ref{sec:General causality condition for a smeared field operations} we derive a condition for a general Kraus update map to be causal, which also covers the case of ideal measurements. The utility of our functional calculus approach is most evident here, as our causality condition boils down to a condition on the functions used in defining the given Kraus update map. After this we give our main result --- that ideal measurements of smeared fields are acausal --- in Section~\ref{sec:The acausality of an ideal measurement of a smeared field}. 

Following this result, in Section~\ref{sec:Revisiting previous results with an illustrative example} to the heuristic arguments of previous literature~\cite{Sorkin_impossible,Benincasa_2014,Borsten_2021} where it was suggested that ideal measurements of smeared fields would be causal. Upon closer inspection, it appears those arguments are evaded (and thus our result is not inconsistent with the results in~\cite{Borsten_2021}) by a mathematical technicality, which we illustrate with a simple example in 2D non-relativistic quantum mechanics. In Section~\ref{sec:A decoherence functional/path integral perspective} we consider the decoherence functional, or path integral, perspective to provide further intuition as to why an ideal measurement in QFT enables a superluminal signal, and hence why it is not realisable. Finally, in Section~\ref{sec:Discussion} we discuss our results and future directions, as well as propose potential avenues to circumvent the acausality of the projection postulate in QFT.

\section{Spacetime Setup}

\subsection{Continuum Spacetime Geometry}\label{sec:Continuum Spacetime geometry}

We consider a Lorentzian spacetime $(\spctm{M},\metric)$, where $\spctm{M}$ is a time orientable manifold and $\metric$ is a metric on $\spctm{M}$ with mostly plus signature. We will assume that $(\spctm{M} , \metric )$ is \textit{globally hyperbolic}, namely that it contains a \textit{Cauchy surface}, $\Sigma$, which is a hypersurface for which no timelike inextendible curve intersects more than once (Figure~\ref{fig:spacetime geometry a}).

\begin{figure}
\begin{center}
\begin{subfigure}[b]{0.5\textwidth}
     \centering
     \includegraphics[width=\textwidth]{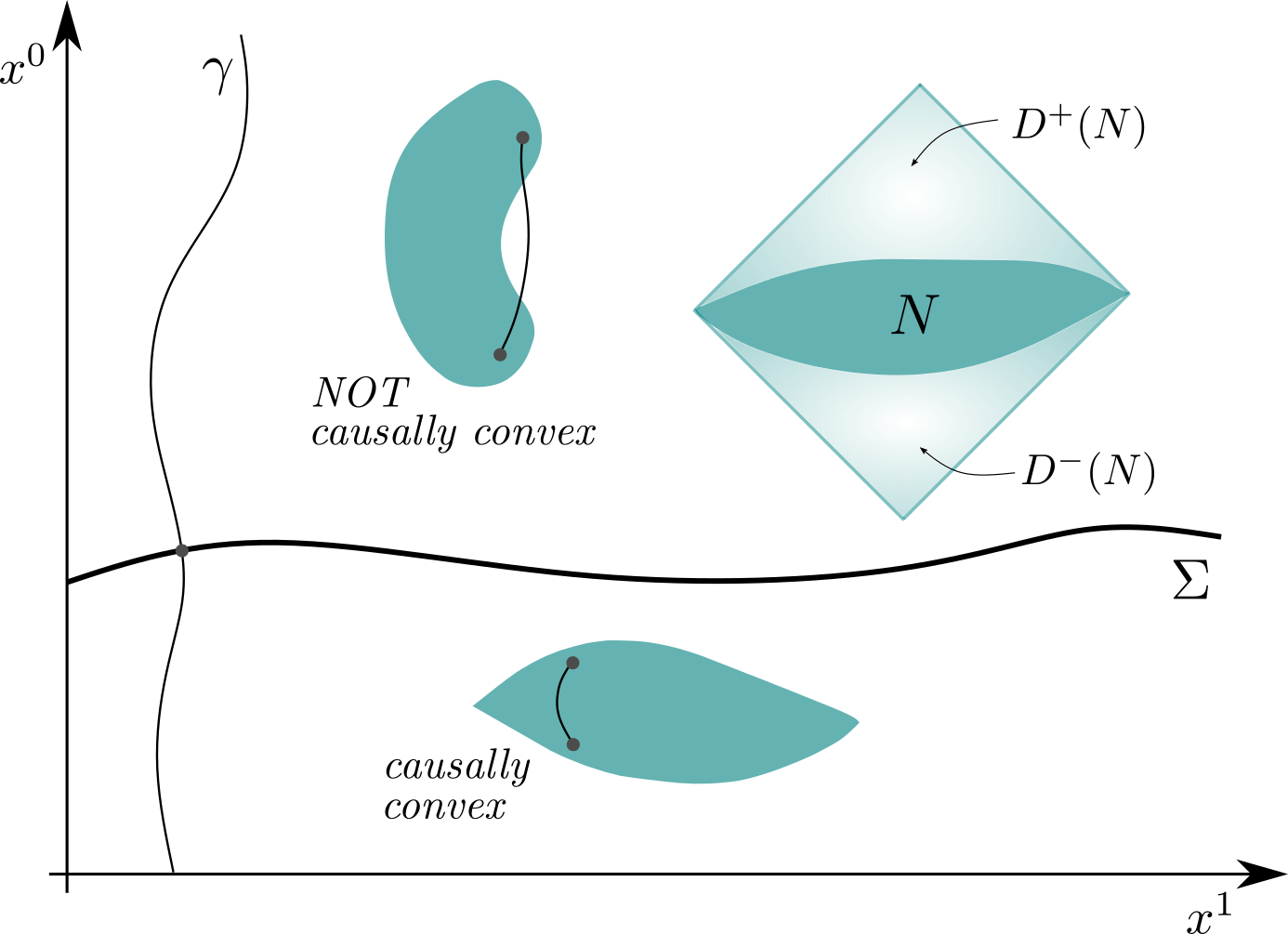}
     \caption{Spacetime diagram (time and space coordinates $x^0$ and $x^1$) illustrating a Cauchy surface ($\Sigma$), causal convexity, and the domain of dependence $D(N) = D^+(N)\cup D^-(N)$.}
     \label{fig:spacetime geometry a}
 \end{subfigure}
 \hspace{5mm}
 \begin{subfigure}[b]{0.45\textwidth}
     \centering
     \includegraphics[width=\textwidth]{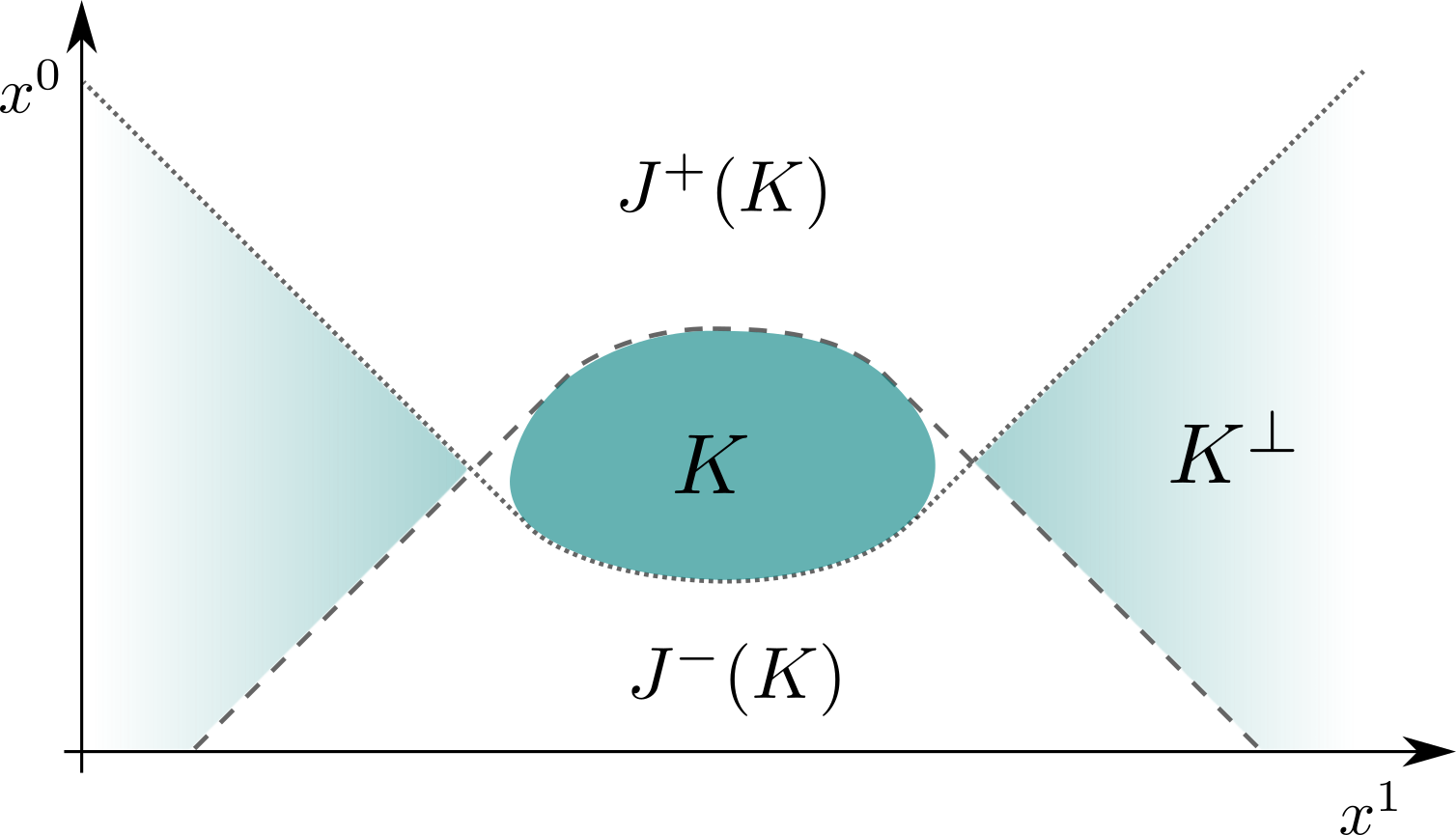}
     \caption{Illustration of future and past sets and causal complement (the shaded parts on either side of $K$). Above/below the dashed/dotted line corresponds to the out/in-region $K^{\pm}$.}
     \label{fig:spacetime geometry b}
 \end{subfigure}
\end{center}
\end{figure}

Given a subset $S \subseteq \spctm{M}$, its causal (timelike or lightlike) future/past is denoted by $J^{\pm}(S)$ (Fig.~\ref{fig:spacetime geometry b}), and the chronological (timelike) future/past by $I^{\pm}(S)$. We say $S$ is causally convex if any causal curve with endpoints in $S$ is entirely contained in $S$ (Fig.~\ref{fig:spacetime geometry a}). In the following, we reserve the word \textit{region} for any open causally convex subset $\spctm{R} \subseteq \spctm{M}$. Note, endowing any region $\spctm{R}\subseteq \spctm{M}$ with the metric $\metric$ yields a globally hyperbolic spacetime.

For any subset $S\subseteq\spctm{M}$, the \textit{domain of dependence} is $D(S)=D^+(S) \cup D^-(S)  $, where $D^{\pm}(S)$ denotes the future/past domain of dependence, i.e. the set of points $x \in \spctm{M}$ for which every past/future inextendible causal curve passing through $x$ intersects $S$ (Fig.~\ref{fig:spacetime geometry a}). Given equations of motion on the spacetime, one can think of $D(S)$ is the subset of spacetime where one can determine the solutions uniquely from data on $S$.

The \textit{causal complement} of a subset $S$ consists of all points spacelike to, or causally disconnected from $S$, which we denote as $S^\perp=\spctm{M} \setminus ( J^+(S) \cup J^-(S)) $ (Fig.~\ref{fig:spacetime geometry b}). We can associate to a compact (closed and bounded) subset, $K$, an \textit{in/out region} $K^{\mp}=\spctm{M} \setminus J^{\pm}(K) $. Such regions are open, causally convex, and hence constitute globally hyperbolic spacetimes in their own right when endowed with the metric $\metric$. 

We use the notation $C_0^{\infty}(\spctm{M})$ for the space of real-valued \textit{test functions} $f: \spctm{M}\rightarrow \mathbb{R}$, which are smooth and compactly supported, i.e. $\supp f\subset \spctm{M}$ is a compact subset. Note, if we consider complex valued test functions at any point below we will highlight this. The value of $f$ at a spacetime point $x\in\spctm{M}$ is $f(x)$. Finally, $C^\infty(\spctm{M})$ denotes smooth functions that are not necessarily of compact support. 

Since we will also consider discrete spacetimes such as those in Causal Set Theory, we briefly review some core concepts of the latter.

\subsection{Causal Set Theory}\label{sec:Causal Set Theory}

Causal set theory is an approach to quantum gravity that replaces the continuum spacetime with a discrete collection of events, or spacetime points, ordered by causality~\cite{Surya_2019}. This conjecture for the underlying structure of spacetime is motivated by the fact that all the components of the metric, except one, can be recovered from the causal structure of the given continuum spacetime \cite{Levichev} \cite{Malament}~\footnote{A more detailed discussion on relativity and causality is given by Penrose, Hawking and Ellis in~\cite{Penrose}, \cite{Hawking_Ellis}.}. The missing component encodes the volume information of the continuum spacetime, but for a discrete collection of spacetime points, such as those in a causal set, the volume information is given to us for free by simply counting the \textit{number} of discrete elements comprising the causal set, or any subset thereof.

Thus, the hope in Causal Set Theory is that one can recover all of the geometrical information (all the components of the metric) of the continuum spacetime (at sufficiently large scales) via the discreteness of the underlying spacetime points (giving us our notion of volume), and the causal ordering of those discrete points (giving us the causal structure). This is summarised by the motto: ``Order $+$ Number $=$ Geometry''.

Technically speaking, a causal set (or causet) is a locally finite partially ordered set (see Fig.~\ref{fig:causet}), i.e. a pair ($\spctm{C} ,  \preceq $) where $\spctm{C} $ is a set and $\preceq$ is a partial order relation on $\spctm{C}$ satisfying:
\begin{itemize}
  \item  Reflexivity: $\forall \ x \in \spctm{C},\ x \preceq a$
  \item Acyclicity: $\forall \ x,y \in \spctm{C}, \ x \preceq y \preceq x \Rightarrow x=y $
  \item Transitivity:  $\forall \ x,y,z \in \spctm{C}, \ x \preceq y \preceq z \Rightarrow x\preceq z $
  \item Local finiteness: $\forall \ x,z \in \spctm{C}, | \ [x,z] \ | < \infty $, where the set $[x, z] := \{y \in \spctm{C}|\ x \preceq y \preceq z \}$ is a causal interval and $|S|$ is the cardinality of a set $S$.
\end{itemize}

\begin{figure}
 \centering
 \includegraphics[width=0.5\textwidth]{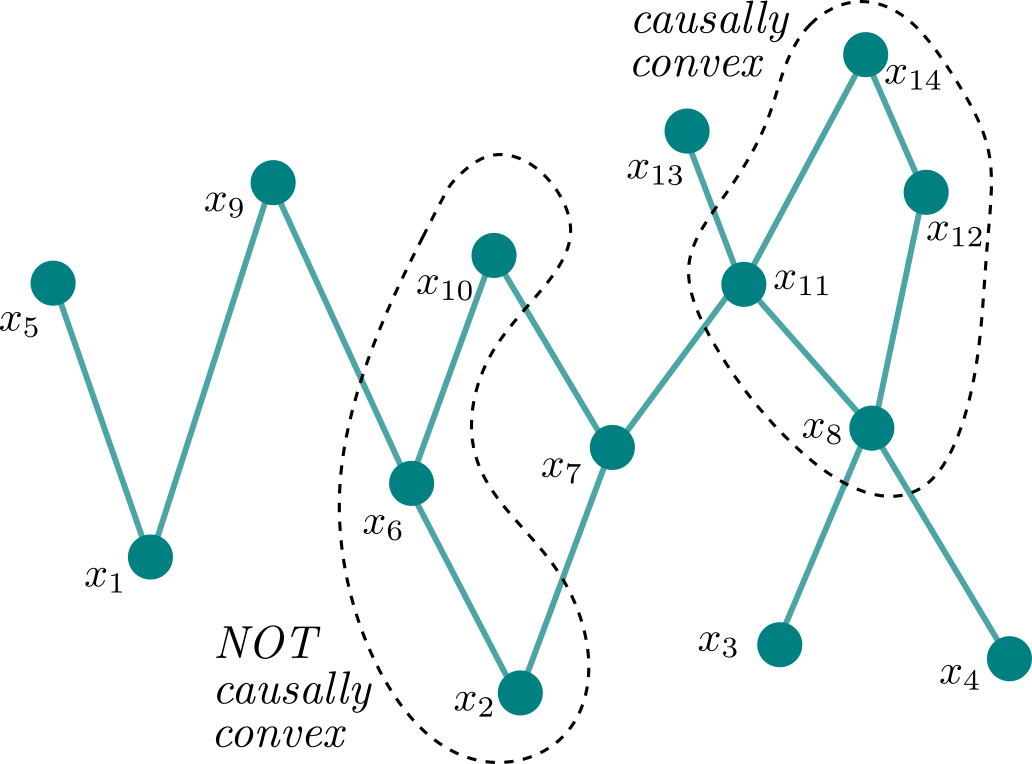}
 \caption{Illustration of causal convexity on a naturally labelled causet. Note, the direction of the order relations always goes upwards.}
 \label{fig:causet}
\end{figure}

The set $ \spctm{C}$ represents the spacetime points, while the partial order $\preceq $ represents the spacetime causal relations. For instance, by $x$ $\preceq $ $y$ we understand that $x$ is in the causal past of $y$.
 
Note that the first three conditions on $\preceq$ above are satisfied by the causal order of a continuum \textit{causal} Lorentzian spacetime~\cite{minguzzi2008causal}. The fourth condition is failed by any continuum spacetime (as any continuum causal interval contains an uncountable infinitude of points), and so encodes the fact that a causal set is discrete. This proposed discreteness is also motivated by a need to resolve certain singularities that plague our current continuum theories, e.g. the singularities inside black holes in general relativity.

For any finite subset of a causet $\spctm{X}\subseteq \spctm{C}$, it is always possible to label the points in $\spctm{X}$ with natural numbers such that, for any pair $x_j , x_k \in\spctm{X}$, if $x_j \preceq x_k$ then the labels satisfy $j<k$ as numbers. We call this a \emph{natural labelling} (see Fig.~\ref{fig:causet}).

Recall that a region $\spctm{R}$ of a continuum spacetime $\spctm{M}$ is an open causally convex subset. Similarly, for a causet $\spctm{C}$ we use the word \textit{region} to describe any subset $\spctm{R}\subseteq \spctm{C}$ that is causally convex with respect to the order relation $\preceq$, i.e. for any pair $x,z\in \spctm{R}$, if $x\preceq z$, then for any $y$ between $x$ and $z$ ($x\preceq y \preceq z$), $y \in\spctm{R}$ also.

To any subset $K\subseteq \spctm{C}$ we can similarly associate an out/in-region as $K^{+/-} = \spctm{C}\setminus J^{-/+}(K)$, which is a causally convex subset of $\spctm{C}$.

Since any causet $\spctm{C}$ is discrete, a test function on a causet is simply a real-valued function, $f:\spctm{C}\rightarrow \mathbb{R}$, which is non-zero only on a finite number of points. To align the notation with the continuum, we denote the space of all such test functions as $C^{\infty}_0(\spctm{C})$, even though there is no sense (that we have defined) in which these functions are infinitely differentiable. To further align notation, we denote the value of $f$ at a causet element, $x\in\spctm{C}$, as $f(x)$. $C^{\infty}(\spctm{C})$ denotes the space of smooth functions.

Below we use $\spctm{S}$ to denote either a continuum spacetime manifold $\spctm{M}$, or a discrete causet $\spctm{C}$, and we use $\spctm{M}$ or $\spctm{C}$ when the given discussion applies only to the continuum or causet case respectively.

Given some test function $f$ on a (continuum or discrete) spacetime $\mathcal{S}$ we write
\begin{equation}
    \int_{\mathcal{S}}dx \, f(x) \; ,
\end{equation}
to denote the integral of $f$ over $\mathcal{S}$. For a continuum spacetime, $\mathcal{S}=\spctm{M}$, `$dx$' denotes the usual volume measure with $\sqrt{-\det \metric}$, and for a causal set, $\mathcal{S}=\spctm{C}$, `$dx$' denotes the counting measure over elements $x\in\spctm{C}$. That is, the `integral' becomes the sum:
\begin{equation}
    \sum_{x\in\spctm{C}}f(x) \; .
\end{equation}

\section{Real Scalar Field Theory}

\subsection{Classical Real Scalar Field Theory}\label{sec:Classical Real Scalar Field Theory}

Consider a (free) classical real scalar field $\varphi$ in a globally hyperbolic (continuum) spacetime manifold $\spctm{M}$, which obeys the field equation
\begin{equation}
    P\varphi=0
\end{equation}
where $P = \Box + m^2$ is the Klein-Gordon operator with mass $m$. The Klein-Gordon operator $P$ has unique advanced ($-$) and retarded ($+$) Green operators $G^{\pm} : C_0^\infty(\spctm{M}) \rightarrow C^\infty(\spctm{M})$ obeying 
\begin{equation}\label{eq:Green_operators_properties}
    G^{\pm} Pf=f, \quad P G^{\pm} f=f, \quad \supp(G^{\pm} f) \subseteq J^{\pm}(\supp(f))
\end{equation}
for all test functions $f \in C_0^\infty(\spctm{M})$~\cite{Green_hyperbolic_B_r_2014}. The action of $G^{\pm}$ can be described via an integral kernel, also denoted by $G^{\pm}(x,y)$~\footnote{We have overloaded notation by using $G^{\pm}$ for the integral kernel as well. Any confusion between the operator and integral kernel can be resolved by noting that the operator takes a test function as its argument, while the integral kernel takes two spacetime points.}, as
\begin{equation}
    (G^{\pm} f) (x) = \int_{\spctm{M}} dy \, G^{\pm}(x,y) f(y) \; .
\end{equation}
We will sometimes refer to $G^{\pm}(x,y)$ as a \textit{Green function}, rather than an integral kernel.

We call a solution $\varphi$ \textit{spatially compact} if its Cauchy data ($\varphi$ and its normal derivative) on any Cauchy surface is compactly supported on that surface. The space of all smooth spatially compact solutions will be denoted as $\text{Sol}(\spctm{M})$.

We define the \textit{Pauli-Jordan} operator as $\Delta= G^- -G^+$. Notably, any $\varphi\in\text{Sol}(\spctm{M})$ can be generated by some test function $f\in C^{\infty}_0 (\spctm{M})$ as $\varphi = \Delta f$. In terms of the associated integral kernel $\Delta(x,y) = G^-(x,y) - G^+(x,y)$ (also called the \textit{Pauli-Jordan function}) we have
\begin{equation}
    \varphi(x) = \int_{\spctm{M}}dy \, \Delta (x,y) f(y) \; .
\end{equation}
Roughly speaking, we can think of the given test function $f$ as a source/sink for the `part' of $\varphi$ that is to the future/past of $\supp f$. More specifically, let $\varphi^{\pm} = G^{\pm}f$, then $\varphi^{\pm}(x) = \mp \varphi (x)$ at any point $x \in (\supp f)^{\pm}$, i.e. in the out/in-region associated to $\supp f$. Moreover, $\varphi^{\pm}$ is an inhomogeneous solution of $P\varphi^{\pm} = f$ with support only in the future/past of $\supp f$, and thus by standard terminology we would call $f$ a \textit{source}/\textit{sink} for $\varphi^{\pm}$. 

From~\eqref{eq:Green_operators_properties} we see that, for any $f\in C^{\infty}_0(\spctm{M})$, we have $\Delta P f = G^- P f - G^+ P f = 0$. That is, the image of the operator $P$ is the kernel of $\Delta$, denoted by $\text{ker}\Delta$. This implies that the test function that generates a given solution $\varphi\in \text{Sol}(\spctm{M})$ is not unique. Explicitly, for any pair of test functions $f,h\in C^{\infty}_0(\spctm{M})$, both $f$ and $f+Ph$ generate the same solution when acted on by $\Delta$. This lack of uniqueness is related to the fact that the same solution, post some given Cauchy surface, can be sourced in different ways prior to the Cauchy surface. More specifically, given two test functions, $f$ and $g$, which both generate $\varphi$ as $\varphi = \Delta f = \Delta g$, then $f$ and $g$ are both valid potential sources/sinks for $\varphi^{\pm}$ in the out/in-region for $\supp f \cup \supp g$.

Now consider a causet $\spctm{C}$. In this case, there are no equations of motion for the field, but one can define dynamics by instead starting from a retarded Green function $G^+(x,y)$, where $x$ and $y$ are causet points. If there are $N$ points in $\spctm{C}$ then one can think of $G^+(x,y)$ as the $x,y$ element of an $N\times N$ matrix, which we denote for brevity as $G^+$. One particular candidate for the matrix $G^+$ is a weighted sum over all chains/paths between pairs of elements $x$ and $y$ (see ~\cite{Johnston} for details).

Defining the advanced Green function (or matrix) as the transpose of the retarded Green function~\footnote{The equation $G^-(x,y) = G^+(y,x)$ is also valid in many continuum spacetimes.}, $G^-=(G^+)^T$, we can then define the causet Pauli-Jordan function (or matrix) as above, i.e. $\Delta= G^- - G^+$. Both the causet and continuum Pauli-Jordan functions respect the causal structure, in the sense that $\Delta(x,y) = 0$ if $x$ and $y$ are spacelike.

In analogy with the continuum, the space of solutions on the given causet $\spctm{C}$, denoted by $\text{Sol}(\spctm{C})$, is taken to be the image of the Pauli-Jordan matrix $\Delta$. Going further, one can show that the matrix $i\Delta$ is skew-symmetric and Hermitian. These two features guarantee that this matrix has even rank, and that its non-zero eigenvalues are real and come in positive and negative pairs~\cite{Perlis}.

Finally, for a (continuum or discrete) spacetime $\spctm{S}$ we define the \textit{smeared} Pauli-Jordan function $\Delta(f,g)$, for two test functions $f,g\in C^{\infty}_0(\spctm{S})$, as
\begin{equation}
    \Delta(f,g) = \int_{\spctm{S}} dx dy \, f(x) \Delta(x,y) g(y) \; .
\end{equation}

\subsection{Real Scalar Quantum Field Theory}\label{sec:Real Scalar Quantum Field Theory}

To construct free real scalar QFT on some (continuum or discrete) spacetime, $\spctm{S}$, we will first build the Hilbert space, or Fock space, on which the operators act. To do this, one simply needs to specify a bilinear form $W:C_0^{\infty}(\spctm{S})\times C_0^{\infty}(\spctm{S}) \rightarrow\mathbb{C}$ which satisfies $W(f,g)-W(g,f) = i\Delta(f,g)$, in addition to some other properties (see Sec.4.3. of~\cite{aqft_fewster_rejzner} for details). This bilinear form will shortly become the 2-point function of the QFT.

We then construct the single particle Hilbert space, $\alg{H}$, as the completion of the vector space $\text{Sol}(\spctm{S})$ using the inner product $\braket{\varphi | \psi} = \braket{\Delta f | \Delta g} = W(f,g)$, which is well defined for any pair of solutions $\varphi , \psi \in\text{Sol}(\spctm{S})$, since there always exists some test function $f$ such that $\varphi = \Delta f$ for any $\varphi\in\text{Sol}(\spctm{S})$. Note that this requires $W(f,g)$ to be independent of which test functions are used to generate the given solutions $\varphi$ and $\psi$.

It is instructive to consider the specific case of $d$-dimensional Minkowski spacetime, $\mathbb{M}^d$. In this case, there is a bijection between $\text{Sol}(\mathbb{M}^d)$ and smooth, compactly supported Cauchy data on the $t=0$ Cauchy surface. This data takes the form of a pair of real-valued, smooth, compactly supported functions on $\mathbb{R}^{d-1}$ (one for $\varphi$ on the surface and one for its normal derivative), and thus we can equivalently think of the space $C^{\infty}_0(\mathbb{R}^{d-1})\oplus C^{\infty}_0(\mathbb{R}^{d-1})$ instead of $\text{Sol}(\mathbb{M}^d)$. Completing the latter space using the inner product $\braket{\Delta f | \Delta g} = W(f,g)$ then yields the space of complex-valued square integrable functions on $\mathbb{R}^{d-1}$, denoted $L^2(\mathbb{R}^{d-1} ; \mathbb{C})$ (see~\cite{kay2007quantum} for details). That is, for the single-particle Hilbert space we have $\alg{H}\cong L^2(\mathbb{R}^{d-1} ; \mathbb{C})$.

Going back to the general case for some spacetime $\spctm{S}$, one can then construct the usual bosonic Fock space as the (completion of the) symmetrised direct sum:
\begin{equation}\label{eq:fock_space}
    \alg{F} = \bigoplus_{n=0}^{\infty} S(\alg{H}^{\otimes n}) \; ,
\end{equation}
where $S(\cdot)$ denotes symmetrisation, and where $S(\alg{H}^{\otimes 0}) = \mathbb{C}$ is the 0-particle sector spanned by the ground state $\gs$. The space of operators we will consider on $\alg{F}$, denoted here by $\alg{A}$, will consist of linear maps from dense subspaces of $\alg{F}$ to itself.

In textbook descriptions one quantises the classical field $\varphi(x)$ into the \textit{field operator} $\phi(x)$. As $x\in\spctm{S}$ is a spacetime point, we are implicitly working in the Heisenberg picture where the operators carry the dynamics. Technically speaking, in a continuum spacetime $\phi(x)$ is an \textit{operator-valued distribution}, and thus one needs to integrate, or \textit{smear}, $\phi(x)$ against some test function $f\in C_0^{\infty}(\spctm{S})$ to form something which is actually an operator on the Fock space $\alg{F}$~\footnote{For a causet $\spctm{C}$ the operator $\phi(x)$ at some point $x\in\spctm{C}$ is a valid operator on $\alg{F}$, however.}. We thus introduce the \textit{smeared field operator}:
\begin{equation}\label{eq:smeared_field_in_terms_of_phi}
    \phi(f) = \int_{\spctm{S}} dx f(x) \phi (x) \; .
\end{equation}
Note, we have overloaded the notation for $\phi$ as we did for $\Delta$ above. If $f$ is supported in some region $\spctm{R}\subseteq \spctm{S}$, we say that $\phi(f)$ is \textit{localisable} in $\spctm{R}$. Together with the identity, $\mathds{1}$, the smeared field operators, $\phi(f)$, for all $f\in C_0^{\infty}(\spctm{S})$, generate the algebra of local operators.   

For some $f\in C_0^{\infty}(\spctm{S})$ the action of the associated smeared field, $\phi(f)$, on $\alg{F}$ can be described explicitly by first writing
\begin{equation}\label{eq:smeared field in terms of ladder ops}
    \phi(f) = a( \Delta f ) + a^{\dagger}(\Delta f) \; ,
\end{equation}
where $a(\Delta f)$ and $a^{\dagger}(\Delta f)$ are bosonic ladder operators associated to the corresponding `mode', or solution, generated by $f$. These ladder operators satisfy $[a( \Delta f ),a^{\dagger}(\Delta f)]= W(f,g)\mathds{1}$, and $a(\Delta f)\gs = 0$. One can then build up $n$-excitations of a given solution, or `mode', $\Delta f$ by acting on the ground state as $(a^{\dagger}(\Delta f))^n\gs$.

With~\eqref{eq:smeared field in terms of ladder ops} one can then explicitly check that a smeared field operator, $\phi(f)$, is an unbounded operator acting on a dense domain of $\alg{F}$. As an example of its action, starting with the ground state, $\gs\in\alg{F}$, the action of $\phi(f)$ gives the associated single-particle state, i.e. $\phi(f)\gs = \ket{\Delta f}$. The states one can reach by acting on $\gs$ with finite combinations of smeared fields are dense in $\alg{F}$. 

The 2-point function of the ground state is $\gsb  \phi(f) \phi(g) \gs  = W(f,g)$, and any $n$-point function is determined via this 2-point function, i.e. $\gs$ is a Gaussian, or quasifree, state. In particular, we have
\begin{equation}\label{eq:moments_of_phi_f_in_gs}
    \gsb \phi(f)^n \gs = \begin{cases}
(n-1)!! \, W(f,f)^{n/2} &\text{$n$ even} \; ,\\
0 &\text{$n$ odd} \; .
\end{cases}
\end{equation}

Smeared field operators have some important
\begin{properties}[Smeared field operators]\label{prop:smeared_fields}
\phantom{a}
\begin{enumerate}
    \item \textit{Linearity}: $\phi(a f+ bg) = a\phi(f) + b\phi(g)$ for any $a,b\in\mathbb{C}$ and any test functions $f,g$.
    \item \textit{Hermiticity}: $\phi(f)^{\dagger} = \phi(f^*)$ for any complex-valued test function $f$ (here $(\cdot)^*$ denotes complex conjugation).
    \item \textit{Field equation}: $\phi(f)=0$ if $f\in \ker \Delta$.
    \item \textit{Covariant commutation relations}: $[\phi(f) , \phi(g)] = i \Delta(f,g) \mathds{1}$.
\end{enumerate}
\end{properties}
\noindent Temporarily considering a continuum spacetime $\spctm{S}=\spctm{M}$, point (iii) essentially ensures that the operator-valued distribution $\phi(x)$ satisfies the equations of motion. Recall in the continuum that if $f\in\ker\Delta$, then there exists a test function $g$ such that $f = Pg$. Thus, $\phi(Pg)=0$ for any test function $g$. It is important to note that this means that the localisation region of a given smeared field, and thus of a given operator more generally, is not unique. In particular, for any smeared field $\phi(f)$ localisable in some region $\spctm{R}\subset \spctm{M}$, given any region $\spctm{R}'$ such that $\spctm{R}\subseteq D(\spctm{R}')$, i.e. such that $\spctm{R}$ is in the domain of dependence of $\spctm{R}'$ (e.g. set $\spctm{R}'$ to be $N$ in Fig.~\ref{fig:spacetime geometry a} and $\spctm{R}$ to be any other region within $D(N)$), then there exists a test function $g$ supported in $\spctm{R}'$ such that $f-g=Ph$ for some $h\in C_0^{\infty}(\spctm{M})$, and hence $\phi(f) = \phi(g)$. Thus, $\phi(f)$ is also localisable in $\spctm{R}'$. This generalises to any operator $A\in\alg{A}(\spctm{R})$, i.e. $A$ is also localisable in $\spctm{R}'$. This lack of uniqueness is related to the notion in classical or quantum physics that one can measure the same observable of the theory at different times, e.g. one can measure the position and momentum of a classical projectile at time $t$, or, equivalently, one can measure its position and momentum at some other time $t'$ and use the equations of motion to infer its position and momentum at $t$.

\begin{remark}[Interpretation of a smeared field operator]\label{remark:Interpretation of a smeared field operator}
It is worth pausing to build some physical intuition for what $\phi(f)$ represents as an observable of the theory. Since a smeared field operator, $\phi(f)$, can be mapped (bijectively) to a smooth spatially compact solution, $\Delta f$, it makes sense to think of $\phi(f)$ as the quantum operator associated with the classical solution, or `mode', $\Delta f$, in the same way that we think of $\hat{x}$ in non-relativistic quantum mechanics as the quantum operator associated with the classical position of the particle. Furthermore, since the spectrum of a smeared field operator is the whole real line $\mathbb{R}$ (the spectrum of $\hat{x}$ is also $\mathbb{R}$), when measuring this operator we can get any real number as an outcome, with some probability distribution over such outcomes. It seems reasonable to interpret this real number as a measure of `how much' the associated classical mode $\Delta f$ has been excited in the quantum field, or more specifically, the `amplitude' of this mode. It also seems reasonable to interpret the support of the given test function appearing in $\phi(f)$ as, roughly speaking, `where and when' we are measuring the amplitude of this mode. The many-to-one relationship between test functions $f$ and operators $\phi(f)$ also aligns with our intuition that we can measure the amplitude of a given mode at different points in time and/or space.
\end{remark}

Returning to the general case of a continuum or discrete spacetime, $\spctm{S}$, given some region $\spctm{R}\subseteq \spctm{S}$, there is an associated subalgebra of operators, $\alg{A}(\spctm{R})\subseteq \alg{A}$, which is generated by the identity, $\mathds{1}$, and the smeared fields, $\phi(f)$, localisable in $\spctm{R}$. Point (iv) above, and the fact that $\Delta(f,g)=0$ for test functions $f$ and $g$ with spacelike supports, implies what is sometimes called the \textit{Einstein Causality} condition. Specifically, for spacelike regions $\spctm{R}$ and $\spctm{R}'$, the associated subalgebras commute: $[\alg{A}(\spctm{R}) , \alg{A}(\spctm{R}')] = 0$.

Given some compact subset of some continuum spacetime, $K\subset \spctm{M}$, we have the associated out/in-subalgebra $\alg{A}(K^{\pm})$ corresponding to the out/in-region for $K$. Now, since the out/in-region contains a Cauchy surface for $\spctm{M}$, we can localise any $X\in\alg{A}$ in $\alg{A}(K^{\pm})$, and thus the subalgebra $\alg{A}(K^{\pm})$ is really the entire algebra $\alg{A}$.

It will be convenient for our calculations below to introduce the subalgebra, $\alg{B}\subset \alg{A}$, of \textit{bounded} operators on $\alg{F}$, as well as the local bounded subalgebras $\alg{B}(\spctm{R})$ for regions $\spctm{R}$.

\section{Functional Calculus}\label{sec:Functional calculus}

In order to reach our goal of analysing the causal nature of ideal measurements of smeared fields, we need to introduce some basic notions from functional calculus. Readers who are familiar with basic functional calculus are advised to skip to Section~\ref{sec:Functional calculus for smeared fields}.

\subsection{The Sectral Theorem and Functional Calculus}\label{sec:The spectral theorem and functional calculus}

By the spectral theorem~\cite{reed1980functional}, for any (essentially) self-adjoint operator $A\in\alg{A}$, defined on a dense domain of $\alg{F}$, we know there exists a corresponding spectral, or projection-valued measure (p.v.m), $P^A : \BorR\rightarrow \text{proj}(\alg{F})$, mapping the Borel sets over $\mathbb{R}$, $\BorR$, to projectors on $\alg{F}$, $\text{proj}(\mathrm{F})\subset \alg{A}$. $P^A$ satisfies some notable
\begin{properties}[p.v.m]\label{prop:pvm}
\phantom{a}
\begin{enumerate}
    \item $P^A(\mathbb{R}) = \mathds{1}$
    \item $P^A(\varnothing) = 0$
    \item $P^A(\bor{B})P^A(\bor{B}') = P^A(\bor{B}')P^A(\bor{B}) = P^A(\bor{B}\cap \bor{B}')$ for Borel sets $\bor{B},\bor{B}'\in\BorR$
    \item If $\bor{B}=\cup_{n=1}^{\infty}\bor{B}_n$, for mutually disjoint $\bor{B}_n\in\BorR$, then $P^A(\bor{B}) = \sum_{n=1}^{\infty}  P^A(\bor{B}_n)$,
\end{enumerate}
\end{properties}
\noindent where the infinite sum in (iv) converges in the strong operator topology. With this p.v.m we then write
\begin{equation}\label{eq:spectral_thm_A}
    A = \int_{\mathbb{R}} \lambda \, dP^A(\lambda ) \;  .
\end{equation}

\vspace{3mm}
\begin{remark}[Pure point spectrum]\label{remark:diag_op}
If $A$ has a pure point spectrum, i.e. if we can write
\begin{equation}\label{eq:A_diagonalisable}
    A = \sum_{n\in I} \lambda_n P_n \; ,
\end{equation}
where $I$ is some countable indexing set, $\lambda_n\in\mathbb{R}$ are distinct eigenvalues, and $P_n$ are the associated projectors (not necessarily of rank 1, or even finite rank), then the p.v.m $P^A$ can be related to the set of projectors $\lbrace P_n \rbrace_{n\in I}$ in a simple way. Specifically, if $\bor{B}\in\BorR$ is some subset of $\mathbb{R}$ that only includes a single eigenvalue $\lambda_n$, then $P^A(\bor{B}) = P_n$. More generally, if $\bor{B}\in\BorR$ only contains some subset of eigenvalues, $\lbrace\lambda_n\rbrace_{n\in J}$, for some subset $J\subset I$, then $P^A(\bor{B}) = \sum_{n\in J} P_n$. In this case the integral in~\eqref{eq:spectral_thm_A} reduces to the sum in~\eqref{eq:A_diagonalisable}.
\end{remark}

\vspace{3mm}
We now briefly review some concepts from functional calculus that are important for our purposes. The core idea is that one can use the projection-valued measure, $P^A$, for a given self-adjoint operator, $A$, to define what we mean by functions of the given operator, e.g. the exponential function: $e^A$.

Given two states, $\ket{\psi}, \ket{ \chi}\in \alg{F}$, we use the projection-valued measure $P^A$ to define a complex measure, $\mu^A_{\psi , \chi}: \BorR\rightarrow \mathbb{C}$, as $\mu^A_{\psi , \chi}(\bor{B}) = \bra{\psi} P^A (\bor{B}) \ket{ \chi}$, for any Borel $\bor{B}\subseteq\mathbb{R}$. Provided $\ket{\psi} , \ket{\chi }$ are in the domain of $A$, one can then compute the associated ``transition element'', $\bra{\psi }A \ket{ \chi}$, via the integral on the rhs:
\begin{equation}
    \bra{\psi }A \ket{ \chi} = \int_{\mathbb{R}}\lambda \, d\mu^A_{\psi , \chi}(\lambda ) \; .
\end{equation}
Now, given some function on the real numbers, $\zeta :\mathbb{R}\rightarrow \mathbb{C}$, we can define the associated operator $\zeta(A)$~\footnote{Note we have overloaded the notation for the function $\zeta$, as it should technically be acting $\mathbb{R}$ and not some operator.}. Specifically, $\zeta(A)$ is defined as the operator which, for any $\ket{\psi}, \ket{ \chi} \in \alg{F}$ for which the function $\zeta$ is square-integrable against the measure $\mu^A_{\psi , \chi}$, the associated transition element, $\bra{\psi }\zeta(A) \ket{ \chi}$, is given by the integral on the rhs:
\begin{equation}
    \bra{\psi }\zeta(A) \ket{ \chi} = \int_{\mathbb{R}} \zeta (\lambda ) \, d\mu^A_{\psi , \chi}(\lambda ) \; .
\end{equation}
In using the integral form on the rhs to define transition elements for all viable pairs $\ket{\psi}, \ket{ \chi} \in \alg{F}$, we have essentially defined the operator $\zeta (A)$; we say we have \textit{defined $\zeta (A)$ by functional calculus}. We also note that, if $\zeta$ is a bounded function, then the operator $\zeta(A)$ is bounded with norm $\norm{\zeta(A)}\leq \norm{\zeta}_{\infty}$.

\begin{example}[Indicator functions and projectors]\label{example:Indicator functions and projectors}
An important type of function for our purposes is the indicator function, $1_{S} : \mathbb{R} \rightarrow \mathbb{R}$, for some subset $S\subseteq \mathbb{R}$. We set $1_S(x) =1$ for $x\in S$, and $1_S(x)=0$ otherwise. Given some indicator function $1_{\bor{B}}$, for some Borel set $\bor{B}\in\BorR$, the associated operator defined through functional calculus is in fact equivalent to the associated projector. That is,
\begin{equation}\label{eq:pvm_as_indicator_function}
    1_{\bor{B}}(A) = P^A(\bor{B}) \; .
\end{equation}
In what follows, we will sometimes swap between these two choices of notation for the spectral projectors, depending on the context.
\end{example}

\vspace{3mm}
Finally, in Appendix~\ref{app:unitary_action_on_pvm} we show the following useful result:
\begin{equation}\label{eq:unitary_on_pvm}
    U^{-1}P^A (\cdot ) U = P^{U^{-1}AU}(\cdot ) \; ,
\end{equation}
where $U$ is some unitary operator ($U^{\dagger}=U^{-1}$).

\subsection{Functional Calculus for Smeared Fields}\label{sec:Functional calculus for smeared fields}

Let us now turn to the specific case of smeared field operators. The main goal for this section is to develop the machinery to compute the following ground state expectation value:
\begin{equation}\label{eq:main goal expectation value for this section}
    \gsb \zeta(\phi(f) ) e^{it\phi(g)}\gs \; ,
\end{equation}
where $f$ and $g$ are test functions, $\zeta : \mathbb{R}\rightarrow\mathbb{C}$ is some function, and thus $\zeta (\phi(f))$ is some operator defined via functional calculus. Expectation values of this form will turn up in the later sections on measurements and other operations. 

For some real-valued test function $f$, the Hermiticity property of smeared fields implies self-adjointness: $\phi(f)^{\dagger} = \phi(f)$. By the Spectral Theorem we know there exists an associated p.v.m $P^{\phi(f)}$.

In Appendix~\ref{app:unitary_kick_on_pvm_phi_f} we show the following property of this p.v.m:
\begin{equation}\label{eq:unitary_kick_on_pvm_phi_f}
    e^{-i\phi(g)}P^{\phi(f)}(\bor{B} ) e^{i\phi(g)} = P^{\phi(f) - \Delta(f,g)}(\bor{B}) = P^{\phi(f)}(\bor{B}+ t\Delta(f,g)) \; ,
\end{equation}
where $g$ is a real-valued test function, $\bor{B}$ is any Borel subset of $\mathbb{R}$, and by $\bor{B}+c$ for $c\in\mathbb{R}$ we mean the Borel set $\bor{B}$ translated by $c$. Here $P^{\phi(f) - \Delta(f,g)}$ denotes the spectral measure of the operator $\phi(f) - \Delta(f,g)\mathds{1}$. In Appendix~\ref{app:unitary_kick_on_pvm_phi_f} we further show that, for some function $\zeta:\mathbb{R}\rightarrow\mathbb{C}$,
\begin{equation}\label{eq:unitary_kick_on_zeta_phi_f}
    e^{-i\phi(g)}\zeta(\phi(f)) e^{i\phi(g)} = \zeta(\phi(f) - \Delta(f,g)) \; ,
\end{equation}
on an appropriate dense domain of $\alg{F}$.

Given $P^{\phi(f)}$, and the ground state $\gs\in\alg{F}$, we can further define the associated (real and non-negative) measure $\mu^{\phi(f)}_{\Omega,\Omega}(\cdot ) = \gsb P^{\phi(f)}(\cdot ) \gs$. For any function $\zeta :\mathbb{R}\rightarrow\mathbb{C}$, which is square-integrable against the measure $\mu^{\phi(f)}_{\Omega,\Omega}$, we can then compute the ground state expectation value of the operator $\zeta(\phi(f))$ via the integral on the rhs:
\begin{equation}
    \gsb \zeta(\phi(f) ) \gs = \int_{\mathbb{R}} \zeta(\lambda ) d\mu^{\phi(f)}_{\Omega,\Omega}(\lambda )\; .
\end{equation}

While the existence of the measure $\mu^{\phi(f)}_{\Omega,\Omega}$ is guaranteed by the Spectral Theorem, to compute such expectation values in practice we need to know the explicit form of this measure. We first note that, for $\zeta(\lambda) = \lambda^n$, we have
\begin{equation}\label{eq:moments}
    \gsb \phi(f)^n \gs = \int_{\mathbb{R}} \lambda^n d\mu^{\phi(f)}_{\Omega,\Omega}(\lambda )\; ,
\end{equation}
where we recall that the lhs was given explicitly in~\eqref{eq:moments_of_phi_f_in_gs}. One can further show that the explicit expressions for the moments in~\eqref{eq:moments_of_phi_f_in_gs} imply (via the Hamburger moment theorem~\cite{aqft_fewster_rejzner}) that the measure $\mu^{\phi(f)}_{\Omega,\Omega}$ is uniquely given as a Gaussian distribution with width $\sqrt{W(f,f)}$. That is, under the integral we can write $d\mu^{\phi(f)}_{\Omega,\Omega}(\lambda) = d\lambda \, p(\lambda)$, where $d\lambda$ denotes the usual Lebesgue measure on $\mathbb{R}$, and
\begin{equation}\label{eq:p Gaussian function}
    p(\lambda ) = \frac{1}{\sqrt{2\pi W(f,f)}}e^{-\frac{\lambda^2}{2W(f,f)}} \; ,
\end{equation}
is the \textit{measure density}. Note that this explicit form for $\mu^{\phi(f)}_{\Omega,\Omega}$ implies that $\mu^{\phi(f)}_{\Omega,\Omega}$ is \textit{absolutely continuous} with respect to the Lebesgue measure on $\mathbb{R}$. That is, for any Borel set $\bor{B}\subseteq \mathbb{R}$ of Lebesgue measure $0$, we have that $\mu^{\phi(f)}_{\Omega,\Omega}(\bor{B})=0$.

With the explicit form of $p$, ground state expectation values of functions of $\phi(f)$ can be computed by integrating the given function against $p$, i.e.
\begin{equation}
    \gsb \zeta(\phi(f) ) \gs = \int_{\mathbb{R}} \zeta(\lambda ) p(\lambda ) d\lambda \; .
\end{equation}
In particular, for $\zeta(\lambda) = e^{it\lambda}$ (for some $t\in\mathbb{R}$), we get
\begin{equation}\label{eq:exp_phi_measure}
    \gsb e^{it\phi(f)} \gs = \int_{\mathbb{R}} e^{it\lambda} p(\lambda ) d\lambda \; .
\end{equation}
Defining the Fourier transform of a function $\zeta$ as
\begin{equation}\label{eq:def FT}
    \fwt{F}\lbrace \zeta \rbrace (t ) = \frac{1}{\sqrt{2\pi}}\int_{\mathbb{R}}dx \, e^{itx} \zeta(x) \; ,
\end{equation}
we can rewrite the rhs of~\eqref{eq:exp_phi_measure} as $\sqrt{2\pi} \fwt{F}\lbrace p \rbrace (t)$. The Fourier transform of the function $p$ in~\eqref{eq:p Gaussian function} can be computed explicitly, and hence we get
\begin{equation}\label{eq:exp_phi_f_gs}
    \gsb e^{it\phi(f)} \gs = e^{-t^2\frac{W(f,f)}{2}} \; .
\end{equation}
This relation will be useful in what follows.

Recall our main goal for this section to compute~\eqref{eq:main goal expectation value for this section}. Working towards this we turn to the density of a slightly different measure, specifically, the complex measure
\begin{equation}
    \mu^{\phi(f)}_{\Omega , e^{it\phi(g)}\Omega}(\bor{B}) = \gsb P^{\phi(f)}(\bor{B}) \ket{e^{it \phi(g)}\Omega } = \gsb P^{\phi(f)}(\bor{B}) e^{it \phi(g)}\gs \; ,
\end{equation}
where $g$ is a real-valued test function, and where $\bor{B}\in\BorR$. In Appendix~\ref{app:absolute_continuity} we show that there exists a Lebesgue integrable density, $q$, such that, under the integral, we can write $d\mu^{\phi(f)}_{\Omega , e^{it\phi(g)}\Omega}(\lambda) = q(\lambda ) d\lambda$. For $s,t\in\mathbb{R}$, we can then write
\begin{align}\label{eq:ft_of_q}
    \gsb e^{is\phi(f)}e^{it\phi(g)}\gs & = \int_{\mathbb{R}} e^{is\lambda} q(\lambda ) d\lambda \nonumber
    \\
    & = \sqrt{2\pi} \fwt{F}\lbrace q \rbrace (s) \; ,
\end{align}
where we know that the Fourier transform of $q$ exists since $q$ is Lebesgue integrable.

On the other hand, using the Baker–Campbell–Hausdorff formula, e.g. $e^Xe^Y  = e^{X+Y + \frac{1}{2}[X,Y]+...}$, we find that
\begin{align}
    e^{is\phi(f)}e^{i t\phi(g)} & = e^{is\phi(f)+it\phi(g)-\frac{1}{2}st[\phi(f),\phi(g)]} \nonumber
    \\
    & = e^{is\phi(f)+it\phi(g)-\frac{i}{2}st\Delta(f,g)} \nonumber
    \\
    & = e^{-\frac{i}{2}st\Delta(f,g)}e^{i\phi(s\, f + t\, g )} \; ,
\end{align}
where we have used the linearity of $\phi(\cdot )$ in the last line. We then have that
\begin{align}\label{eq:explicit_ft_of_q}
    \gsb e^{is\phi(f)}e^{it\phi(g)}\gs & = e^{-\frac{i}{2}st\Delta(f,g)} \gsb e^{i\phi(s\, f + t\, g )}\gs \nonumber
    \\
    & = e^{-\frac{i}{2}st\Delta(f,g)}  e^{-\frac{W(s\, f + t\, g,s\, f + t\, g)}{2}} \; ,
\end{align}
using~\eqref{eq:exp_phi_f_gs} in the last line. Note that, from the bilinearity of $W(\cdot , \cdot)$, we can also write
\begin{equation}
    W(s\, f + t\, g,s\, f + t\, g) = s^2W(f,f) + t^2 W(g,g) + st (W(f,g) + W(g,f) ) \; .
\end{equation}
Importantly,~\eqref{eq:ft_of_q} and~\eqref{eq:explicit_ft_of_q} together imply that
\begin{equation}
    \sqrt{2\pi} \fwt{F}\lbrace q \rbrace (s) = e^{-\frac{i}{2}st\Delta(f,g)}  e^{-\frac{W(s\, f + t\, g,s\, f + t\, g)}{2}} \; .
\end{equation}
By dividing by $\sqrt{2\pi}$, and taking an inverse Fourier transform, we can then explicitly determine the complex density $q$ as a Gaussian of width $\sqrt{W(f,f)}$ shifted in the imaginary direction:
\begin{equation}\label{eq:complex_density}
    q(\lambda) = \frac{1}{\sqrt{2\pi W(f,f)}}e^{-t^2\frac{W(g,g)}{2}} e^{-\frac{(\lambda- i t W(f,g) )^2}{2 W(f,f)}} \; .
\end{equation}

Now, given any function $\zeta:\mathbb{R}\rightarrow \mathbb{C}$ that is square-integrable with the density $q$, we can explicitly compute the associated inner product by integrating $\zeta$ against this density $q$. That is,
\begin{equation}
    \gsb \zeta(\phi(f) ) e^{it\phi(g)}\gs = \int_{\mathbb{R}} \zeta(\lambda ) q(\lambda ) d\lambda \; .
\end{equation}
Thus we have met our main goal for this section.

We can go further and rewrite this using the Weierstrass transform of a function $\zeta:\mathbb{R}\rightarrow\mathbb{C}$, which is defined as
\begin{equation}\label{eq:weierstrass_transform}
    \fwt{W}\lbrace \zeta \rbrace (z) = \int_{\mathbb{R}}dx \, \frac{1}{\sqrt{4\pi}}e^{-\frac{(x-z)^2}{4}}\zeta(x) \; ,
\end{equation}
where $z\in\mathbb{C}$. We then have
\begin{equation}\label{eq:weierstrass_inner_product}
    \gsb \zeta(\phi(f) ) e^{it\phi(g)}\gs = e^{-t^2\frac{W(g,g)}{2}} \fwt{W}\Bigg\lbrace \zeta\Bigg(\sqrt{\frac{W(f,f)}{2}} \, \cdot \Bigg) \Bigg\rbrace ( z) \; .
\end{equation}
where $z = i t\sqrt{\frac{2}{W(f,f)}} W(f,g)$. That is, to compute the inner product of the operator $\zeta(\phi(f))$, for any sufficiently well-behaved function $\zeta$, we simply need to compute the value of its Weierstrass transform at $z$.

A final property of the Weierstrass transform worth noting is that
\begin{equation}\label{eq:weierstrass_shift}
    \fwt{W}\lbrace \zeta(a(\cdot ) + b) \rbrace (z) = \fwt{W}\lbrace \zeta (a(\cdot ) ) \rbrace \left(z + \frac{b}{a}\right) \; ,
\end{equation}
where $a,b\in\mathbb{R}$ and $a\neq 0$.

\section{Measurements and Other Operations}\label{sec:Measurements and Other Operations}

\subsection{Ideal Measurements}\label{sec:Ideal measurements}

We now recall the textbook description of the \textit{projection postulate} and \textit{ideal measurements}, and highlight some important aspects for the following sections.

Given some self-adjoint operator $A$ corresponding to some observable of the system, and some state $\ket{\Psi }$, we interpret $\bra{ \Psi } A \ket{ \Psi}$ as the expected value one would obtain from performing measurements of the given observable in the given state. Furthermore, given the associated p.v.m, $P^A$, the inner product $\bra{ \psi}P^A(\mathtt{B}) \ket{\psi}$ is interpreted as the probability that a measurement of the observable will have an outcome in the (Borel) subset $\bor{B}\subseteq \mathbb{R}$ (in the given state). For a density matrix $\rho$ (pure if $\rho = \ket{\psi}\bra{\psi}$ for some vector $\ket{\psi}$), expectation values are computed as $\text{tr}(\rho \, A)$.

Measurement outcomes are always real numbers corresponding to some point in the spectrum of the given self-adjoint operator $A$. We can divide up $\mathbb{R}$ into mutually disjoint disjoint Borel sets, $\bor{B}_n$, where $\bor{B}_n \cap \bor{B}_m = \varnothing$ for $n\neq m$. These Borel sets will act as `bins' to put our different measurement outcomes into. For example, we could have $\bor{B}_n = [n,n+1)$ for $n\in\mathbb{Z}$, or we could have just two Borel sets, e.g. $\bor{B}_1 = [-1,1]$ and $\bor{B}_2 = \mathbb{R}\setminus B_1$. The particular choice of bins (e.g. their width, how many there are etc.) is dependent on the specifics of the experiment and how well we can resolve particular outcomes. We make the following
\begin{definition}[Resolution]\label{def:resolution}
The \textit{resolution}, $\bor{R}$, of our measurement is a countable set (not necessarily finite) set of Borel sets $\bor{R} = \lbrace \bor{B}_n \rbrace_{n\in I}$ (where $I$ is some countable indexing set), where the Borel sets are mutually disjoint ($\bor{B}_n\cap \bor{B}_m = \varnothing$ for $n\neq m$) and cover the real line ($\cup_{n\in I}\bor{B}_n = \mathbb{R}$).
\end{definition}

Now, given some resolution, $\bor{R} = \lbrace \bor{B}_n \rbrace_{n\in I}$, we can form the corresponding projectors, $E_n$, using the p.v.m. Explicitly, $E_n = P^A(\bor{B}_n)$. Since the Borel sets are mutually disjoint, we have
\begin{align}
    E_n E_m & = P^A(\bor{B}_n )P^A(\bor{B}_m ) \nonumber
    \\
    & = P^A(\bor{B}_n\cap \bor{B}_m) \nonumber
    \\
    & = P^A (\varnothing) \nonumber
    \\
    & = 0 \; ,
\end{align}
for $n\neq m$. Line 2 follows from property (iii) in~\ref{prop:pvm}, and the last line follows from property (ii) in~\ref{prop:pvm}. In addition, the fact that the Borel sets cover $\mathbb{R}$ means that the following sum (which converges in the strong operator topology) gives
\begin{align}
    \sum_{n\in I} E_n & = \sum_{n\in I} P^A(\bor{B}_n) \nonumber
    \\
    & = P^A ( \cup_{n \in I} \bor{B}_n ) \nonumber
    \\
    & = P^A (\mathbb{R}) \nonumber
    \\
    & = \mathds{1} \; ,
\end{align}
where line 2 follows from property (iv) in~\ref{prop:pvm}, and the last line follows from property (i) in~\ref{prop:pvm}.

The \textbf{\textit{projection postulate}} states that after a measurement of $A$, with outcome in $\bor{B}_n$, the state goes from $\ket{\psi}\mapsto \frac{E_n \ket{\psi}}{\sqrt{p_n}}$, where the square root of the probability of this outcome, $p_n = \bra{\psi}E_n\ket{\psi}$, is required to normalise the resulting state. In terms of a density matrix $\rho$ we have $\rho \mapsto \rho_n = \frac{E_n \rho E_n}{p_n}$, where in this case the probability is $p_n = \text{tr}(\rho E_n)$.

In updating the state (or density matrix) in this way we say we have made an \textbf{\textit{ideal measurement}} of $A$. Note, we do not care about how this measurement has been implemented, only that its effect on the state of the system amounts to the given update. This is an operationalist perspective on measurement. Such a measurement is called `ideal' as, in reality, it is often the case that the state update does not take precisely this form. One can, however, hope to get closer and closer to this ideal measurement update with more refined measurement apparatus.

In a \textit{non-selective} ideal measurement we do not condition on any outcome. In this case, the act of making the measurement leaves us in a statistical distribution over the possible output states $\rho_n$, where each is weighted by its corresponding probability $p_n$. That is, 
\begin{equation}\label{eq:non_selective_ideal_rho}
    \rho\mapsto \tilde{\rho} = \sum_{n\in I} p_n \, \rho_n = \sum_{n\in I}E_n \rho E_n \; ,
\end{equation}
where on the far rhs the probabilities $p_n$ have cancelled out.

\begin{remark}\label{remark:the purpose of the projection postulate}
The practical purpose of the projection postulate, or equivalently the state update rule, is to tell us how the statistics of our system are affected by the given measurement of $A$. More specifically, for any other self-adjoint operator $A'$ that we wish to measure after our measurement of $A$, we compute its expectation value using the updated density matrix as $\text{tr}(\tilde{\rho}A')$, and not with the original density matrix $\rho$. Note that this means that future measurement statistics can be altered by the act of our measurement alone, even in the case where we do not condition on any particular outcome. We can, of course, condition on a particular outcome $\bor{B}_n$ and compute the associated conditional expectation value $\text{tr}(\rho_n A')$.
\end{remark}

\subsubsection{Interpreting Resolution}\label{sec:Interpreting resolution}

Given some self-adjoint operator of the form $A = \sum_n \lambda_n P_n$, with distinct eigenvalues $\lambda_n\in\mathbb{R}$ and associated projectors $P_n$ (not necessarily of rank 1, or even finite rank), there is a `canonical' choice of bins for the outcomes; a canonical choice of resolution $\bor{R} = \lbrace \bor{B}_n \rbrace_{n\in I}$. Specifically, we can pick bins, $\bor{B}_n$, that cover one, and only one, eigenvalue $\lambda_n$. Now, the projectors for these bins are precisely those in the spectral decomposition of $A$, i.e. $E_n = P^A(\bor{B}_n) = P_n$.

Consider, for example, a system with 3 energy eigenstates $\ket{1}$, $\ket{2}$, and $\ket{3}$. The Hamiltonian is then
\begin{equation}\label{eq:3_level_H}
    \hat{H} = \sum_{n=1}^3 h_n \ket{n}\bra{n} \; ,
\end{equation}
where we assume the energies of each level are distinct and satisfy $h_1< h_2 < h_3$. In this case, the textbook description of an ideal measurement of energy tells us to update the density matrix to $\tilde{\rho} = \sum_{n=1}^3 \ket{n}\bra{n} \rho \ket{n}\bra{n}$. This corresponds to a choice of resolution with 3 Borel sets, $\bor{R} = \lbrace \bor{B}_1 , \bor{B}_2 , \bor{B}_3 \rbrace$ with $h_n\in \bor{B}_n$. For example, we could have $\bor{B}_1 = (-\infty , h_1]$, $\bor{B}_2 = (h_1 , h_2]$, and $\bor{B}_3 = (h_2,\infty)$. The corresponding projectors are then $E_n = P^A(\bor{B}_n) = \ket{n}\bra{n}$. This is the `canonical' choice referred to above.

On the other hand, one could imagine a measurement of energy where the measurement accuracy is not good enough to resolve the two highest energies, $h_2$ and $h_3$ say. This can be described by the resolution $\bor{R}' = \lbrace \bor{B}_1 , \bor{B}'_2 \rbrace$, where $\bor{B}'_2 = \bor{B}_2\cup \bor{B}_3$. In this case, there are two projectors coming into the ideal measurement update formula: $E'_1 = E_1 = P^A(\bor{B}_1) = \ket{1}\bra{1}$ and $E'_2 = P^A(\bor{B}'_2) = \ket{2}\bra{2}+\ket{3}\bra{3}$. Now, one may take the point of view that this latter case really corresponds to an ideal measurement of a \textit{different} observable, e.g.
\begin{equation}
    \hat{H}' = h'_1 \ket{1}\bra{1} + h'_2 \left( \ket{2}\bra{2}+\ket{3}\bra{3} \right) \; ,
\end{equation}
with $h'_1 \neq h'_2$, rather than an ideal measurement of $\hat{H}$ with a different resolution $\bor{R}'$.

This ambiguity in interpretation --- whether a different resolution corresponds to i) a measurement of a different operator, or ii) a measurement of the same operator but with a different measurement accuracy --- is particular to operators which have a pure point spectrum, as those are the operators for which a `canonical' choice for the resolution exists. For operators with a continuous spectrum, e.g. $\hat{x}$, there is no `canonical' choice for the resolution, and hence the point of view in i) is removed as an option.

To see this, consider the position operator $\hat{x}$ in 1-dimensional non-relativistic quantum mechanics, where the Hilbert space is $L^2(\mathbb{R};\mathbb{C})$ (complex-valued square-integrable functions on $\mathbb{R}$). $\hat{x}$ is an (unbounded) self-adjoint operator on a dense domain of $L^2(\mathbb{R};\mathbb{C})$, and thus there exists an associated p.v.m, $P^{\hat{x}}$. Given some interval on the real line, or more generally some Borel set $\bor{B}\subseteq\mathbb{R}$, we can talk about an ideal measurement of position with outcome in $\bor{B}$. We can similarly form the associated projector $P^{\hat{x}}(\bor{B})$, which, for any wavefunction $\psi(x)\in L^2(\mathbb{R};\mathbb{C})$, gives $\tilde{\psi} = P^{\hat{x}}(\bor{B})\psi$ where $\tilde{\psi}(x)=\psi(x)$ for any $x\in \bor{B}$ and $\tilde{\psi}(x)=0$ otherwise. 

Given some resolution of the real line, $\bor{R} = \lbrace \bor{B}_n \rbrace_{n\in I}$ (some bins for which we can say the particle's position was measured to fall within), we can form the associated projectors $E_n = P^{\hat{x}}(\bor{B}_n)$ and determine the updated state $\tilde{\rho}$ via~\eqref{eq:non_selective_ideal_rho}. That is, even for unbounded operators such as $\hat{x}$, the ideal measurement update map in~\eqref{eq:non_selective_ideal_rho} is still well defined. One simply has to be careful how to define the projectors, $E_n$, as no normalisable eigenstate $\ket{\psi}$ exists for $\hat{x}$, and hence no projectors of the form $E_n = \ket{\psi}\bra{\psi}$ exist either.

While we can still make sense of~\eqref{eq:non_selective_ideal_rho} for unbounded operators such as $\hat{x}$, we cannot diagonalise $\hat{x}$ and write it as some sum as in~\eqref{eq:3_level_H}, and thus there is no canonical choice of projectors to use in the ideal measurement update formula (unlike in the example of $\hat{H}$ from~\eqref{eq:3_level_H}). Equivalently, there is no `canonical' choice of bins for our outcomes; the latter being entirely dependent on the details and quality of our experiment. Furthermore, given two different choices of bins, i.e. two different choices of resolution, we do not say we are making ideal measurements of two different observables. In both cases we say we are measuring the particle's $x$-position, but potentially with different levels of experimental accuracy.

\begin{remark}\label{remark:different resolutions not different operators}
\textit{Thus, to treat pure point and continuous (and mixed) spectrum self-adjoint operators democratically, we find it more appropriate to think of the ideal measurement update for different resolutions, $\bor{R}$ and $\bor{R}'$ (for the same self-adjoint operator $A$), not as measurements of different observables, but as measurements of the same observable but with different experimental setups.}
\end{remark}

\subsubsection{The Dual Picture}\label{sec:The dual picture}

Before moving on to the QFT case, let us briefly introduce the dual picture where the update formula is applied, not to the density matrix, but to the operators. Given some self-adjoint operator $A$ and resolution $\bor{R} = \lbrace \bor{B}_n \rbrace_{n\in I}$, we define the non-selective ideal measurement update map on operators as
\begin{equation}
    \mathcal{E}_{A,\bor{R}}(X) = \sum_{n\in I} E_n X E_n \; ,
\end{equation}
where $X$ is any operator, and $E_n = P^A(\bor{B}_n)$. After an ideal measurement of $A$, the expectation value of some operator $X$ is computed as
\begin{align}
    \text{tr}(\rho \mathcal{E}_{A,\bor{R}}(X) ) & = \sum_{n\in I} \text{tr}(\rho  E_n X E_n ) \nonumber
    \\
    & = \sum_{n\in I} \text{tr}(E_n \rho  E_n X  ) \nonumber
    \\
    & = \text{tr}(\tilde{\rho} X ) \; ,
\end{align}
where we have used the cyclic property of the trace in line 2. Thus, future expectation values, and more generally the statistics of any future measurements, are the same in this dual picture.

\subsection{Ideal Measurements in QFT}\label{sec:Ideal measurements in QFT}

In a relativistic theory such as QFT we must specify where in spacetime (`where' and `when') a measurement takes place. We will do this by specifying some compact subset of spacetime, $K\subset \spctm{S}$. The only observables, $A$, that we can measure in this subset of spacetime are those for which $A\in\alg{A}(\spctm{R})$ for some region $\spctm{R}\subset K$ (i.e. $A$ is localisable in some region inside $K$).

The reason for using a compact subset $K$ is the following. Firstly, the boundedness of $K$ (in both space and time) encodes the fact that our measurement has finite spatial extent and duration in time. Secondly, recalling Section~\ref{sec:Continuum Spacetime geometry} we can talk about the in/out-regions, $K^{\mp} = \spctm{S}\setminus J^{\pm}(K)$, associated with $K$. The in-region corresponds to the subset of spacetime which is \textit{definitely not} to the future of our measurement of $A$ in $K$, and is thus the region in which other measurements should not be affected by our measurement of $A$ in $K$. Conversely, the out-region is the subset of spacetime which is not to the past of our measurement of $A$ in $K$, and so any measurements occurring here have the \textit{potential} to be affected by our measurement of $A$ in $K$ (measurements happening in regions spacelike to $K$ should not be affected, however). We can encode all of this for the in, or out, region by specifying that we do not, or do, apply the update map for the measurement of $A$ defined below. We will discuss this in more detail shortly.

Consider, then, some compact $K\subset \spctm{S}$ and some self-adjoint operator $A\in\alg{A}(\spctm{R})$ for some region $\spctm{R}\subset K$. Given some density matrix $\rho$ (referred to simply as a state from now on, and pure if $\rho = \ket{\psi}\bra{\psi}$ for some $\ket{\psi}\in\alg{F}$), we interpret $\text{tr}(\rho A)$ as the expected value one would obtain from performing measurements of the given observable within $K$ (in the given state). Given the associated p.v.m, $P^A$, and some $\bor{B}\in\BorR$, the expected value of $P^A(\bor{B})$ in the state $\rho$, i.e. $\text{tr}(\rho P^A(\bor{B}))$, is interpreted as the probability that a measurement of the observable will have an outcome within $\bor{B}$ in the given state. In summary, the state $\rho$ tells us, via its trace against operators in $\alg{A}$, the statistics (including expectation values, variances, probabilities, etc.) of any measurements in our theory. It can be more helpful then, as is done in AQFT, to think of the state simply as some linear map $\text{tr}(\rho \, \cdot): \alg{A}\rightarrow \mathbb{C}$. 

Even without conditioning on any particular outcome, the act of measuring the observable $A$ in $K$ may affect the statistics of future measurements. We can encode this effect through the following ideal measurement update map:

\begin{definition}[Ideal measurement in QFT]\label{def:Ideal measurement in QFT}
For some self-adjoint $A\in\alg{A}$, and some resolution $\bor{R} = \lbrace \bor{B}_n \rbrace_{n\in I}$, the associated non-selective ideal measurement update map is defined, for any $X\in\alg{B}$, as
\begin{equation}\label{eq:ideal_measurement_qft}
    \mathcal{E}_{A,\bor{R}}(X) = \sum_{n\in I} E_n X E_n \; ,
\end{equation}
where $E_n = P^A(\bor{B}_n)$, and where we have restricted $X$ to the subalgebra of bounded operators $\alg{B}\subset\alg{A}$ to avoid complications regarding operator domains.
\end{definition}
\noindent We take an operationalist viewpoint in the sense that we do not care how this measurement is implemented on the system; we simply assume that we do \textit{something} to the system which amounts to a change in the statistics as described by the above update map.

For any observable $X\in\alg{B}$ if we measure its expectation value in the in-region $K^-$, then, since this should not be affected by our measurement of $A$ in $K$, we compute the associated expectation value as $\text{tr}(\rho X)$. If we measure it in the out-region $K^+$ then we compute it using the update map as $\text{tr}(\rho\mathcal{E}_{A,\bor{R}}(X))$, as this measurement of $X$ now has the potential to be affected by our measurement of $A$ in $K$. It is in this precise sense that we have encoded the effect of our measurement of $A$, and that it took place in $K$.

For simplicity, we ignore the case where one wants to measure an expectation value in the remaining subset of the spacetime, $\spctm{S}\setminus ( K^+ \cup K^- )$. We essentially consider this subset as being `blocked off' for the purposes of making our `controlled' measurement of $A$ in $K$.

For any operators localisable in regions spacelike to $K$, i.e. any $X\in \alg{B}(K^{\perp})$, the above tells us to compute the expectation value both as $\text{tr}(\rho X)$ \textit{and} as $\text{tr}(\rho\mathcal{E}_{A,\bor{R}}(X))$. In fact, the Einstein causality condition --- that spacelike operators commute --- ensures that $[A , X]=0$, and thus $[E_n , X]=0$, which then implies that 
\begin{align}\label{eq:locality_of_ideal_map_phi}
    \mathcal{E}_{A,\bor{R}}(X) & = \sum_{n\in I} E_n X E_n \nonumber
    \\
    & = \sum_{n\in I} {E_n}^2 X \nonumber
    \\
    & = \left( \sum_{n\in I} E_n \right) X \nonumber
    \\
    & = X \; ,
\end{align}
where line 3 follows as ${E_n}^2 = E_n$, and the last line follows as $\sum_{n\in I}E_n = \mathds{1}$. Thus, $\mathcal{E}_{A,\bor{R}}$ acts trivially on $\alg{B}(K^{\perp})$, which implies that $\text{tr}(\rho \mathcal{E}_{A,\bor{R}}(X)) = \text{tr}(\rho X)$, and so the statistics of measurements happening in regions spacelike to $K$ are unchanged by our measurement in $K$. Since we are not conditioning on any outcome from our measurement of $A$, this physically makes sense. If we had conditioned on some outcome we could have expected to see some change to statistics in spacelike regions, which would be due to any spacelike correlations in our theory.

In what follows we will be particularly interested in ideal measurements of smeared field operators $\phi(f)$. Similarly to $\hat{x}$ in the non-relativistic case, the spectrum of the self-adjoint $\phi(f)$ is the entire real line, i.e. it is an unbounded operator. As with $\hat{x}$, one can still make sense of the ideal measurement update map, provided we have some resolution $\bor{R} = \lbrace \bor{B}_n \rbrace_{n\in I}$ (see the discussion in Section~\ref{sec:Interpreting resolution}). Specifically, given the p.v.m $P^{\phi(f)}$ for $\phi(f)$ we can form the projectors $E_n = P^{\phi(f)}(\bor{B}_n)$ that enter into~\eqref{eq:ideal_measurement_qft}. $\mathcal{E}_{\phi(f),\bor{R}}$ is then a well defined map from $\alg{B}$ to itself.

We again stress our operationalist viewpoint, namely, that we do not assume, or attempt to describe, some apparatus or procedure that implements an ideal measurement of $\phi(f)$ in $K$; we only assume that our actions on the system amount to a change of the statistics in $K^+$ which is described by the update map $\mathcal{E}_{\phi(f),\bor{R}}$. If such an update map violates causality by enabling a superluminal signal, then we can conclude that no such apparatus exists which implements the ideal measurement, and hence we do not have to concern ourselves with the question of how to describe any particular apparatus or procedure. We leave the question of whether $\mathcal{E}_{\phi(f),\bor{R}}$ enables a superluminal signal to Section~\ref{sec:The acausality of an ideal measurement of a smeared field}, and for now just note that the update map $\mathcal{E}_{\phi(f),\bor{R}}$ is perfectly well-defined in the theory, just as $\mathcal{E}_{\hat{x},\bor{R}}$ is for an ideal measurement of $\hat{x}$ in non-relativistic quantum mechanics.

For $\hat{x}$, the projectors $P^{\hat{x}}(\bor{B}_n)$ can be associated to binary (yes/no) measurements of whether the $x$-position falls within some Borel subset, $\bor{B}_n$, of the $x$-axis. Following Remark~\ref{remark:Interpretation of a smeared field operator}, the projectors $P^{\phi(f)}(\bor{B}_n)$ for $\phi(f)$ are associated to binary measurements of whether the measured amplitude of excitation of the classical mode $\Delta f$ falls within some $\bor{B}_n\subseteq \mathbb{R}$ (and thus has nothing to do with position).

It is worth highlighting that the locality of the operator $\phi(f)$, and all its projectors $P^{\phi(f)}(\bor{B})$ for any $\bor{B}\in\BorR$ (`locality' in the sense that they all commute with any spacelike operators $X\in\alg{B}(K^{\perp})$), leads to the fact that $\mathcal{E}_{\phi(f),\bor{R}}(X) = X$ for any $X\in\alg{B}(K^{\perp})$, and hence that expectation values/probabilities of measurements in spacelike regions are unaffected by a non-selective ideal measurement of $\phi(f)$ in $K$. This means that, as long as you do not consider Sorkin's pivotal scenario in Section~\ref{sec:sorkinscenario}, the projection postulate (if utilised for local operators as is done here) raises no issues with relativistic causality. This is in contrast to the usual perspective that the projection postulate is generally \textit{acausal} in QFT, as it `collapses the wavefunction across all of space simultaneously', or something along those lines. In fact, as we have just shown, one can set up the projection postulate and ideal measurements for local observables in a perfectly local way that respects causality, at least for the time being. The question of consistency with relativistic causality ends up being more subtle, and indeed one must confront Sorkin's scenario in Section~\ref{sec:sorkinscenario} to reveal any potential issues.

\subsection{Local and Causal Operations in QFT}\label{sec:Local and causal operations in QFT}

In Quantum Information (QI), and other fields in quantum physics, one considers operations that are more general than ideal measurements. Any physical operations on a given system are described in QI via completely positive trace-preserving update maps on the density matrix. In the dual picture, where the update maps act on the operators, the trace-preserving condition becomes a unit-preserving condition.

Similarly, in our QFT setup we will use the term \textit{update map} for any completely-positive map on the bounded operators, $\mathcal{E}:\alg{B}\rightarrow\alg{B}$, that is unit-preserving, i.e. $\mathcal{E}(\mathds{1}) = \mathds{1}$. This ensures that the state is always normalised: $\text{tr}(\rho \mathcal{E}(\mathds{1})) = \text{tr}(\rho \mathds{1}) = 1$. Note, the term \textit{quantum channel} is also used in QI.

The map, $\mathcal{E}_{\phi(f),\bor{R}}$, for an ideal measurement of $\phi(f)$ with resolution $\bor{R}$, is one such example of an update map in our theory. As another example, given any unitary operator $U\in\alg{B}$ the map $X\mapsto U X U^{-1}$, for $X\in\alg{B}$, defines an update map.

\subsubsection{Local Update Maps}

We now formalise the notion of a local update map.
\begin{definition}[local]\label{def:local}
An update map $\mathcal{E}$ is \textit{local} to a compact subset $K\subset \spctm{S}$ if
\begin{equation}
    \mathcal{E}(X) = X \; ,
\end{equation}
for all $X\in\alg{B}(K^{\perp})$.
\end{definition}
That is, the map acts trivially on operators localisable in regions spacelike to $K$. Again, $\mathcal{E}_{\phi(f),\bor{R}}$ is an example of a local update map. In particular, it is local to $\supp f$. This was essentially shown in~\eqref{eq:locality_of_ideal_map_phi}. The locality of any update map $\mathcal{E}$ (local to $K$) ensures it cannot affect the statistics of any measurements occurring in spacelike regions, as for any $X\in\alg{B}(K^{\perp})$ we have $\text{tr}(\rho \mathcal{E}(X)) = \text{tr}(\rho X)$.

Given any self-adjoint $A\in\alg{A}(\spctm{R})$, the associated \textit{unitary kick}, defined on any $X\in\alg{B}$ as
\begin{equation}
    \mathcal{U}_A (X) = e^{iA}Xe^{-iA} \; ,
\end{equation}
defines an update map local to any compact $K\supset\spctm{R}$. To see this, we note that if $X\in\alg{B}(K^{\perp})$, then $[A,X] = 0$, and hence $[e^{iA},X]=0$, which implies the locality condition $\mathcal{U}_A(X) = X$.

\subsubsection{Kraus Update Maps}\label{sec:Kraus update maps}

We now introduce the notion of \textit{Kraus update maps} as a single framework through which we can discuss a variety of local update maps (including ideal measurements and other more general operations) constructed from local operators in the theory. Given some self-adjoint operator $A\in\alg{A}$ (possibly unbounded), and some bounded, what we will call, \textit{Kraus function} $\kappa : \mathbb{R}\rightarrow\mathbb{C}$, we can use functional calculus to define the bounded operator $\kappa(A)\in\alg{B}$. 
Going further, we can define a {family of Kraus functions} as follows. Consider some measure space $(\Gamma , \sigma (\Gamma) , \nu )$, where $\Gamma$ is a set, $\sigma (\Gamma )$ is a $\sigma$-algebra of subsets of $\Gamma$, and $\nu$ is a non-negative measure defined on $\sigma(\Gamma)$. Now, consider a family of bounded Kraus functions $\kappa(\cdot , \gamma): \mathbb{R} \rightarrow\mathbb{C}$, indexed by $\gamma\in\Gamma$. For each such function we can use functional calculus to define the bounded operator $\kappa(A , \gamma)$. The use of the measure space $(\Gamma , \sigma (\Gamma) , \nu )$ allows us to define integration over the family of Kraus functions in what follows.
\begin{definition}[Kraus update map]
A \textit{Kraus update map}, $\mathcal{E}_{A,\kappa}:\alg{B}\rightarrow\alg{B}$, for some (possibly unbounded) self-adjoint operator $A\in\alg{A}$ and some family of bounded Kraus functions $\lbrace \kappa(\cdot , \gamma)\rbrace_{\gamma\in\Gamma}$, is of the form
\begin{equation}
    \mathcal{E}_{A, \kappa}(X) = \int_{\Gamma}d\nu (\gamma ) \, \kappa(A,\gamma) X \kappa(A, \gamma )^{\dagger} \; ,
\end{equation}
for any $X\in\alg{B}$.
\end{definition}
The unit-preserving condition on the update map $\mathcal{E}_{A,\kappa}$ amounts to the following normalisation condition on the Kraus functions:
\begin{equation}\label{eq:kappa normalisation condition}
    \int_{\Gamma}d\nu (\gamma ) \, \kappa(\lambda,\gamma) \kappa(\lambda, \gamma )^* = 1 \; ,
\end{equation}
for any $\lambda\in\mathbb{R}$. This ensures that $\mathcal{E}_{A,\kappa}(\mathds{1}) = \mathds{1}$.

The two examples of update maps discussed above are of this Kraus form, as we show now. 
\begin{example}[Unitary kicks]\label{ex:unitary kick kraus}
If $\Gamma = \lbrace \gamma \rbrace$ is some singleton set, and $\nu (\gamma) = 1$, the integral reduces as
\begin{equation}
    \mathcal{E}_{A, \kappa}(X) = \int_{\Gamma}d\nu (\gamma ) \, \kappa(A,\gamma) X \kappa(A, \gamma )^{\dagger} = \kappa(A,\gamma) X \kappa(A, \gamma )^{\dagger} \; .
\end{equation}
Now, taking $\kappa(\lambda , \gamma ) = e^{i\lambda}$, for $\lambda\in\mathbb{R}$, we get the unitary kick with respect to $A$. That is, $\mathcal{E}_{A, \kappa}(X) = \mathcal{U}_A(X) = e^{iA}Xe^{-iA}$.
\end{example}

\begin{example}[Ideal measurements]\label{ex:ideal measurement}
For the case of ideal measurements, we can take $\Gamma = I$, where $I$ is some countable indexing set, and for each $\gamma\in I$ we set $\nu (\gamma ) = 1$. The integral in the Kraus update map definition becomes a sum:
\begin{equation}
    \mathcal{E}_{A, \kappa}(X) = \int_{\Gamma}d\nu (\gamma ) \, \kappa(A,\gamma) X \kappa(A, \gamma )^{\dagger} = \sum_{\gamma\in I} \kappa(A,\gamma) X \kappa(A, \gamma )^{\dagger} \; .
\end{equation}
If we then take $\kappa(\lambda , \gamma ) = 1_{\bor{B}_{\gamma}}(\lambda)$, where $\bor{B}_{\gamma}$ are the Borel sets appearing in some resolution $\bor{R}=\lbrace \bor{B}_{\gamma}\rbrace_{\gamma\in I}$ indexed by $\gamma\in I$, then the Kraus operators in the sum above are self-adjoint, and further can be written as
\begin{equation}
    \kappa(A, \gamma ) = 1_{\bor{B}_{\gamma}}(A) = P^A(\bor{B}_{\gamma}) = E_{\gamma} \; ,
\end{equation}
using~\eqref{eq:pvm_as_indicator_function} for the middle equality. Therefore, this Kraus update map reduces to the ideal measurement update map, i.e. $\mathcal{E}_{A, \kappa}(X) = \mathcal{E}_{A,\bor{R}}(X)$. Note, one can consider `less' ideal measurements too by simply replacing the step functions with smooth approximations, for example.
\end{example}
As an example not yet considered, let us now introduce the update map associated with weak, or Gaussian, measurements (see~\cite{Jubb_2022} for further discussion in the QFT context). Note, we do not assume these weak/Gaussian update maps arise from the tracing out of some auxiliary system coupled to our QFT, as is often the case for weak measurements. We take an operationalist perspective (as we do for all of our update maps) and only assume that some process takes place which gives rise to this precise form of the update map on our system of interest, and we remain agnostic to its origin.
\begin{example}[Weak/Gaussian measurements]\label{ex:weak measurement}
Take $\Gamma = \mathbb{R}$ with $\sigma(\Gamma)$ the usual $\sigma$-algebra on $\mathbb{R}$, set $\nu$ to be the usual Lebesgue measure on $\mathbb{R}$, and set
\begin{equation}
    \kappa(\lambda , \gamma ) = \frac{e^{-\frac{(\lambda - \gamma)^2}{4\sigma^2}}}{(2\pi \sigma^2)^{1/4}} \; ,
\end{equation}
for some $\sigma >0$ defining the width of the Gaussian. The Kraus update map then becomes the weak measurement update map, denoted $\mathcal{W}_{A,\sigma}$. Specifically,
\begin{equation}
    \mathcal{E}_{A, \kappa}(X) = \mathcal{W}_{A, \sigma}(X) = \frac{1}{\sqrt{2\pi \sigma^2}} \int_{\mathbb{R}}d\gamma \, e^{-\frac{(A - \gamma)^2}{4\sigma^2}} X e^{-\frac{(A - \gamma)^2}{4\sigma^2}}  \; .
\end{equation}
\end{example}

\begin{example}[$L^2$-Kraus update maps]\label{ex:L2 kraus}
More generally, instead of a Gaussian function one can consider any normalised complex-valued $k\in L^2(\mathbb{R};\mathbb{C})$, and set $\kappa(\lambda , \gamma) = k(\lambda - \gamma)$. We then have a Kraus update map of the form
\begin{equation}
    \mathcal{E}_{A,k}(X) = \int_{\mathbb{R}}d\gamma \, k(A - \gamma )X k(A- \gamma)^{\dagger} \; ,
\end{equation}
for any $X\in\alg{B}$. Here we have kept $\Gamma = \mathbb{R}$ (with the same $\sigma$-algebra) as we had in the Gaussian measurement example. The fact that $k$ is normalised ensures that $\mathcal{E}_{A,k}(\mathds{1}) = \mathds{1}$.
\end{example}

\vspace{5mm}
An important consequence of the general form of these Kraus update maps is that, if $A$ is localisable in $\spctm{R}$, then any associated Kraus maps, $\mathcal{E}_{A,\kappa}$, will be local to any compact $K\supset\spctm{R}$. This can be seen from the following calculation, assuming $X\in\alg{B}(K^{\perp})$:
\begin{align}
    \mathcal{E}_{A,\kappa}(X) & = \int_{\Gamma}d\nu (\gamma ) \, \kappa(A,\gamma) X \kappa(A, \gamma )^{\dagger}\nonumber
    \\
    & = \left( \int_{\Gamma}d\nu (\gamma ) \, \kappa(A,\gamma)\kappa(A, \gamma )^{\dagger} \right) X \nonumber
    \\
    & = \mathcal{E}_{A,\kappa}(\mathds{1})X
    \\
    & =  X \; ,
\end{align}
as desired for a map to be local to $K$. Line 2 follows from Einstein causality, i.e. $[A,X]=0$ in this case, and the last line follows from the unit-preserving condition on $\mathcal{E}_{A,\kappa}$.

\subsubsection{Multiple Local Update Maps}\label{sec:Multiple local update maps}

In~\cite{Hellwig_Kraus}, Hellwig and Kraus described a general framework for discussing multiple measurements in QFT in a manner consistent with relativistic causality, though they did not consider the key scenario later introduced by Sorkin in~\cite{Sorkin_impossible}. Here we briefly recap some aspects of~\cite{Hellwig_Kraus}.

Consider two disjoint, causally convex, compact subsets $K,K'\subset \spctm{S}$. Without loss of generality we take $K'\subset K^-$, or equivalently $K\subset K'^+$ (e.g. Fig.~\ref{fig:two_update_maps_1}).

\begin{figure}
\begin{subfigure}[b]{0.32\textwidth}
     \centering
     \includegraphics[width=\textwidth]{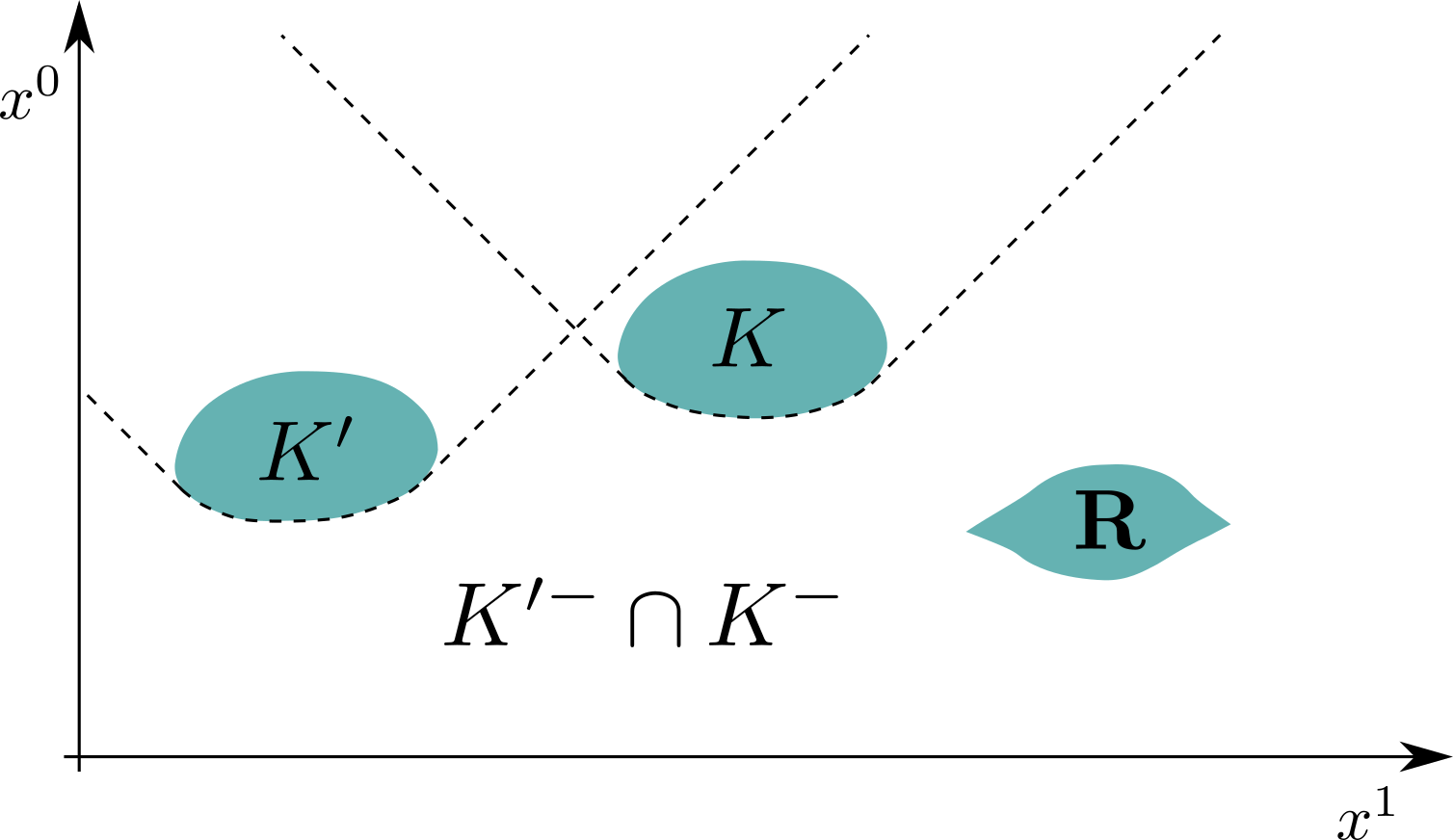}
     \caption{$K'^-\cap K^-$ is under the dashed lines. $K'$ and $K$ are spacelike in this example.}
     \label{fig:two_update_maps_1}
 \end{subfigure}
 \hfill
 \begin{subfigure}[b]{0.32\textwidth}
     \centering
     \includegraphics[width=\textwidth]{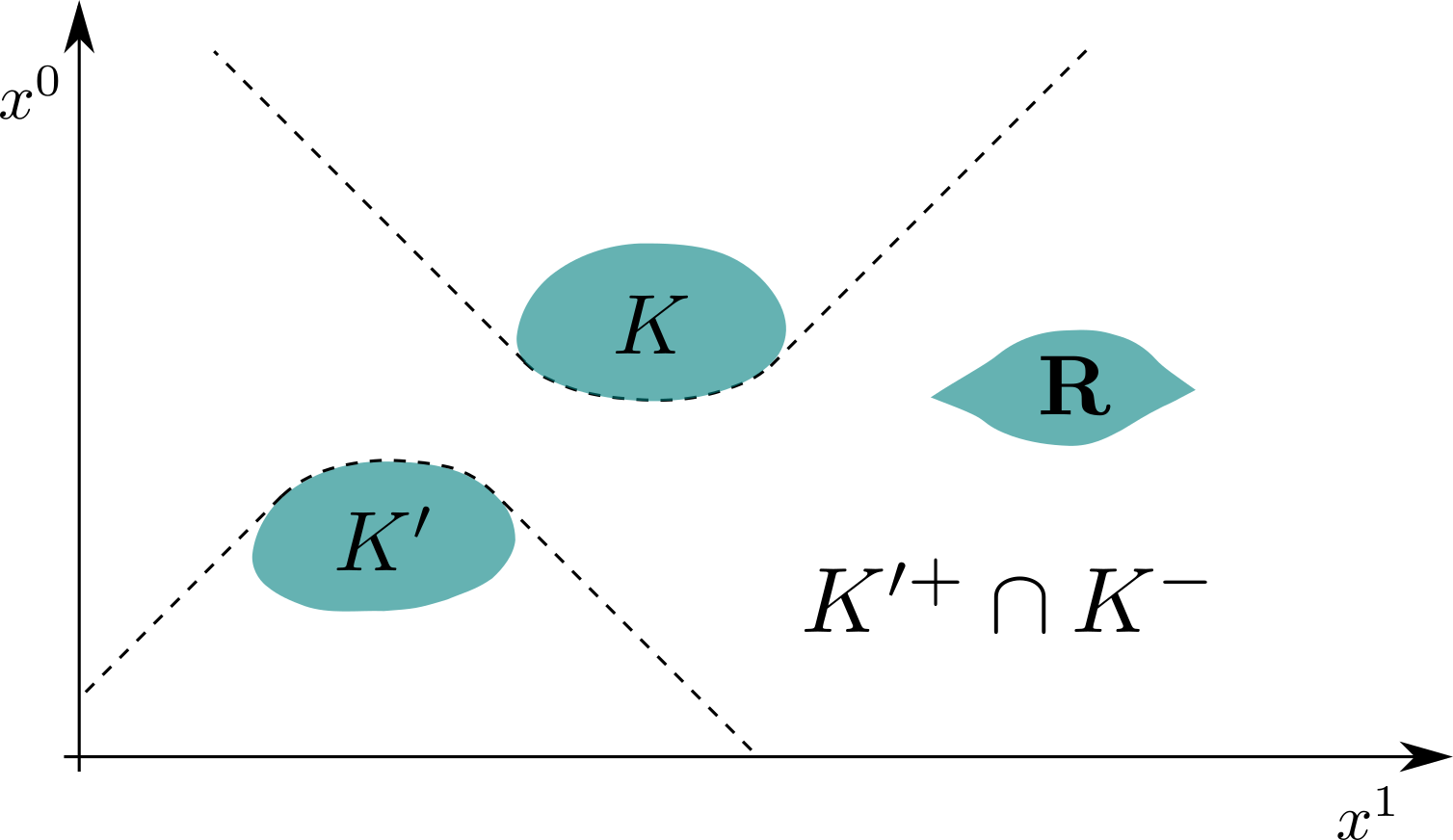}
     \caption{$K'^+\cap K^-$ is between the dashed lines. $K$ is partly to the future of $K'$ in this example.}
     \label{fig:two_update_maps_2}
 \end{subfigure}
  \hfill
 \begin{subfigure}[b]{0.32\textwidth}
     \centering
     \includegraphics[width=\textwidth]{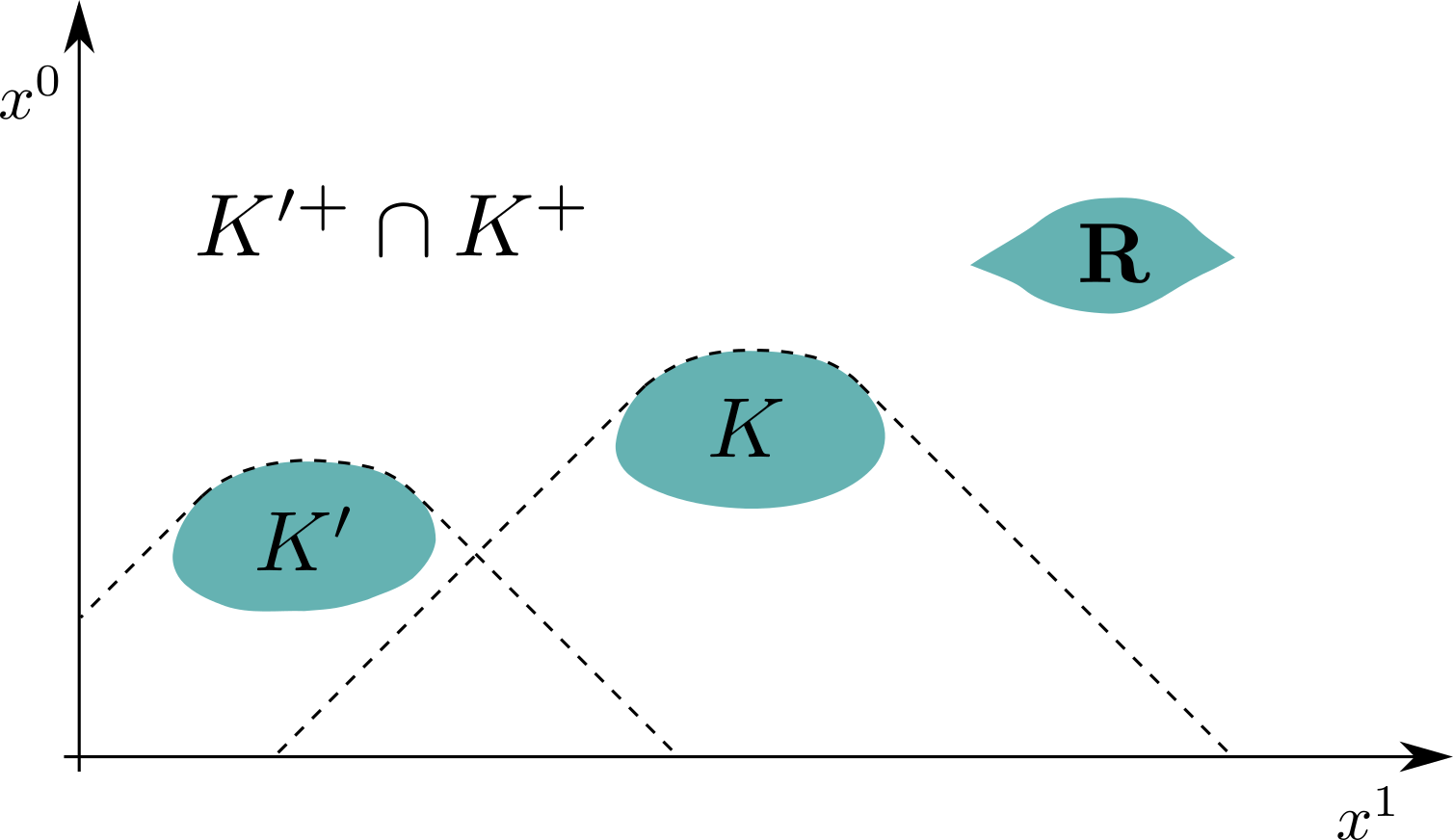}
     \caption{$K'^+\cap K^+$ is above the dashed lines. Note that both $K$ and $\spctm{R}$ are spacelike to $K'$ here, as described in Remark~\ref{remark:Independence from spacelike Kraus update maps}.}
     \label{fig:two_update_maps_3}
 \end{subfigure}
\end{figure}

Now consider two update maps, $\mathcal{E}$ and $\mathcal{E}'$, local to $K$ and $K'$ respectively. These update maps describe operations taking place within their respective spacetime subsets. For a continuum spacetime, $\spctm{M}$, the fact that $K'\subset K^-$ implies there always exists some frame, or more specifically, some time function, $t:\spctm{M}\rightarrow\mathbb{R}$, such that according to this time function the operation in $K'$ is completed before the operation in $K$ begins, i.e. $\text{max}_{x\in K}\lbrace t(x)\rbrace < \text{min}_{x\in K'}\lbrace t(x)\rbrace$. If $K$ is at least partly to the future of $K'$ (e.g. Fig.~\ref{fig:two_update_maps_2}) then this is not true the other way around; there does not exist a time function where the operation in $K$ is seen to be completed first. Conversely, if $K$ and $K'$ are spacelike (e.g. Fig.~\ref{fig:two_update_maps_1} and~\ref{fig:two_update_maps_3}), then one can find time functions in which either operation is seen to be completed first. For a causet, $\spctm{C}$, a similar statement is that there exists some natural labelling of $\spctm{C}$ such that all the points in $K'$ have smaller labels than all the points in $K$.

Consider some region $\spctm{R}$ within which we want to compute the expectation value of some operator $X\in\alg{B}(\spctm{R})$. For simplicity, we assume that $\spctm{R}$ lies somewhere in the in/out-regions for $K$ and $K'$ (anywhere else is off limits). There are a few cases to consider (note, these cases are not necessarily mutually exclusive):

\begin{itemize}
    \item $\spctm{R}\subseteq {K'}^-\cap K^-$ (Fig.~\ref{fig:two_update_maps_1}): Any measurement in $\spctm{R}$ is definitely not to the future of the operations in $K$ and $K'$ (described by $\mathcal{E}$ and $\mathcal{E}'$ respectively), and hence should not be affected by both operations. Thus, we compute the expectation value as $\text{tr}(\rho X)$.
    \item $\spctm{R}\subseteq {K'}^+\cap K^-$ (Fig.~\ref{fig:two_update_maps_2}): Any measurement in $\spctm{R}$ cannot be affected by the operation in $K$, but has the potential to be affected by the operation in $K'$, and hence we compute the expectation value as $\text{tr}(\rho \mathcal{E}'(X))$.
    \item $\spctm{R}\subseteq {K'}^+\cap K^+$ (Fig.~\ref{fig:two_update_maps_3}): Any measurement in $\spctm{R}$ has the potential to be affected by both operations, and hence we compute the expectation value as $\text{tr}(\rho \mathcal{E}(\mathcal{E}'(X)))$. The ordering of this composition --- $\mathcal{E}$ first and then $\mathcal{E}'$ --- may be in the opposite order to what one might expect, since we previously discussed how the operation $\mathcal{E}'$ in $K'$ can be seen to happen before $\mathcal{E}$ in $K$. This reversed order arises as we are in the dual picture where the update maps act on the operators instead of the state. 
\end{itemize}

\begin{remark}[Spacelike regions and the commutativity of Kraus update maps]\label{remark:Spacelike regions and update map commutativity}
In the case where $K$ and $K'$ are spacelike (e.g. Fig.~\ref{fig:two_update_maps_1} and~\ref{fig:two_update_maps_3}) one would like the two update maps to commute as $\mathcal{E}(\mathcal{E}'(\cdot )) = \mathcal{E}'(\mathcal{E}(\cdot ))$, since there is no canonical way to causally order $K$ and $K'$, and hence no canonical order to apply their update maps. Indeed, if the two update maps are of the Kraus form above, i.e. $\mathcal{E}=\mathcal{E}_{A,\kappa}$ and $\mathcal{E}'=\mathcal{E}_{A',\kappa'}$ for operators $A,A'\in\alg{A}$ localisable in regions contained within $K$ and $K'$ respectively, and some Kraus families $\kappa$ and $\kappa'$, then Einstein causality ensures that the two update maps commute, i.e. $\mathcal{E}_{A,\kappa}(\mathcal{E}_{A',\kappa'}(\cdot )) = \mathcal{E}_{A',\kappa'}(\mathcal{E}_{A,\kappa}(\cdot ))$. This further ensures that any statistics are independent of the arbitrary order we give to the two spacelike regions (i.e. which one our time function or natural labelling says happens first).
\end{remark}

\begin{remark}[Independence from spacelike Kraus update maps]\label{remark:Independence from spacelike Kraus update maps}
If, in addition, $\spctm{R}$ is spacelike to either $K$ or $K'$, then Einstein causality ensures that the associated update drops out. For example, if $\spctm{R}\subseteq K'^+\cap K^+$ is spacelike to $K'$ (and assuming $K$ and $K'$ are spacelike) as in Fig.~\ref{fig:two_update_maps_3}, then the expectation value of $X$ simplifies as
\begin{align}\label{eq:kraus map drops out when spacelike to X}
    \text{tr}(\rho \mathcal{E}_{A',\kappa'}(\mathcal{E}_{A,\kappa}(X))) & = \text{tr}(\rho \mathcal{E}_{A,\kappa}(\mathcal{E}_{A',\kappa'}(X))) \nonumber
    \\
    & = \text{tr}(\rho \mathcal{E}_{A,\kappa}(X)) \; ,
\end{align}
where the first line follows from the commutativity of the spacelike update maps, and the second line follows as $\mathcal{E}_{A',\kappa'}$ is local to $K'$, and hence it acts trivially on $X$. Physically this makes sense, as the statistics in $\spctm{R}$ should not depend on any operations taking place in the spacelike subset $K'$.
\end{remark}

One can extend the above discussion to any finite number of mutually disjoint compact subsets $K_1,...,K_N$, each with their respective local update maps $\mathcal{E}_1,...,\mathcal{E}_N$, but in fact, as first noted by Sorkin in~\cite{Sorkin_impossible}, there is more to consider with regards to causality in the case of only two update maps.

\subsection{Sorkin's Scenario and Causality} \label{sec:sorkinscenario}

In~\cite{Sorkin_impossible} Sorkin essentially took the spacetime setup from Remark~\ref{remark:Independence from spacelike Kraus update maps} (Fig.~\ref{fig:two_update_maps_3}), and imagined a slightly modified setup with $\spctm{R}\subseteq K'^+\cap K^+$ still spacelike to $K'$, but with $K$ and $K'$ no longer mutually spacelike (see Fig.~\ref{fig:sorkin_scenario}). Though Sorkin did not describe the situation in terms of Kraus update maps, we will do so in order to follow on from the previous section.
\begin{figure}
 \centering
 \includegraphics[width=0.65\textwidth]{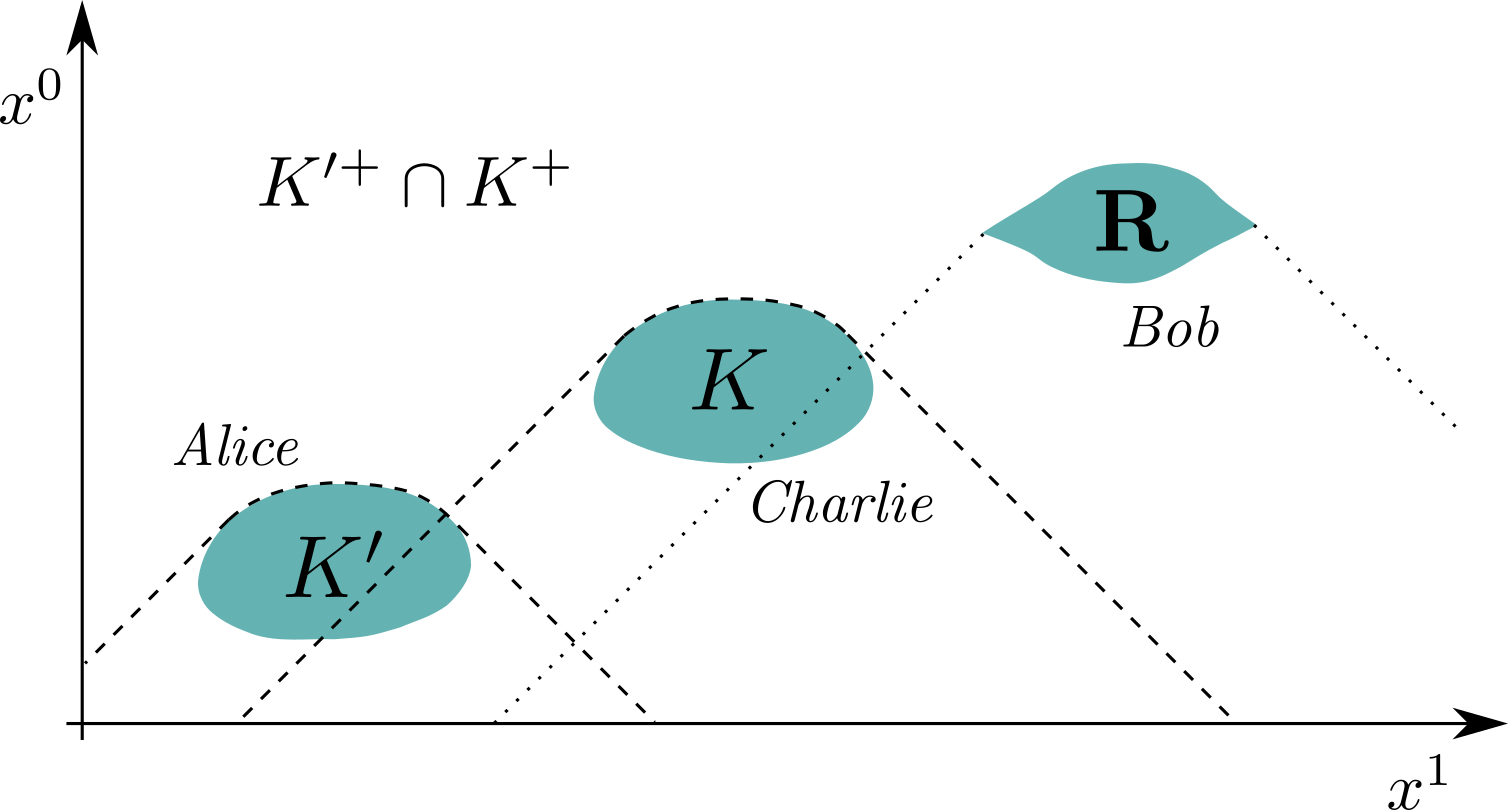}
 \caption{Sorkin's scenario. Bob's region, $\spctm{R}$, lies partly to the future of Charlie's subset, $K$, which lies partly to the future of Alice's subset, $K'$, unlike in Fig.~\ref{fig:two_update_maps_3} where $K$ and $K'$ were spacelike. $K'^+\cap K^+$ lies above the dashed lines, and thus $\spctm{R}\subset K'^+\cap K^+$. As can be seen by the dotted lines from Bob's region, $\spctm{R}$ is spacelike to Alice's subset $K'$.}
 \label{fig:sorkin_scenario}
\end{figure}

Let us also introduce three agents, Alice, Bob and Charlie, to aid the discussion. The agent \textit{Alice} performs the operation described by $\mathcal{E}_{A',\kappa'}$ in $K'$, \textit{Charlie} performs the operation described by $\mathcal{E}_{A,\kappa}$ in $K$, and \textit{Bob} measures the expectation value of $X$ in $\spctm{R}$. We assume that sufficiently many copies of the experiment can be set up in parallel so that the sample mean of Bob's measured values --- one measured value for each copy of the setup --- is as accurate an estimate of the true expectation value of $X$ as one could desire. 

With this setup, Section~\ref{sec:Multiple local update maps} tells us to compute Bob's expectation value as $\text{tr}(\rho \mathcal{E}_{A',\kappa'}(\mathcal{E}_{A,\kappa}(X)))$. On the other hand, as $\spctm{R}$ is spacelike to $K'$, we expect on physical grounds that Bob's expectation value will be independent from Alice's update map $\mathcal{E}_{A',\kappa'}$ in $K'$. That is, Bob's expectation value should simply be $\text{tr}(\rho \mathcal{E}_{A,\kappa}(X))$ (c.f. Remark~\ref{remark:Independence from spacelike Kraus update maps}). This latter value is also what one would compute if Alice had not done anything. Now, it may be the case that both values are equal, i.e.
\begin{equation}\label{eq:desired causality condition for kraus maps}
    \text{tr}(\rho \mathcal{E}_{A',\kappa'}(\mathcal{E}_{A,\kappa}(X))) = \text{tr}(\rho \mathcal{E}_{A,\kappa}(X)) \; ,
\end{equation}
but this equality may not be satisfied in general. Alarmingly, if it is not satisfied then Alice and Bob can exploit the difference of these two values to send a \textit{pathological} faster-than-light signal between them (specifically from Alice in $K'$ to Bob in $\spctm{R}$) as follows. If Bob measures the value $\text{tr}(\rho \mathcal{E}_{A,\kappa}(X))$ then it means that Alice has done nothing, and if Bob measures the value $\text{tr}(\rho \mathcal{E}_{A',\kappa'}(\mathcal{E}_{A,\kappa}(X)))$, then it means Alice has performed her operation. Thus, given that these two values are not the same, Bob can tell if Alice has or has not done her operation. The two agents can decide beforehand that Alice performing her operation should be interpreted as a binary $\texttt{1}$, and Alice not doing anything should be interpreted as a binary $\texttt{0}$. Therefore, by either doing her operation or not, and Bob measuring his expectation value, Alice can send a single bit --- a signal --- to Bob. Since $\spctm{R}$ is spacelike to $K'$, this signal must be faster than light. This is discussed in more detail in~\cite{Sorkin_impossible,Beckman_2001,Benincasa_2014,Borsten_2021,Jubb_2022}.

Of course, such causality violating superluminal signals are not possible, and hence it \textit{must} be the case that~\eqref{eq:desired causality condition for kraus maps} is satisfied (c.f.~\eqref{eq:kraus map drops out when spacelike to X} in Remark~\eqref{remark:Independence from spacelike Kraus update maps}) for any physical operations performed by the agents. That is, for any physical operations that Alice and Charlie can do, the associated descriptions in terms of the update maps $\mathcal{E}_{A',\kappa'}$ and $\mathcal{E}_{A,\kappa}$ respectively, must satisfy~\eqref{eq:desired causality condition for kraus maps} for any operator $X$ that Bob wishes to measure.

Note, if Charlie's does nothing (i.e. we remove the map $\mathcal{E}_{A,\kappa}$ from~\eqref{eq:desired causality condition for kraus maps}), then~\eqref{eq:desired causality condition for kraus maps} is trivially satisfied due to the locality of Alice's map $\mathcal{E}_{A',\kappa'}$ (in this case both sides reduce to $\text{tr}(\rho X)$). Therefore, this potential issue of a superluminal signal (between Alice in $K'$ and Bob in $\spctm{R}$) has only arisen because of the presence of Charlie's map $\mathcal{E}_{A,\kappa}$ local to $K$. \textit{We thus think of~\eqref{eq:desired causality condition for kraus maps} as a condition that Charlie's map $\mathcal{E}_{A,\kappa}$ has to satisfy to be physically realisable; it is the requirement that $\mathcal{E}_{A,\kappa}$ does not enable a superluminal signal between the two spacelike subsets $K'$ and $\spctm{R}$.} 

Surprisingly, while Einstein causality is enough to ensure that~\eqref{eq:kraus map drops out when spacelike to X} holds, it is not enough to ensure that~\eqref{eq:desired causality condition for kraus maps} holds. Specifically, Einstein causality is enough to ensure that $\mathcal{E}_{A,\kappa}$ does not enable a superluminal signal between $K'$ and $\spctm{R}$ when $K'$ and $K$ are spacelike (Fig.~\ref{fig:two_update_maps_3}), but it is not enough if $K'$ and $K$ are at least partly timelike related (Fig.~\ref{fig:sorkin_scenario}). In fact, in~\cite{Jubb_2022} there are several simple examples given where~\eqref{eq:desired causality condition for kraus maps} is violated, as well as many examples of where~\eqref{eq:desired causality condition for kraus maps} is satisfied. For clarity, we now review some of these examples.

\begin{example}[\textit{No violation} - unitary kick with a smeared field]\label{ex:unitary kick causal}
Consider some test function $f$ supported in $K$, and let Charlie do the associated unitary kick with the smeared field $\phi(f)$, i.e. $\mathcal{E}_{A,\kappa} = \mathcal{U}_{\phi(f)}$, as in Example~\ref{ex:unitary kick kraus}. As was shown in~\cite{Jubb_2022}, this update map has the useful property that if $X\in\alg{B}(\spctm{R})$, then $\mathcal{U}_{\phi(f)}(X)\in\alg{B}(\spctm{R})$ also. Thus, given any map $\mathcal{E}_{A',\kappa'}$ local to $K'$ (or more generally any map $\mathcal{E}'$ local to $K'$, even if it is not of Kraus form), then $\mathcal{E}_{A',\kappa'}$ will act trivially on $\mathcal{U}_{\phi(f)}(X)$. This follows from the locality of $\mathcal{E}_{A',\kappa'}$, and the fact that $\mathcal{U}_{\phi(f)}(X)$ is localisable in $\spctm{R}\subset K'^{\perp}$. Thus, $\mathcal{E}_{A',\kappa'}( \mathcal{U}_{\phi(f)}(X) ) = \mathcal{U}_{\phi(f)}(X)$, and hence~\eqref{eq:desired causality condition for kraus maps} is satisfied for any state $\rho$. That is, a unitary kick with a smeared field can never enable a superluminal signal \textit{a la} Sorkin, and is thus physically realisable (at least in principle).
\end{example}

\begin{example}[\textit{No violation} - $L^2$-Kraus update map with a smeared field]\label{ex:weak measurement causal}
Consider any normalised $L^2$ function $k:\mathbb{R}\rightarrow\mathbb{C}$, and let Charlie do the associated $L^2$-Kraus map for a smeared field $\phi(f)$, i.e. $\mathcal{E}_{A,\kappa} = \mathcal{E}_{\phi(f),k}$, as in Example~\ref{ex:L2 kraus}. It was also shown in~\cite{Jubb_2022} that such update maps have the same useful property as smeared field unitary kicks --- they do not change the localisation regions of operators. Following the same logic as in Example~\ref{ex:unitary kick causal}, we see that $\mathcal{E}_{\phi(f),k}$ similarly never enables a superluminal signal. Note that this includes the case of weak measurements (Example~\ref{ex:weak measurement}) of smeared fields if we set $k$ to be a Gaussian function.
\end{example}

\begin{example}[\textit{Violation} - unitary kicks/$L^2$-Kraus update maps with $\phi(f)^2$]\label{ex:acausal}
In both of the previous examples, if instead of $\phi(f)$ we take the square of the smeared field, $\phi(f)^2$, then Charlie's map is of the form $\mathcal{U}_{\phi(f)^2}$ and $\mathcal{E}_{\phi(f)^2,k}$ respectively. Then,~\eqref{eq:desired causality condition for kraus maps} can be violated for certain simple choices of the state $\rho$, Bob's operator $X\in\alg{B}(\spctm{R})$, the function $k$ in $\mathcal{E}_{\phi(f)^2,k}$, and Alice's update map $\mathcal{E}_{A',\kappa'}$ in $K'$. The violation of~\eqref{eq:desired causality condition for kraus maps} is particularly easy to see for the unitary kick case $\mathcal{U}_{\phi(f)^2}$, and so we will just show that result. Take the ground state $\rho=\gs \gsb$, take $X = e^{i\phi(g)}$ for some test function $g$ supported in $\spctm{R}$, and finally take $\mathcal{E}_{A',\kappa'} = \mathcal{U}_{\phi(h)}$ for some test function $h$ supported in $K'$. In this case the rhs~\eqref{eq:desired causality condition for kraus maps} evaluates to
\begin{align}\label{eq:rhs violation}
    \text{tr}(\rho \mathcal{E}_{A,\kappa}(X)) & = \gsb \mathcal{U}_{\phi(f)^2}(e^{i\phi(g)}) \gs \nonumber
    \\
    & = \gsb e^{-i\Delta(f,g)^2}e^{-i2\Delta(f,g)f}e^{i\phi(g)} \gs \nonumber
    \\
    & = \gsb e^{i\phi(\tilde{g})} \gs \nonumber
    \\
    & = e^{-\frac{W(\tilde{g},\tilde{g})}{2}} \; ,
\end{align}
where line 2 follows from Equation (17) in~\cite{Jubb_2022}, and in line 3 we have defined $\tilde{g} = g-2\Delta(f,g)f$, and in the last line we have used~\eqref{eq:exp_phi_f_gs}. On the other hand, using the fact that $h$ and $g$ have spacelike supports, one can show that the lhs of~\eqref{eq:desired causality condition for kraus maps} evaluates to
\begin{align}\label{eq:lhs violation}
    \text{tr}(\rho \mathcal{E}_{A',\kappa'}(\mathcal{E}_{A,\kappa}(X))) & = \gsb \mathcal{U}_{\phi(h)}(\mathcal{U}_{\phi(f)^2}(e^{i\phi(g)})) \gs \nonumber
    \\
    & = e^{-i2\Delta(f,g)\Delta(f,h)}e^{-\frac{W(\tilde{g},\tilde{g})}{2}} \; .
\end{align}
Since~\eqref{eq:rhs violation} and~\eqref{eq:lhs violation} are not equal in general, we have a violation of~\eqref{eq:desired causality condition for kraus maps}. Physically, the expectation value of $e^{i\phi(g)}$ that Bob measures in $\spctm{R}$ depends on whether Alice does a unitary kick with $\phi(h)$ in the $K'$ or not. Such an acausal dependency of Bob's expectation value on Alice's actions in a subset of spacetime that is spacelike to Bob cannot be possible, and hence we conclude that something in this example must be physically impossible.

Now, Alice's kick seems physically reasonable, as it is a local operation that cannot enable any superluminal signals on its own (see Example~\ref{ex:unitary kick causal}). Further, the assumption that Bob can measure their expectation value is similarly reasonable. In fact, such an expectation value of $e^{i\phi(g)}$ can even be recovered from weak measurements of the smeared field $\phi(g)$ as we discuss in Section~\ref{sec:Assumptions} (also see~\cite{Jubb_2022}), and such measurements were shown to be causal in Example~\ref{ex:weak measurement causal}. Thus, the most contentious aspect of this causality violating example is Charlie's operation described by the map $\mathcal{U}_{\phi(f)^2}$, and hence we conclude that it must be impossible to realise this map via any physical process.
\end{example}

\vspace{3mm}
In summary, when discussing update maps in QFT, the Einstein causality condition --- that spacelike operators commute --- is not enough to ensure consistency with relativistic causality. Update maps arising from any physical processes must further satisfy the requirement that they do not enable superluminal signals in Sorkin's scenario. 

In~\cite{Jubb_2022} this requirement is turned into a condition only on Charlie's map $\mathcal{E}_{A,\kappa}$ by requiring~\eqref{eq:desired causality condition for kraus maps} to be satisfied for every $K'\subset K^-$, every region $\spctm{R}\subseteq K^+\cap K'^{\perp}$, for every map $\mathcal{E}'$ local to $K'$, every $X\in\alg{B}(\spctm{R})$, and every state $\rho$. This condition is then shown to be equivalent to the more intuitive \textit{past-support non-increasing (PSNI)} property, which for the sake of brevity we will not discuss here.

One can ask, however, whether it physically makes sense to require~\eqref{eq:desired causality condition for kraus maps} to be satisfied for every state $\rho$. Should we, for example, consider states of the form $\rho = \ket{\psi}\bra{\psi}$ with $\ket{\psi} = e^{-i\phi(f)^2}\gs$, for some smeared field $\phi(f)$? Example~\ref{ex:acausal} highlights that unitarily kicking with an operator of the form $e^{-i\phi(f)^2}$ can lead to causality violations, and hence one might worry whether a state of the form $e^{-i\phi(f)^2}\gs$ can even be physically created from the ground state. In Section~\ref{sec:Causality of smeared field operations} we take such worries into account and restrict ourselves to the ground state only.

\section{Causality of Smeared Field Operations}\label{sec:Causality of smeared field operations}

The PSNI condition in~\cite{Jubb_2022} is a fairly general condition that ensures an update map cannot enable a superluminal signal. However, as previously discussed, the derivation of this PSNI condition assumes that~\eqref{eq:desired causality condition for kraus maps} is satisfied for all states $\rho$, and this assumption may be too strong given that some states may not be physically sensible.

Furthermore, for certain update maps, the PSNI condition is not immediately useful for determining their causal nature. One of the first update maps that comes to mind is that of an ideal measurement of some self-adjoint operator, and one of the simplest, if not \textit{the} simplest, operators one can consider is a smeared field. Even in this simple case it is not obvious whether the update map $\mathcal{E}_{\phi(f),\bor{R}}$ (for some resolution $\bor{R}$) is PSNI, and hence whether it enables a superluminal signal.

Preliminary calculations in~\cite{Jubb_2022} hint that the update map $\mathcal{E}_{\phi(f),\bor{R}}$ does in fact enable a superluminal signal, despite previous heuristic arguments to the contrary~\cite{Sorkin_impossible,Benincasa_2014,Borsten_2021}. In this Section we show conclusively that a superluminal signal generically arises in the specific case of the ground state $\rho=\gs \gsb$, and for \textit{any} choice of resolution $\bor{R}$.

There is an important loophole to our calculation, namely, whether a Sorkin scenario exists for the given smeared field $\phi(f)$ and the compact subset $K$ in which we measure it. In particular, a Sorkin scenario does not exist if $K$ is a \textit{transitive} subset of the spacetime $\spctm{S}$ --- a property we define in Section~\ref{sec:Transitive loophole and discrete spacetimes}. The question of whether a given $K$ is transitive seems to depend on the spatial topology of the spacetime, and could depend on the spacetime geometry more generally. This question also has important consequences for discrete spacetimes, e.g. causal sets, as we discuss in Section~\ref{sec:Transitive loophole and discrete spacetimes}.

Even with a non-transitive $K$, we still need to make some (fairly pedestrian) assumptions in our calculation, detailed in Section~\ref{sec:Assumptions}. These include an assumption about the form of the classical mode $\Delta f$ generated by the test function in $\phi(f)$, which seems generically true (we have not managed to concoct an example where it fails, and can concoct infinitely many examples where it is true), and at least in the simple case of a (massless and massive) real scalar field in $1+1$-Minkowski spacetime one can show that it is always true (see Appendix~\ref{app:Sorkin scenario results}).

Provided a Sorkin scenario is possible, we find the aforementioned superluminal signal by first deriving a general \textit{no-signalling condition} in Section~\ref{sec:General causality condition for a smeared field operations} for update maps constructed from a single smeared field $\phi(f)$, i.e. Kraus update maps of the form $\mathcal{E}_{\phi(f),\kappa}$, for some Kraus family $\kappa$. This includes the case of ideal measurements, as well as many more examples of update maps.

Our no-signalling condition is phrased in terms of a property of the Kraus family $\kappa$ associated to the map $\mathcal{E}_{\phi(f),\kappa}$. Not all Kraus families $\kappa$ have this property, and in particular, we show the update map for an ideal measurement of $\phi(f)$ (for any resolution $\bor{R}$) does not have this property in Section~\ref{sec:The acausality of an ideal measurement of a smeared field}. Since not all Kraus maps $\mathcal{E}_{\phi(f),\kappa}$ have this property, not all Kraus update maps are physically realisable in experiments; only those for which $\kappa$ satisfies this property have the potential to be realised in an experiment. This is in contrast to the situation in Quantum Information, where any update map, i.e. any quantum channel, is in principle physically realisable.

The fact that this no-signalling condition is expressed in terms of $\kappa$ means that, while it is not as general as the PSNI condition in~\cite{Jubb_2022} (as it only pertains to update maps of the Kraus form $\mathcal{E}_{\phi(f),\kappa}$), it is more readily useful than the PSNI condition for such update maps. Given a $\kappa$ one can immediately compute whether the given update map enables a superluminal signal or not, though it can still be challenging to determine whether all the families in some class enable a superluminal signal or not (see the ideal measurement case in Section~\ref{sec:The acausality of an ideal measurement of a smeared field}). This no-signalling condition is also, in principle, stronger than the PSNI condition for the case of Kraus update maps where it applies, as it only assumes the ground state $\rho=\gs \gsb$, instead of any state.

\subsection{Transitive Loophole and Discrete Spacetimes}\label{sec:Transitive loophole and discrete spacetimes}

\begin{definition}[Transitive]\label{def:transitive}
A subset $T\subseteq \spctm{S}$ is \textit{transitive} if, for all $x,y\in\spctm{S}\setminus T$,
\begin{equation}
    x\preceq z\preceq y \; \text{ for some }z\in T  \;\;\; \Longrightarrow \;\;\; x\preceq y  \; .
\end{equation}
\end{definition}

\noindent This property is important for our purposes as, if the compact subset where Charlie acts, $K\subset \spctm{S}$, is transitive, then no Sorkin scenario exists. That is, for any subset in which Alice acts, $K'\subset K^-$, and any region in which Bob measures his expectation value, $\spctm{R}\subseteq K^+$, if we try and concoct a Sorkin scenario by taking $K'$ to the past of $K$, and $K$ to the past of $\spctm{R}$, then by the transitive nature of $K$ we know that $K'$ is to the past of $\spctm{R}$, i.e. Alice and Bob are not spacelike. Thus, it is not possible to construct the Sorkin scenario illustrated in Fig.~\ref{fig:sorkin_scenario}. Let us consider some examples for clarity.

\begin{example}[Cylinder spacetime]\label{ex:Cylinder spacetime}
Consider the continuum spacetime $\spctm{S}=\spctm{M}=\mathbb{R}\times S^1$, with metric $\metric = -dt^2 +d\theta^2$, where $\theta\in [0,2\pi)$. Say Charlie blocks off the compact subset $K = [0,2\pi]\times [0,2\pi)$ for their measurement (Fig.~\ref{fig:cylinder_spacetime}). Now, any point $x$ strictly to the past of $K$ ($t<0$) is to the past of every point $y$ which is strictly to the future of $K$ ($t>2\pi$). Thus, $K$ is transitive.
\end{example}

\begin{figure}
     \centering
 \begin{subfigure}[b]{0.375\textwidth}
 \includegraphics[width=\textwidth]{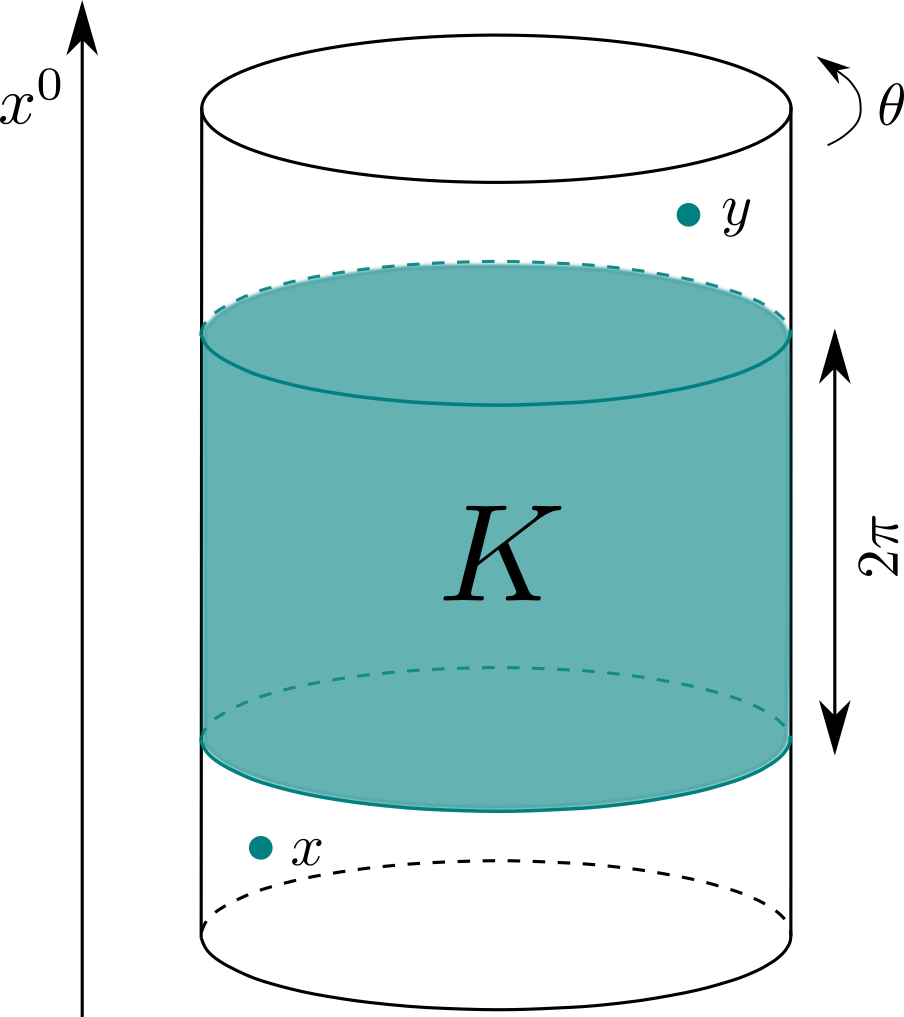}
 \caption{$K$ is transitive as it is long enough in time so that any point $y$ to its future is also to the future of any point $x$ to its past.}
 \label{fig:cylinder_spacetime}
 \end{subfigure}
  \hspace{15mm}
 \begin{subfigure}[b]{0.375\textwidth}
     \centering
 \includegraphics[width=\textwidth]{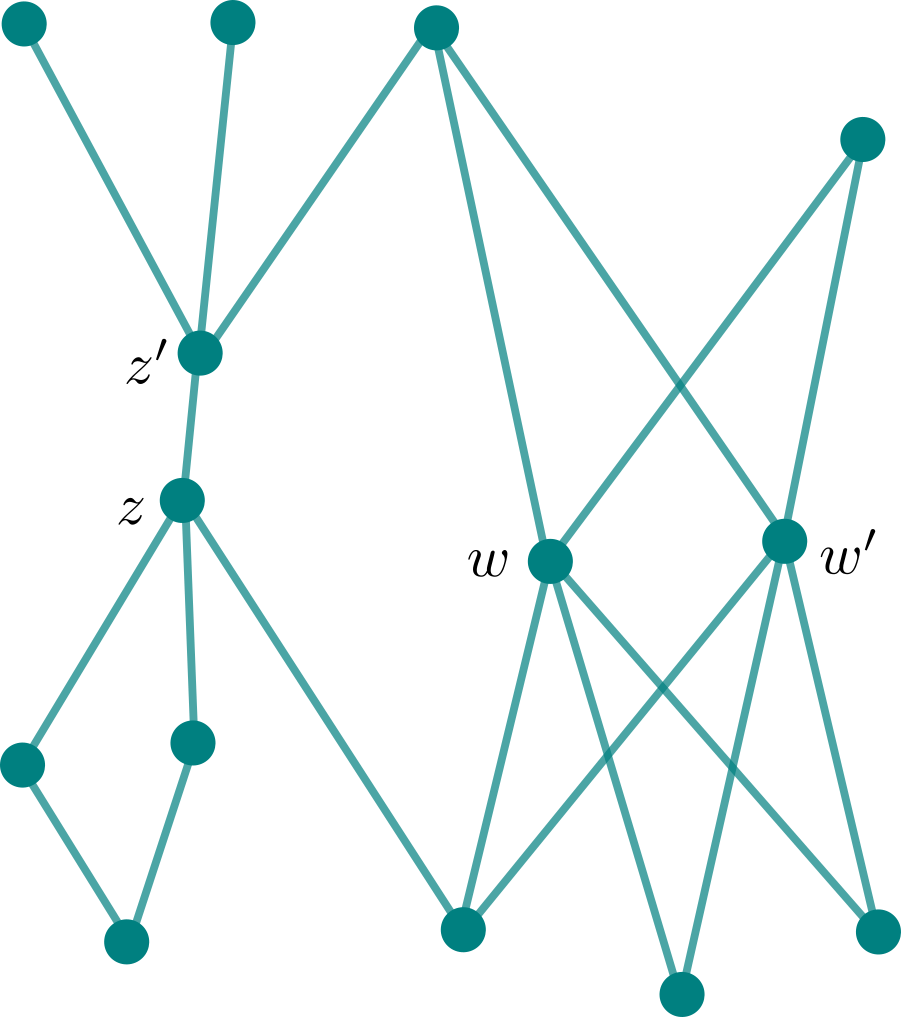}
 \caption{The subset $\lbrace z,z'\rbrace$ is transitive as any point in its future is to the future of any point in its past. Similarly, $\lbrace w,w'\rbrace$ is transitive.}
 \label{fig:transitive_causet}
 \end{subfigure}
\end{figure}

\begin{example}[Single causal set point]\label{ex:Single causal set element}
For any causal set $\spctm{S}=\spctm{C}$, if we take $K=\lbrace z \rbrace$ for some single causet point $z\in\spctm{C}$, then clearly for any points $x,y\neq z$ such that $x\preceq z$ and $z\preceq y$, we have $x\preceq y$. Thus, a single point is transitive. 
\end{example}

\begin{example}[Non-singleton subset of a causal set]\label{ex:Non-singleton subset of a causal set}
For a given causal set $\spctm{S}=\spctm{C}$ it may be possible to find subsets with more than a single point that are transitive. For example, consider Fig.~\ref{fig:transitive_causet}. The subset $\lbrace z,z'\rbrace$ is transitive, as $z\preceq z'$, $J^-(z')\setminus\lbrace z'\rbrace = J^-(z)$, and $J^+(z)\setminus\lbrace z\rbrace = J^+(z')$. The subset $\lbrace w ,w' \rbrace$ is also transitive as $w$ and $w'$ are spacelike and $J^{\pm}(w) = J^{\pm}(w')$. 
\end{example}

In all of the above examples, the fact that no Sorkin scenario exists means there is no restriction (coming from relativistic causality) on whether a given Kraus update map $\mathcal{E}_{A,\kappa}$ (for any self-adjoint $A$ localisable in $K$) can be implemented in a physical experiment in $K$. So, for example, it is in principle possible (with regards to causality) to do something to the system inside $K$ such that in $K^+$ the statistics of the system are changed in the way described by the ideal measurement update map $\mathcal{E}_{\phi(f),\bor{R}}$ for a smeared field $\phi(f)$, where $\supp f\subset K$.

In Example~\ref{ex:Cylinder spacetime} it makes physical sense why there is no issue with causality, as $K$ is large enough (in time) to allow light to travel around the whole circle, and thus any causality considerations are washed out. Example~\ref{ex:Cylinder spacetime} also highlights how this transitive property depends on the spatial topology, as, for example, any compact $K$ in Minkowski spacetime (with its non-compact spatial sections) would not be transitive. The geometry of the spacetime may also be important, as it could be the case that there exists a transitive $K$ in some spacetime with non-compact spatial sections, but with non-trivial geometry (we have not found any examples but have no proof that rules this out).

Example~\ref{ex:Single causal set element} highlights an important way in which discrete spacetimes, such as causal sets, can evade questions of causality \textit{a la} Sorkin. Specifically, for discrete spacetimes one can consider singleton sets $K$, and associated functions $f$ that have support only on that single point in the discrete spacetime. This is not possible in the continuum, as if $f(z)\neq 0$ at some given point $z\in\spctm{M}$, then the smoothness of $f$ implies it is non-zero in a neighbourhood of $z$.

\textit{Discrete spacetimes thus offer a particularly natural workaround to the issues raised by Sorkin in~\cite{Sorkin_impossible}. Rather than using the no-signalling condition in Section~\ref{sec:General causality condition for a smeared field operations}, or the PSNI condition in~\cite{Jubb_2022}, to carve out the space of physically viable update maps in QFT, one could instead impose that the space of viable update maps consists \textbf{only} of those that are constructed from operators local to a single discrete spacetime point. In fact, this would amount to all Kraus maps of the form $\mathcal{E}_{\phi(f),\kappa}$, for all Kraus families $\kappa$, and all $\phi(f)$ with $f$ supported on a single spacetime point.}

In such a framework, all physical operations would be composed of operations local to a single spacetime point. While this seems natural, it is worth pointing out that there of course causality respecting operations that are not necessarily local to a single spacetime point (e.g. Examples~\ref{ex:unitary kick causal} and~\ref{ex:weak measurement causal}) which this framework, therefore, cannot capture.

Finally, we note that Example~\ref{ex:Non-singleton subset of a causal set} is somewhat contrived, as the probability of a causal set sprinkling into Minkowski spacetime giving rise to a transitive set like $\lbrace z ,z'\rbrace$ is zero.

\subsection{Assumptions}\label{sec:Assumptions}

While the assumptions we make in this derivation are standard in any conventional treatment of QFT, one may be a little hesitant to assume anything at this stage, given the aforementioned issues with causality. Thus, it is important to be clear about the precise assumptions going into this calculation, and into the derivation of this no-signalling condition.

We consider some smeared field operation described by a Kraus update map, $\mathcal{E}_{\phi(f),\kappa}$, for some Kraus family $\kappa$, and some smeared field $\phi(f)$, where $f$ is supported in a compact subset $K\subset \spctm{S}$. We then make the following

\vspace{2mm}
\noindent \textbf{Assumptions}:
\vspace{-1mm}
\begin{enumerate}
    \item \textit{Sorkin's scenario exists}: $K$ is not transitive, and there exist test functions $h\in C^{\infty}_0 (K^-)$ and $g\in C^{\infty}_0 (K^+)$ such that $\Delta(f,h)\neq 0$ and $\Delta(f,g)\neq 0$, and such that $\supp h$ is spacelike to $\supp g$.
    \item \textit{Local Ground state}: We can start in the ground state $\rho = \gs \gsb$, or at least in a state $\rho'$ for which any measurable statistics in some given region containing our setup are indistinguishable from those of the ground state, i.e. $\text{tr}(\rho' X) = \gsb X \gs$ for any $X$ localisable in some given region $\tilde{\spctm{R}}\supset K\cup \supp h \cup\supp g$.
    \item \textit{Alice}: An independent agent Alice can unitarily kick with the smeared field $\phi(h)$.
    \item \textit{Bob}: An independent agent Bob can measure (in some causally consistent way) the smeared field, $\phi(g)$, and from their measurements recover the expectation value of $e^{i t\phi(g)}$ for some $t\in\mathbb{R}$.
\end{enumerate}

Given some non-transitive $K$, it seems generically true that for any test function $f$ supported in $K$ the rest of the conditions in (i) are true. In Appendix~\ref{app:Sorkin scenario results} we show that this is the case iff there exist two mutually spacelike points $x^{\pm}\in K^{\pm}$ which are also in the support of the classical solution $\Delta f$ generated by the test function $f$ (Claim~\ref{claim:Sorkin scenario equivalence 1}). One can concoct infinitely many examples of test functions $f$ (for a given non-transitive $K$) for which such a pair of points exist, and it seems impossible to concoct an example where no such pair exists, though we do not have a proof that this is impossible in general. In Appendix~\ref{app:Sorkin scenario results} we show that in the simple case of a (\textit{massless} in Claim~\ref{claim:massless KG field sorkin scenario always} and \textit{massive} in Claim~\ref{claim:massive KG field sorkin scenario always}) real scalar field in $1+1$ Minkowski spacetime, given any test function $f$, and any compact $K\supseteq \supp f$, it is always possible to find two such points.

In addition, in Claim~\ref{claim:sorkin lab K exists for any mode}, Appendix~\ref{app:Sorkin scenario results} we show that for any test function $f$ in any globally hyperbolic spacetime $\spctm{M}$, one can always find some equivalent test function $\tilde{f}$, i.e. it generates the same classical solution as $\Delta f = \Delta \tilde{f}$, and some $K = \supseteq \tilde{f}$ that is \textit{short enough in time} to permit a Sorkin scenario. This means that, for any smeared field operator, $\phi(f)$, there exists some compact subset $K$ in which there exists an equivalent $\tilde{f}$ (so that $\phi(f) = \phi(\tilde{f})$), and such that a Sorkin scenario exists. Since the existence of a Sorkin scenario means one has to worry about the corresponding update map $\mathcal{E}_{\phi(f),\bor{R}}$ violating causality, it seems reasonable to physically summarise this as follows: \textit{for any classical solution, $\varphi = \Delta f$, i.e. some spatially compact smooth wavepacket, there always exists a subset of spacetime $K$ which is large enough in space to encompass the wavepacket, but short enough in time to warrant concerns about whether an ideal measurement of the associated quantum operator, $\phi(f)$, and the data collection involved in such a measurement, can be performed in a causally consistent manner in $K$}. 

The local ground state assumption (ii) means that we can take $\rho=\gs \gsb$ in any calculations of expectation values of operators in our setup. Regarding Alice's assumption (iii), a unitary kick with a smeared field, i.e. the update map $\mathcal{U}_{\phi(f)}$, seems reasonable, as it amounts to a local operation which, by Example~\ref{ex:unitary kick causal}, cannot enable any superluminal signals on its own. This operation can also be seen as introducing a local and compactly supported linear coupling of the scalar field to a classical source term in the Lagrangian; something which is quite pedestrian in QFT.

Regarding Bob's assumption (iv), it is important to note that we are \textit{not} assuming it is possible to make an ideal measurement of $\phi(g)$, as this is one of the things we will be testing for $\phi(f)$ in what follows. Instead, we are only assuming that it is possible to measure $\phi(g)$ in \textit{some} way that is consistent with causality. This seems to be a reasonable assumption. Conversely, if one were not able to measure, in any way, the most basic observable of the theory, it is hard to imagine what the theory could be used for. 

To further justify (iv), we note that the weak measurement update map in Example~\ref{ex:weak measurement causal} (and in~\cite{Jubb_2022}) provides an example of a causally consistent update map for a measurement of a smeared field. In fact, Example~\ref{ex:weak measurement causal} tells us that there are infinitely many possibilities of such causally consistent candidates for update maps describing measurements; one for every normalisable $L^2$-Kraus update map. For these $L^2$-Kraus update maps one can further show that it is possible to recover the expectation value of $e^{it\phi(g)}$ for some $t\in\mathbb{R}$, by taking particular functions of the outcomes of a sample of measurements (Appendix~\ref{app:recovering tr rho exp from L2 kraus}).

\subsection{General Causality Condition for Smeared Field Operations}\label{sec:General causality condition for a smeared field operations}

Consider the setup and assumptions in Section~\ref{sec:Assumptions}. We want to test whether the Kraus update map $\mathcal{E}_{\phi(f),\kappa}$ enables a superluminal signal from Alice to Bob (Fig.~\ref{fig:sorkin_scenario}). 

We start in the ground state $\rho = \gs\gsb$. Alice unitarily kicks with $\phi(h)$ in $K'\subset K^-$, which is described by the map $\mathcal{U}_{s\phi(h)}$ (for $s\in\mathbb{R}$). Charlie implements the Kraus operation $\mathcal{E}_{\phi(f),\kappa}$ in $K\supseteq \supp f$, and Bob uses their measurement of $\phi(g)$ in the region $\spctm{R}\subset K^+\cap K'^{\perp}$ to recover the expectation value of $e^{it\phi(g)}$ (for $t\in\mathbb{R}$). We let
\begin{align}\label{eq:chi s}
    \chi (s) & = \gsb \mathcal{U}_{s\phi(h)}( \mathcal{E}_{\phi(f),\kappa}(e^{it\phi(g)}) )\gs \nonumber
    \\
    & =  \gsb e^{is\phi(h)} \mathcal{E}_{\phi(f),\kappa}(e^{it\phi(g)}) e^{-is\phi(h)}\gs \; ,
\end{align}
denote this value for Bob as a function of $s$, which can be thought of as the strength of Alice's kick. Recall that this order of composition --- $\mathcal{E}_{\phi(f),\kappa}$ first then $\mathcal{U}_{s\phi(h)}$ --- follows as $K'\subset K^-$, and since we are in the dual picture where the maps act on the operators instead of the state. From our assumptions in Section~\ref{sec:Assumptions} we have that $[\phi(f) , \phi(g)] = i\Delta(f,g)\neq 0$ and $[\phi(f) , \phi(h)] = i\Delta(f,h)\neq 0$, and that $[\phi(g) , \phi(h)] = i\Delta(g,h)=0$.

Now, for no superluminal signal to be possible in this setup, Bob's expectation value must be independent of Alice's kick strength $s$. That is, $\chi(s)$ must be constant in $s$ for the Kraus update map, $\mathcal{E}_{\phi(f),\kappa}$, to not enable a superluminal signal in our setup. We now make the following claim that relates this statement about causality to the Kraus family $\kappa$:

\begin{claim}[Causality condition on the Kraus family]\label{claim:causality condition on kappa tilde}
The Kraus update map, $\mathcal{E}_{\phi(f),\kappa}$, will not enable a superluminal signal in our setup iff the Kraus family, $\kappa$, has the property that the function $\tilde{\kappa}:\mathbb{R}\times \mathbb{R}\rightarrow \mathbb{C}$, defined as
\begin{equation}\label{def:kappa tilde}
    \tilde{\kappa}(\lambda , t) = \int_{\Gamma}d\nu(\gamma ) \kappa(\lambda,\gamma)\kappa(\lambda+t\Delta(f,g),\gamma)^* \; ,
\end{equation}
is constant almost everywhere in $\lambda$ for any $t\in\mathbb{R}$.
\end{claim}

\begin{proof}
Focussing on $\mathcal{E}_{\phi(f),\kappa}(e^{it\phi(g)})$ we have
\begin{align}
    \mathcal{E}_{\phi(f),\kappa}(e^{it\phi(g)}) & = \int_{\Gamma}d\nu(\gamma ) \kappa(\phi(f),\gamma)e^{it\phi(g)}\kappa(\phi(f),\gamma)^{\dagger} \nonumber
    \\
    & = \int_{\Gamma}d\nu(\gamma ) \kappa(\phi(f),\gamma)e^{it\phi(g)}\kappa(\phi(f),\gamma)^{\dagger}e^{-it\phi(g)}e^{it\phi(g)} \nonumber
    \\
    & = \int_{\Gamma}d\nu(\gamma ) \kappa(\phi(f),\gamma)\kappa(\phi(f)+t\Delta(f,g),\gamma)^{\dagger}e^{it\phi(g)} \nonumber
    \\
    & = \tilde{\kappa}(\phi(f),t)e^{it\phi(g)} \; ,
\end{align}
where we have inserted $\mathds{1} = e^{-it\phi(g)}e^{it\phi(g)}$ in line 2, used~\eqref{eq:unitary_kick_on_zeta_phi_f} in line 3, and in the last line we have used the definition of $\tilde{\kappa}$ in~\eqref{def:kappa tilde}. Inserting this back into~\eqref{eq:chi s} we find
\begin{align}
    \chi (s) & = \gsb e^{is\phi(h)} \mathcal{E}_{\phi(f),\kappa}(e^{it\phi(g)}) e^{-is\phi(h)}\gs \nonumber
    \\
    & = \gsb e^{is\phi(h)} \tilde{\kappa}(\phi(f),t)e^{it\phi(g)} e^{-is\phi(h)}\gs \nonumber
    \\
    & = \gsb e^{is\phi(h)} \tilde{\kappa}(\phi(f),t) e^{-is\phi(h)} e^{it\phi(g)}\gs \nonumber
    \\
    & = \gsb  \tilde{\kappa}(\phi(f)+s\Delta(f,h),t)  e^{it\phi(g)}\gs \; ,
\end{align}
using the fact that $[\phi(g),\phi(h)]=0$ in line 3 to commute $e^{it\phi(g)}$ and $e^{-is\phi(h)}$, and using~\eqref{eq:unitary_kick_on_zeta_phi_f} in the last line.

The last line in the previous equation is now of the form~\eqref{eq:weierstrass_inner_product}, and hence we can write this expectation value as the following Weierstrass transform:
\begin{equation}
    \chi (s) = e^{-t^2\frac{W(g,g)}{2}} \fwt{W}\Bigg\lbrace \tilde{\kappa}\Bigg(\sqrt{\frac{W(f,f)}{2}} (\cdot ) + s\Delta(f,h) ,t \Bigg) \Bigg\rbrace ( z) \; ,
\end{equation}
where $z = i t\sqrt{\frac{2}{W(f,f)}} W(f,g)$. We can then use~\eqref{eq:weierstrass_shift} to move the shift by $s\Delta(f,h)$ from inside the transform to the argument. We then have
\begin{equation}\label{eq:chi to weierstrass transform final}
    \chi(s) = e^{-t^2\frac{W(g,g)}{2}} \fwt{W}\lbrace \tilde{\kappa}( a (\cdot ) ,t ) \rbrace ( a^{-1}\left( s\Delta(f,h) + i t W(f,g) \right)) \; ,
\end{equation}
where $a = \sqrt{\frac{W(f,f)}{2}}$.

Clearly if $\tilde{\kappa}(\lambda , t)$ is constant almost everywhere in $\lambda$, then under the integral in  $\fwt{W}\lbrace \tilde{\kappa}( a (\cdot ) ,t )\rbrace$ we can treat $\tilde{\kappa}(\cdot , t)$ as a constant function. Now, for any constant $c\in\mathbb{C}$ we have $\fwt{W}\lbrace c \rbrace (z) = c$. Thus, an almost everywhere constant $\tilde{\kappa}(\cdot , t)$ implies a constant $\chi(\cdot )$, and hence no superluminal signal. This takes care of the ``if'' part of the claim.

For the ``only if'' part of the claim we need to show that a constant $\chi(\cdot )$ implies an almost everywhere constant $\tilde{\kappa}(\cdot , t)$ for any $t\in\mathbb{R}$. By~\eqref{eq:chi to weierstrass transform final}, a constant $\chi(\cdot )$ implies the Weierstrass transform of $\tilde{\kappa}(\cdot , t)$ is constant along the real axis of the complex plane. Note that $\tilde{\kappa}(\cdot , t)$ is a bounded function. One can see this by first noting that $\kappa (\lambda , \cdot )$ is in $L^2 (\Gamma , \nu)$ (i.e. it is square integrable with the measure $\nu$ on $\Gamma$) for any $\lambda\in\mathbb{R}$, which follows from the normalisation condition~\eqref{eq:kappa normalisation condition}. Then, we note that $\tilde{\kappa}(\lambda , t)$ is simply the $L^2(\Gamma , \nu)$ inner product of $\kappa(\lambda , \cdot)$ and $\kappa(\lambda +t , \cdot)$ in $L^2(\Gamma , \nu)$. This inner product is finite, as by the Cauchy-Schwarz inequality we know it is bounded in absolute value by the $L^2(\Gamma , \nu)$ norm of $\kappa (\lambda , \cdot )$ (which is $1$ by~\eqref{eq:kappa normalisation condition}).

In Appendix~\ref{app:Constant Weierstrass transform implies constant input} we show that if $\fwt{W}\lbrace \zeta (\cdot , t) \rbrace (s+it)$ is constant in $s$, for some bounded (and measurable) $\zeta:\mathbb{R}\times \mathbb{R}\rightarrow\mathbb{C}$, then $\zeta(\cdot  , t)$ is an almost everywhere constant function. Thus, $\tilde{\kappa}(\cdot , t)$ is almost everywhere constant as desired. 

\end{proof}

Let us now consider some examples to see how this condition on the Kraus family can be quickly evaluated.

\begin{example}[Unitary kicks]
Recall Example~\ref{ex:unitary kick kraus} where the Kraus update map was $\mathcal{U}_{\phi(f)}$ for some smeared field $\phi(f)$. There we saw that $\kappa(\lambda , \gamma) = e^{i\lambda}$, and that $\Gamma = \lbrace \gamma\rbrace$. Thus,
\begin{align}
    \tilde{\kappa}(\lambda , t) & = \kappa(\lambda , \gamma)\kappa(\lambda + t \Delta(f,g)\gamma)^* \nonumber
    \\
    & = e^{i\lambda}e^{-i(\lambda+ t \Delta(f,g))} \nonumber
    \\
    & = e^{-it \Delta(f,g)} \; .
\end{align}
This is constant in $\lambda$, and hence by Claim~\ref{claim:causality condition on kappa tilde} we know that $\mathcal{U}_{\phi(f)}$ does not enable a superluminal signal, as shown by other methods in Example~\ref{ex:unitary kick causal}.

Consider, instead, the map $\mathcal{U}_{\phi(f)^2}$ from Example~\ref{ex:acausal}. The only difference to the previous case is that $\kappa(\lambda , \gamma) = e^{i\lambda^2}$. We then get
\begin{align}
    \tilde{\kappa}(\lambda , t) & = e^{i\lambda^2}e^{-i(\lambda+ t \Delta(f,g))^2} \nonumber
    \\
    & = e^{-it \Delta(f,g)(2\lambda + t\Delta(f,g))} \; .
\end{align}
Since this depends on $\lambda$, Claim~\ref{claim:causality condition on kappa tilde} tells us that the map $\mathcal{U}_{\phi(f)^2}$ enables a superluminal signal, which was shown via other means in Example~\ref{ex:acausal}.
\end{example}

\begin{example}[$L^2$-Kraus update map for a smeared field]\label{ex:L2-Kraus update map for a smeared field}
Recall Example~\ref{ex:L2 kraus} where the Kraus update map was $\mathcal{E}_{\phi(f),\kappa}$ for some smeared field $\phi(f)$, where $\Gamma = \mathbb{R}$ with $\nu$ the Lebesgue measure, and where $\kappa(\lambda , \gamma) = k(\lambda - \gamma)$ for some normalised $L^2$ function $k:\mathbb{R}\rightarrow\mathbb{C}$. Thus,
\begin{align}
    \tilde{\kappa}(\lambda , t) & = \int_{\mathbb{R}}d\gamma k(\lambda - \gamma)k(\lambda +t\Delta(f,g)-\gamma)^* \nonumber
    \\
    & = \int_{\mathbb{R}}d\gamma' k(\gamma')k(\gamma' +t\Delta(f,g))^* \; ,
\end{align}
where we have changed integration variables to $\gamma' = \lambda - \gamma$ in the last line. The final integral evaluates to some function of $t$ only, and hence is constant in $\lambda$. Thus, by Claim~\ref{claim:causality condition on kappa tilde} we know this update map does not enable a superluminal signal, which was shown via other methods in Example~\ref{ex:weak measurement causal}.
\end{example}

These examples illustrate the utility of Claim~\ref{claim:causality condition on kappa tilde} in quickly determining the causal nature of a given Kraus update map. In the next section we focus on the case of ideal measurements.

\subsection{The Acausality of an Ideal measurement of a Smeared Field}\label{sec:The acausality of an ideal measurement of a smeared field}

In Example~\ref{ex:ideal measurement} we discussed how ideal measurements can be expressed as a Kraus update map. Given some resolution $\bor{R}=\lbrace \bor{B}_n \rbrace_{n\in I}$, for some countable indexing set $I$, the Kraus family $\kappa$ for the ideal measurement update map $\mathcal{E}_{\phi(f),\bor{R}}$ can be described as follows.

We take $\Gamma = I$, with $\nu$ the counting measure, and set $\kappa(\lambda , n) = 1_{\bor{B}_n}(\lambda)$ for $\lambda\in\mathbb{R}$ and $n\in I$. In this case, we find
\begin{align}\label{eq:kappa tilde for ideal measurement}
    \tilde{\kappa}(\lambda , t) & = \sum_{n\in I} = 1_{\bor{B}_n}(\lambda)1_{\bor{B}_n}(\lambda+t\Delta(f,g)) \nonumber
    \\
    & = \sum_{n\in I}1_{\bor{B}_n}(\lambda)1_{\bor{B}_n-t\Delta(f,g)}(\lambda) \nonumber
    \\
    & = \sum_{n\in I}1_{\bor{B}_n\cap ( \bor{B}_n-t\Delta(f,g) )}(\lambda) \nonumber
    \\
    & = 1_{\bor{R}_{t\Delta(g,f)}}(\lambda) \; ,
\end{align}
where in the last line we have defined, for each $t\in \mathbb{R}$, the following subset of the real line:
\begin{equation}\label{eq:Rt definition}
    \bor{R}_t = \cup_{n\in I}\bor{B}_n\cap ( \bor{B}_n+t ) \; .
\end{equation}
Note, for $t=0$ we get 
\begin{equation}
    \bor{R}_0 = \cup_{n\in I}\bor{B}_n = \mathbb{R} \; .
\end{equation}
Equation~\eqref{eq:kappa tilde for ideal measurement} tells us that, for the case of an ideal measurement of $\phi(f)$, $\tilde{\kappa}(\lambda ,t)$ is simply an indicator function for the set $\bor{R}_t \subseteq \mathbb{R}$.

By Claim~\ref{claim:causality condition on kappa tilde}, such an ideal measurement will not enable a superluminal signal in our setup for any $t\in\mathbb{R}$ if $1_{\bor{R}_t}(\lambda)$ is constant almost everywhere in $\lambda$. As $1_{\bor{R}_t}(\lambda)$ takes values $0$ or $1$, this can only happen if, for each value of $t\in\mathbb{R}$, either $\bor{R}_t \approx \varnothing$ (hence $1_{\bor{R}_t}(\lambda) = 0$ almost everywhere) or $\bor{R}_t \approx \mathbb{R}$ (hence $1_{\bor{R}_t}(\lambda) = 1$ almost everywhere). Here we have used ``$\approx$'' to denote equivalence almost everywhere. That is, the two sets differ on a set of Lebesgue measure zero. 

For a given resolution, $\bor{R}$, it can be straightforward to check whether this is the case, as the following example illustrates.
\begin{example}[Integer bin resolution]\label{ex:Integer bin resolution}
Consider the resolution with bins $\bor{B}_n = [n,n+1)$, for $n\in\mathbb{Z}$. We then have $\bor{B}_n \cap ( \bor{B}_n + t ) = [n+t,n+1) $ for $0\leq t < 1$, and the empty set otherwise. Thus, $\bor{R}_t = \cup_{n\in \mathbb{Z}}[n+t,n+1)$ for $0\leq t < 1$, and the empty set otherwise. For $0 <t <1$ the set $\bor{R}_t$ is clearly not equivalent to $\mathbb{R}$ or $\varnothing$ almost everywhere. Therefore, such ideal measurements enable superluminal signals in our setup, and hence it is physically impossible to measure $\phi(f)$ in such a way that amounts to an ideal measurement update map with this integer bin resolution.
\end{example}

What is not necessarily straightforward to see is whether ideal measurements of the form $\mathcal{E}_{\phi(f),\bor{R}}$ enable superluminal signals for \textit{every} resolution $\bor{R}$. That is, is there some resolution $\bor{R}$ for which an ideal measurement is physically realisable?

We do not consider the trivial resolution containing a single bin, i.e. $\bor{R} = \lbrace \mathbb{R} \rbrace$, as this corresponds to no measurement at all. We also do not consider bins of zero measure, as we can never obtain a measurement outcome in such a zero-measure set (the associated projectors also kill any state). In summary, we only consider resolutions $\bor{R}= \lbrace \bor{B}_n \rbrace_{n\in I}$ for which each $\bor{B}_n$ is of non-zero Lebesgue measure, and such that there are at least two bins. 

At first glance, it may seem likely that every such resolution enables a superluminal signal, given Example~\ref{ex:Integer bin resolution}. On the other hand, the space of possible resolutions is vast, and one can consider extremely complicated cases. For example, one can consider bins, $\bor{B}_n$, that are Cantor sets of non-zero measure (Smith–Volterra–Cantor sets). In such a case what does $\bor{B}_n \cap (\bor{B}_n + t)$, or more importantly $\bor{R}_t$, look like?

In Appendix~\ref{app:Non-triviality of Borel set self-intersection} we show the technical result that for any such resolution, $\bor{R}$, there exists some $t\in\mathbb{R}$ for which $\bor{R}_t$ is not equivalent to $\mathbb{R}$ or $\varnothing$. That is, for every `way' of making an ideal measurement of the smeared field $\phi(f)$ --- every non-trivial choice of bins --- the act of measurement enables a superluminal signal. \textit{Thus, given our assumptions in Section~\ref{sec:Assumptions}, ideal measurements of smeared fields are not physically realisable}. Recall from Section~\ref{sec:Transitive loophole and discrete spacetimes} that there are loopholes to our assumptions in Section~\ref{sec:Assumptions}. For instance, in a discrete spacetime such as a causal set, one can consider ideal measurements of smeared fields local to a single discrete spacetime point without any causal repercussions. 

\textit{However, for continuum spacetimes in which our assumptions in Section~\ref{sec:Assumptions} are always satisfied, e.g. $1+1$ Minkowski spacetime (and potentially higher dimensions too), this means that ideal measurements of smeared field operators are not possible at all. The projection postulate from which the ideal measurement update arises is simply not the correct way to describe a measurement of the field, despite the fact that the update map is perfectly local.} 

The projection postulate nevertheless provides a useful description of measurements in non-relativistic quantum mechanics, in the sense that it yields predictions that accord with experiments. Thus, even if the projection postulate is not the correct description in relativistic QFT, it must be the case that the projection postulate arises as an effective description of some causality respecting update map in the underlying relativistic QFT. Example~\ref{ex:L2-Kraus update map for a smeared field} already provides an infinite family of candidate update maps in QFT which are local and respect causality. It remains to be seen, however, whether any member from this family yields the projection postulate in a non-relativistic limit. We leave this task --- and the general question of recovering the projection postulate from a causal update map in QFT --- for a future project.

\section{Revisiting Previous Literature with an Illustrative Example}\label{sec:Revisiting previous results with an illustrative example}

The result that an ideal measurement of a smeared field generically enables a superluminal signal may seem surprising when contrasted with the heuristic arguments presented in~\cite{Sorkin_impossible,Benincasa_2014,Borsten_2021} that it should in fact not enable a signal.

Those heuristic arguments first note that a smeared field can be written as a sum (or more technically an integral) of the field operator-valued distribution on some Cauchy surface, $\phi(x)$, and its conjugate variable, $\pi(x)$ (see~\cite{Benincasa_2014}, equation (15), for details). Such a `sum' of local operators (or operator-valued distributions more accurately) looks like a `continuous version' of a sum of local operators on a multipartite Hilbert space, e.g. an operator of the form $O = O_A\otimes \mathds{1}_B + \mathds{1}_A \otimes O_B$ acting on some bipartite Hilbert space $\alg{H} = \alg{H}_A \otimes \alg{H}_B$. The intuition, then, is that an ideal measurement of a smeared field will exhibit the same features as an ideal measurement of an operator such as $O$, and in~\cite{Borsten_2021} the latter was shown to not enable a signal between agents acting on the Hilbert spaces $\alg{H}_A$ and $\alg{H}_B$ respectively. Thus, an ideal measurement of a smeared field was not expected to signal.

Following the previous sections, we now know that an ideal measurement of a smeared field does in fact enable a superluminal signal (generically), and so it must evade the above argument in some way. \textit{In short, the result in~\cite{Borsten_2021} that an ideal measurement of $O$ does not signal assumes that both $O_A$ and $O_B$ have pure point spectra. A smeared field does not meet this requirement, and thus it evades the above argument.}

\subsection{An Illustrative Example}\label{sec:An illustrative example}

We can demonstrate the importance of pure point spectra with a simple example. Consider a 2D quantum harmonic oscillator. The Hilbert space consists of complex-valued square integrable functions on $\mathbb{R}^2$: $\alg{H} = L^2(\mathbb{R}^2;\mathbb{C})$. We can write this as $\alg{H}=\alg{H}_A\otimes\alg{H}_B$, where $\alg{H}_A = L^2(\mathbb{R};\mathbb{C})$ is the space of square integrable functions in the $x$-axis, and $\alg{H}_B = L^2(\mathbb{R};\mathbb{C})$ is the space of square integrable functions in the $y$-axis. We assume the system starts in the state $\psi(x,y) = \pi^{-1/2} e^{-(x^2+y^2)/2}$, which we write as $\ket{\psi}$ in bra-ket notation.

Note that in this example the two Hilbert spaces --- $\alg{H}_A$ and $\alg{H}_B$ --- do not correspond to distinct physical systems separated in space. This is in contrast to, say, two spatially separated qubits, as was considered in~\cite{Borsten_2021,Beckman_2001}. It may seem strange, then, to think of a ``signal'' between two agents Alice and Bob acting on the Hilbert spaces $\alg{H}_A$ and $\alg{H}_B$ respectively, as they are both acting on the same physical system --- the 2D oscillator. 
The analogous signalling protocol in this example can be described as follows. Alice first performs some unitary operation with respect to the $x$-axis degrees of freedom of the 2D oscillator. A third agent, Charlie, then performs an ideal measurement of some observable of the 2D oscillator, described by some operator of the form of $O$ above. Finally, Bob measures some operator associated with the $y$-axis degrees of freedom of the 2D oscillator. We can then ask whether Charlie's measurement enables Bob to tell whether Alice performed her operation or not, starting in the state $\ket{\psi}$. If it does, then we say that Alice can signal Bob. Note that if Charlie does nothing then it is easy to see that Bob cannot tell whether Alice has performed her operation, since they act on independent parts of the Hilbert space.

Before moving on, we note that one can equivalently write the $x$ and $y$ position operators, $\hat{x}$ and $\hat{y}$, as $\hat{x}\otimes \mathds{1}_B$ and $\mathds{1}_A\otimes\hat{y}$ respectively. This helps highlight their trivial action on $\alg{H}_A$ and $\alg{H}_B$ respectively. Similarly, we can write the associated momentum operators as $\hat{p}_x$ and $\hat{p}_y$, or as $\hat{p}_x\otimes \mathds{1}_B$ and $\mathds{1}_A\otimes\hat{p}_y$.

Now, consider any self-adjoint operator of the form $O = O_A\otimes \mathds{1}_B + \mathds{1}_A \otimes O_B$, where $O_{A,B}$ has a pure point spectrum, i.e. we can write $O_{A,B}$ as a sum of eigenprojectors (not necessarily of rank 1, or even finite rank) weighted by their respective eigenvalues.

For example, we could take $O_A = {\hat{p}_x}^2 + \hat{x}^2$, and $O_B = {\hat{p}_y}^2 + \hat{y}^2$. $O_A$ then has the eigenstates (labelled by $n=0,1,2...$) $\psi_n(x) = (2^n n!)^{-1/2}\pi^{-1/4}H_n(x)e^{-x^2/2}$ (where $H_n$ is the $n$'th Hermite polynomial), and the associated eigenvalues $2n+1$. $O_B$ similarly has eigenstates $\psi_n(y)$ and eigenvalues $2n+1$. Up to a constant, we can think of an ideal measurement of $O = O_A\otimes \mathds{1}_B + \mathds{1}_A \otimes O_B$ as a measurement of the sum of the energies of our 2D oscillator in the $x$- and $y$-axes.

The result in~\cite{Borsten_2021} tells us that an ideal measurement (using the `canonical' resolution described in Section~\ref{sec:Interpreting resolution}) of any $O$ of this form cannot enable a signal between Alice and Bob, for any starting state. In particular, if Alice kicks the state $\ket{\psi}$ with the unitary operator $U_A = e^{-is\hat{p}_x}$ (for some $s\in\mathbb{R}$), and if Bob uses a variety of measurements to extract the expected value of the operator $e^{i t \hat{p}_y}$ (for some $t\in\mathbb{R}$) after Charlie's measurement of $O$, then Bob's expected value does not depend on Alice's kick strength $s$, i.e. Bob cannot tell if Alice has kicked or not.

Now consider the operator $O = \hat{x}\otimes \mathds{1}_B + \mathds{1}_A\otimes \hat{y}$, which can be equivalently written as $O = \hat{x}+\hat{y}$. Both $\hat{x}$ and $\hat{y}$ have continuous spectra. Thus, to describe an ideal measurement of $O = \hat{x}+\hat{y}$ we now need to pick some resolution $\bor{R}=\lbrace \bor{B}_n\rbrace_{n\in I}$, as the lack of a pure point spectrum means there is no obvious `canonical' resolution (c.f. Section~\ref{sec:Interpreting resolution}). Such an ideal measurement with resolution $\bor{R}$ amounts to measuring the position of the particle in the $(x+y)$-direction and binning the result into the bins $\bor{B}_n$. $O = \hat{x}+\hat{y}$ is still a sum of `local' observables (in the sense that $\hat{x}$ and $\hat{y}$ are local to their respective Hilbert spaces $\alg{H}_A$ and $\alg{H}_B$), but the associated operators, $\hat{x}$ and $\hat{y}$, no longer have pure point spectra. Note, the projector for a given bin $\bor{B}_n$ is simply the indicator function $1_{\bor{B}_n}(x+y)$.

After Alice's kick with $U_A = e^{-is\hat{p}_x}$ and Charlies ideal measurement of $O = \hat{x}+\hat{y}$ with some resolution $\bor{R}$, Bob's expectation value of $e^{i t \hat{p}_y}$ is explicitly $\chi(s) = \bra{\psi}e^{is\hat{p}_x}\mathcal{E}_{O,\bor{R}}(e^{i t \hat{p}_y})e^{-i s \hat{p}_x}\ket{\psi}$, c.f. equation~\eqref{eq:chi s}. Alice's kick with $e^{-i s \hat{p}_x}$ acts as a shift operator in the $x$-direction, and similarly Bob's operator $e^{i t \hat{p}_y}$ acts as a shift in the $y$-direction. We then have $e^{-i s \hat{p}_x} \psi(x,y) = \psi(x-s,y)$. Furthermore, for any operator $T$ we have
\begin{equation}
    \mathcal{E}_{O,\bor{R}}(T) = \sum_{n\in I} 1_{\bor{B}_n}(x+y) T 1_{\bor{B}_n}(x+y) \; .
\end{equation}
Putting all of this together, we then get
\begin{align}
    \chi(s) & = \int_{\mathbb{R}^2}dx dy \, \psi(x-s,y)^* \sum_{n\in I}  1_{\bor{B}_n}(x+y) e^{i t \hat{p}_y} \left( 1_{\bor{B}_n}(x+y) \psi(x-s,y) \right) \nonumber
    \\
    & = \int_{\mathbb{R}^2}dx dy \, \psi(x-s,y)^* \sum_{n\in I}  1_{\bor{B}_n}(x+y) 1_{\bor{B}_n}(x+t+y) \psi(x-s,y+t) \; ,
\end{align}
and
\begin{align}
    \sum_{n\in I}  1_{\bor{B}_n}(x+y) 1_{\bor{B}_n}(x+t+y) & = \sum_{n\in I}  1_{\bor{B}_n}(x+y) 1_{\bor{B}_n-t}(x+y) \nonumber
    \\
    & = \sum_{n\in I}  1_{\bor{B}_n\cap (\bor{B}_n - t)}(x+y) \nonumber
    \\
    & = 1_{\bor{R}_{-t}}(x+y) \; ,
\end{align}
where we recall the definition, $\bor{R}_t = \cup_{n\in I}\bor{B}_n \cap (\bor{B}_n + t)$, from~\eqref{eq:Rt definition}. Thus,
\begin{equation}\label{eq:final x y chi s integral}
    \chi(s) = \int_{\mathbb{R}^2}dx dy \, 1_{\bor{R}_{-t}}(x+y) \psi(x-s,y)^* \psi(x-s,y+t)  \; .
\end{equation}
The indicator function $1_{\bor{R}_{-t}}(x+y)$ in the integrand means that, to compute $\chi(s)$, one simply restricts the integral of $\psi(x-s,y)^* \psi(x-s,y+t)$ from the whole $\mathbb{R}^2$ plane to the subset of points $S_{-t}=\lbrace (x,y)\in\mathbb{R}^2 | x+y\in\bor{R}_{-t}\rbrace$. It may help to visualise the subset $S_{-t}$ as a set of infinitely long `stripes' at right angles to the $(x+y)$-direction, and with this in mind it seems fairly intuitive that, for some generic resolution $\bor{R}$, the integral $\chi(s)$, will depend on $s$, i.e. how much $\psi(x-s,y)^* \psi(x-s,y+t)$ is shifted in the $x$-direction.

In fact, we can show this is true for any resolution $\bor{R}$ by following the calculations in Section~\ref{sec:General causality condition for a smeared field operations} and Appendices~\ref{app:Constant Weierstrass transform implies constant input} and~\ref{app:Non-triviality of Borel set self-intersection}. We first recall that $\psi(x,y) = \pi^{-1/2}e^{-(x^2+y^2)/2}$. Next, the fact that the indicator function in~\eqref{eq:final x y chi s integral} depends only on $x+y$ suggests we change variables to $u=(x-y)/2$ and $v=x+y$. With this we get
\begin{equation}
    \chi(s) = \frac{e^{-s^2 - \frac{t^2}{2}}}{\pi}\int_{\mathbb{R}}dv \, e^{-\frac{v^2}{2}-\left(\frac{t}{2}-s\right)v}1_{\bor{R}_{-t}}(v)\int_{\mathbb{R}}du \, e^{-2u^2 + (2s+t)u} \; .
\end{equation}
The $u$ integral then evaluates to $\sqrt{\frac{\pi}{2}}e^{(2s+t)^2/8}$, and after some manipulation we are left with the $v$ integral
\begin{equation}
    \chi(s) = \frac{e^{-\frac{t^2}{4}}}{\sqrt{2\pi}}\int_{\mathbb{R}}dv \, e^{-\frac{1}{2}\left(v-\left(s-\frac{t}{2}\right)\right))^2} 1_{\bor{R}_{-t}}(v) \; .
\end{equation}
By changing variables to $w = \sqrt{2}v$ we can recast this expression as a Weierstrass transform:
\begin{equation}
    \chi(s) = e^{-\frac{t^2}{4}} \fwt{W}\left\lbrace 1_{\bor{R}_{-t}}\left(\frac{\cdot}{\sqrt{2}} \right) \right\rbrace (z) \; ,
\end{equation}
where $z=\sqrt{2}s - \frac{t}{\sqrt{2}}$.

Finally, following the arguments in Section~\ref{sec:General causality condition for a smeared field operations} and Appendices~\ref{app:Constant Weierstrass transform implies constant input} and~\ref{app:Non-triviality of Borel set self-intersection}, we know that for $\chi(s)$ to be constant in $s$, it must be the case that $1_{\bor{R}_{-t}}(\cdot )$ is a constant function almost everywhere. From Appendices~\ref{app:Constant Weierstrass transform implies constant input} and~\ref{app:Non-triviality of Borel set self-intersection} we know that this cannot be the case for any resolution $\bor{R}$, and hence $\chi(s)$ has a non-trivial dependence on $s$, i.e. Alice can signal Bob.

\subsection{Relation to the QFT Scenario}\label{sec:Relation to the QFT scenario}

This calculation is clearly very similar to the corresponding calculation for an ideal measurement of a smeared field, and in fact, they can be mapped into one another. Explicitly, one maps $\hat{x}\mapsto\phi(f_1)$ and $\hat{y}\mapsto\phi(f_2)$, where $f_1$ and $f_2$ are spacelike supported test functions. Thus, $O = \hat{x}+\hat{y}\mapsto \phi(f_1)+\phi(f_2) = \phi(f_1+f_2) = \phi(f)$, where we have set $f=f_1+f_2$. Next, one maps $\hat{p}_x\mapsto \phi(h)$ and $\hat{p}_y\mapsto \phi(g)$, where $h$ has support in the past of $f_1$ (with $\Delta(h,f_1)\neq 0$) and spacelike to $f_2$, $g$ has support in the future of $f_2$ (with $\Delta(g,f_2)\neq 0$) and spacelike to $f_1$, and the supports of $g$ and $h$ are mutually spacelike. This ensures that, up to constant factors, the commutation relations between $\hat{x}+\hat{y}$, $\hat{p}_x$ and $\hat{p}_y$ mirror those of $\phi(f)$, $\phi(h)$ and $\phi(g)$ respectively, e.g.
\begin{align}
    [\hat{x}+\hat{y},\hat{p}_x] &= i   &  [\hat{x}+\hat{y},\hat{p}_y] &= i & [\hat{p}_x , \hat{p}_y] &= 0 \nonumber
    \\
    [\phi(f),\phi(h)] &= i\Delta(f,h)   &[\phi(f),\phi(g)] &= i\Delta(f,g) &  [\phi(h) ,\phi(g)] &= 0
\end{align}

We stress that this mapping should be understood at the mathematical level alone, and we are not claiming any physical correspondence between a 2D harmonic oscillator and our QFT setup. In particular, the kick by Alice and the measurement by Bob in the QFT setup happen in spacelike separated regions of spacetime, whereas for the 2D oscillator, Alice and Bob's actions occur in the same place in space, one after another, and are thus timelike related. Furthermore, the 2D oscillator example is simple enough that the signalling protocol could (likely) be implemented in a real experiment. Since Alice and Bob's actions are not spacelike in this case, it should not be surprising that Alice can signal Bob using this protocol. We further note that the fact that this signal is (likely) visible in a real experiment \textit{does not} mean that the corresponding \textit{superluminal} signal in the QFT case is visible in any real experiments. The latter is clearly impossible, and simply points to a deficiency in our description of measurement in QFT.

\subsection{Evading the Signal with Pure Point Approximations}\label{sec:Evading the signal with pure point approximations}

Before moving on, let us discuss an interesting way in which the signal can be evaded in the simple 2D harmonic oscillator example. One can imagine approximating the position operators, $\hat{x}$ and $\hat{y}$, by operators $\hat{x}_\varepsilon$ and $\hat{y}_\varepsilon$ with pure point spectra, where $\varepsilon >0$ controls the accuracy of the approximation. For example, we could take $\hat{x}_\varepsilon$ to act on a general state $\Psi(x,y)$ as
\begin{equation}
    \hat{x}_{\varepsilon} \Psi(x,y) = \sum_{n\in\mathbb{Z}} n \varepsilon \, 1_{[n \varepsilon, (n+1)\varepsilon)}(x)\, \Psi(x,y) \; ,
\end{equation}
and similarly for $\hat{y}_\varepsilon$. Both $\hat{x}_{\varepsilon}$ and $\hat{y}_{\varepsilon}$ have pure point spectra with eigenvalues $n\varepsilon$ (labelled by $n\in\mathbb{Z}$), and the associated projectors are the indicator functions $1_{[n \varepsilon, (n+1)\varepsilon)}(x)$ and $1_{[n \varepsilon, (n+1)\varepsilon)}(y)$ respectively.

In making an ideal measurement of either $\hat{x}_{\varepsilon}$ or $\hat{y}_{\varepsilon}$ we can pick any resolution $\bor{R}$, but there is a `canonical' choice of resolution in which each bin contains one and only one eigenvalue. For example, the resolution $\bor{R}_{\varepsilon} = \lbrace [n\varepsilon , (n+1)\varepsilon ) \rbrace_{n\in\mathbb{Z}}$, with open-closed bins $[n\varepsilon , (n+1)\varepsilon )$ does the job. The associated projectors that appear in the ideal measurement update maps $\mathcal{E}_{\hat{x}_{\varepsilon},\bor{R}_{\varepsilon}}$ and $\mathcal{E}_{\hat{y}_{\varepsilon},\bor{R}_{\varepsilon}}$ are then $1_{[n \varepsilon, (n+1)\varepsilon)}(x)$ and $1_{[n \varepsilon, (n+1)\varepsilon)}(y)$ respectively.

One can show that $\norm{(\hat{x}_{\varepsilon}-\hat{x})\Psi}\leq \varepsilon \norm{\Psi}$, for any state $\Psi(x,y)$ in the domains of both $\hat{x}_{\varepsilon}$ and $\hat{x}$, and thus $\norm{(\hat{x}_{\varepsilon}-\hat{x})\Psi}\rightarrow 0$ as $\varepsilon\rightarrow 0$. In this precise sense $\hat{x}_{\varepsilon}$ approximates $\hat{x}$, and the same can be said for $\hat{y}_{\varepsilon}$ and $\hat{y}$. Even though we can approximate $\hat{x}$ and $\hat{y}$ arbitrarily closely in this sense, $\hat{x}_{\varepsilon}$ and $\hat{y}_{\varepsilon}$ have pure point spectra for all $\varepsilon >0$, unlike $\hat{x}$ and $\hat{y}$ which have continuous spectra.

We are interested in a measurement of $O = \hat{x} + \hat{y}$, and so we define $O_{\varepsilon} = \hat{x}_{\varepsilon} + \hat{y}_{\varepsilon}$. From the triangle inequality we can easily see that $\norm{(O_{\varepsilon}-O)\Psi}\leq 2\varepsilon \norm{\Psi}$, and so $O_{\varepsilon}$ approximates $O$ in the same sense as above. We can also write $O_{\varepsilon}$ as
\begin{align}
    O_{\varepsilon} & = \sum_{n\in\mathbb{Z}}n\varepsilon \, 1_{[n \varepsilon, (n+1)\varepsilon)}(x) + \sum_{m\in\mathbb{Z}}m\varepsilon \, 1_{[m \varepsilon, (m+1)\varepsilon)}(y)  \nonumber
    \\
    & = \sum_{n\in\mathbb{Z}}n\varepsilon \, 1_{[n \varepsilon, (n+1)\varepsilon)}(x)\left(\sum_{m\in\mathbb{Z}} 1_{[m \varepsilon, (m+1)\varepsilon)}(y)\right) + \sum_{m\in\mathbb{Z}}m\varepsilon \, 1_{[m \varepsilon, (m+1)\varepsilon)}(y) \left(\sum_{n\in\mathbb{Z}} 1_{[n \varepsilon, (n+1)\varepsilon)}(x)\right) \nonumber
    \\
    & = \sum_{n,m\in\mathbb{Z}}n\varepsilon \, 1_{[n \varepsilon, (n+1)\varepsilon)}(x) 1_{[m \varepsilon, (m+1)\varepsilon)}(y) + \sum_{n,m\in\mathbb{Z}}m\varepsilon \, 1_{[n \varepsilon, (n+1)\varepsilon)}(x) 1_{[m \varepsilon, (m+1)\varepsilon)}(y) \nonumber
    \\
    & = \sum_{n,m\in\mathbb{Z}}(n+ m)\varepsilon \, 1_{[n \varepsilon, (n+1)\varepsilon)}(x) 1_{[m \varepsilon, (m+1)\varepsilon)}(y) \; .
\end{align}
This makes it clear that $O_{\varepsilon}$ also has a pure point spectrum, and its eigenvalues are simply the distinct values of $(n+m)\varepsilon$ that appear in the sum on the final line, i.e. its eigenvalues are $k\varepsilon$ for $k\in\mathbb{Z}$. For a given eigenvalue we can read off the associated projector from the last line above. For eigenvalue $k\varepsilon$ the associated projector is the 2D indicator function $1_{S_k}(x,y)$, where we have defined the subset $S_k\subset\mathbb{R}^2$ as
\begin{equation}
    S_k = \cup_{\substack{n,m\in\mathbb{Z}\, ,\\ n+m=k}}[n \varepsilon, (n+1)\varepsilon)\times [m \varepsilon, (m+1)\varepsilon) \; .
\end{equation}
As an example, we illustrate $S_0$ in Fig.~\ref{fig:borel set}. We can now write
\begin{equation}\label{eq:O_eps diagonalised}
    O_{\varepsilon} = \sum_{k\in\mathbb{Z}}k\varepsilon \, 1_{S_k}(x,y) \; .
\end{equation}
In making an ideal measurement of $O_{\varepsilon}$ we can choose any resolution, but again there is a `canonical' choice corresponding to bins which contain one and only one eigenvalue. Since the eigenvalues of $O_{\varepsilon}$ are $k\varepsilon$, the same choice of resolution as above, i.e. $\bor{R}_{\varepsilon} = \lbrace [k\varepsilon , (k+1)\varepsilon ) \rbrace_{k\in\mathbb{Z}}$, does the job. For a given bin, $[k\varepsilon , (k+1)\varepsilon )$, the associated projector for $O_{\varepsilon}$ is then $1_{S_k}(x,y)$, and these are the projectors that would appear in the ideal measurement update map $\mathcal{E}_{O_{\varepsilon} , \bor{R}_{\varepsilon}}$. 
\begin{figure}
 \centering
 \includegraphics[width=0.5\textwidth]{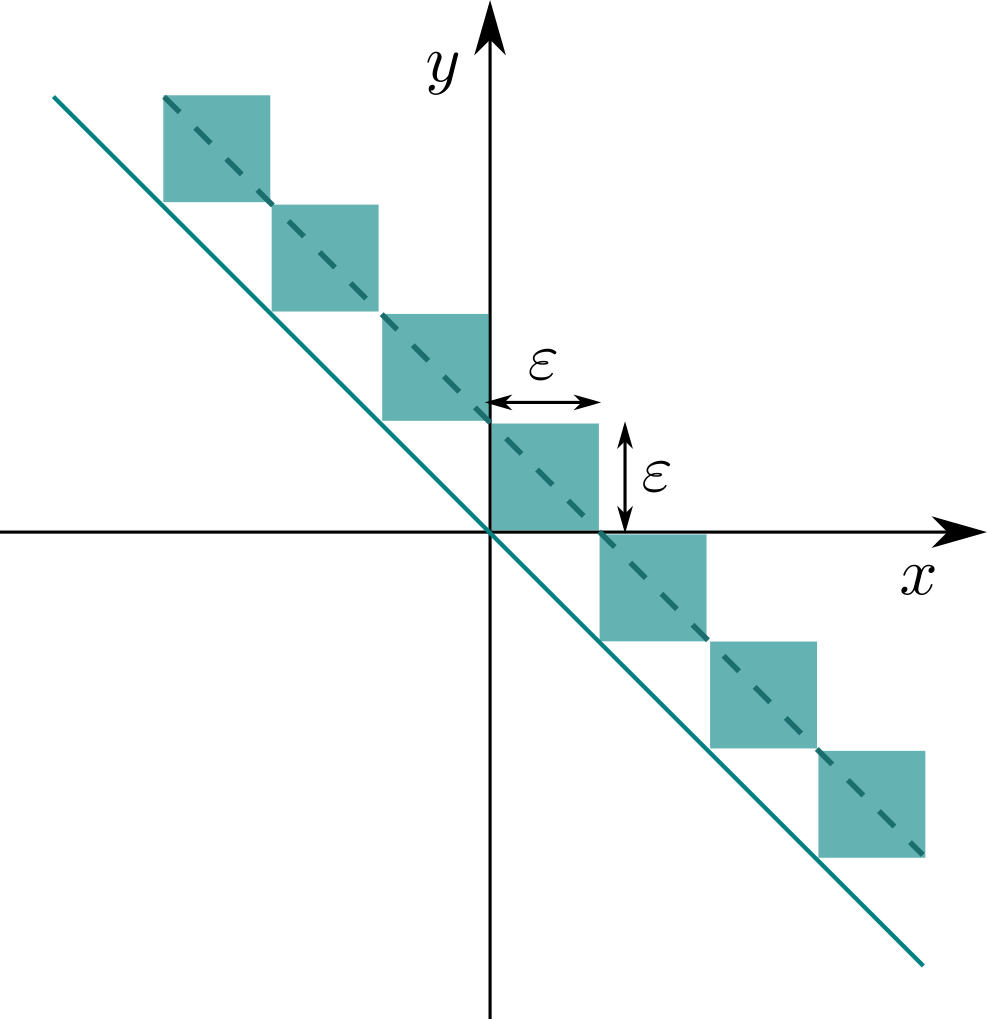}
 \caption{The set $S_0$ is the union of the shaded squares. The support of the indicator function $1_{[0,\varepsilon )}(x,y)$ lies on and above the solid line, and strictly below the dashed line.}
 \label{fig:borel set}
\end{figure}

From~\cite{Borsten_2021} we know that since both $\hat{x}_{\varepsilon}$ and $\hat{y}_{\varepsilon}$ have pure point spectra, an ideal measurement by Charlie of $O_{\varepsilon}$ (with the `canonical' resolution $\bor{R}_{\varepsilon}$) will \textit{not} enable a signal from Alice to Bob, no matter how small $\varepsilon >0$ is. This may seem puzzling since $O$ does enable a signal, and since $O_{\varepsilon}$ can approximate $O$ arbitrarily closely in the sense above. Nevertheless, there is still a physical difference between an ideal measurement of $O = \hat{x}+ \hat{y}$ and an ideal measurement of $O_{\varepsilon} = \hat{x}_{\varepsilon} + \hat{y}_{\varepsilon}$. 

The ideal measurement of $O_{\varepsilon}$ implicitly comes with the `canonical' resolution, e.g. $\bor{R}_{\varepsilon}$, but the ideal measurement of $O = \hat{x}+ \hat{y}$ requires us to specify some resolution. For comparison's sake, we consider the same resolution $\bor{R}_{\varepsilon}$.

Even with the same resolution for both measurements, there is still a physical difference between the two cases, and one that appears to make all the difference in terms of the signalling protocol. Let us focus on the bin $[0,\varepsilon)\in\bor{R}_{\varepsilon}$ to highlight this difference. As noted in Section~\ref{sec:An illustrative example}, given some Borel set, $\bor{B}$, the associated projector for $O=\hat{x}+\hat{y}$ is the indicator function $1_{\bor{B}}(x+y)$, and thus for the bin $[0,\varepsilon)$ the associated projector is $1_{[0,\varepsilon)}(x+y)$. The support of this projector looks like an infinitely long strip in the $(x,y)$-plane at right angles to the $(x+y)$-direction (Fig.~\ref{fig:borel set}). For $O_{\varepsilon} = \hat{x}_{\varepsilon} + \hat{y}_{\varepsilon}$, on the other hand, the associated projector is $1_{S_0}(x,y)$, as mentioned above~\footnote{One can also see this by taking $1_{[0,\varepsilon )}(O_{\varepsilon})$, with $O_{\varepsilon}$ given by the rhs of~\eqref{eq:O_eps diagonalised}. The only way to get $1_{[0,\varepsilon )}(O_{\varepsilon})=1$ is if $O_{\varepsilon}$ evaluates to something in the open-closed interval $[0,\varepsilon)$, and that only happens if if $(x,y)$ lies in the set $S_0$, and thus $1_{[0,\varepsilon )}(O_{\varepsilon}) = 1_{S_0}(x,y)$}. The support of this projector is the set $S_0$ (shown in Fig.~\ref{fig:borel set}), which is clearly different to the former case, and thus amounts to a physically distinct yes/no question. In the former case we are asking whether the particle's position lies within the strip in Fig.~\ref{fig:borel set}, and in the latter case we are asking whether its position lies within any of the squares shown in Fig.~\ref{fig:borel set}.

\section{A Decoherence Functional/Path Integral Perspective}\label{sec:A decoherence functional/path integral perspective}

To finish we turn to the decoherence functional or path integral formalism, as this alternative perspective on the calculations in Section~\ref{sec:Causality of smeared field operations} may offer some further intuition as to why an ideal measurement of a smeared field operator enables a signal. The path integral formalism of Sorkin's scenario was considered in~\cite{Fuksa2021mfe}, though they were not optimistic that it would lead to further insight. 

Here we consider a background causal set (following the corresponding QFT path integral formalism in~\cite{deco_scalarfield}), instead of a continuum spacetime, as the discreteness of spacetime means we can rigorously define the QFT path integral, and thus avoid any technical issues surrounding the continuum path integral.

In reformulating the above calculations in Section~\ref{sec:Causality of smeared field operations}, and in particular the non-selective ideal measurement, we arrive not at the single path integral formalism, but at the \textit{double} path integral formalism, where one sums, or integrates, over all \textit{pairs} of paths. In the context of QFT on a fixed causet $\spctm{C}$, a single `path' correspond to a field configuration on $\spctm{C}$ --- an assignment of a real number to each causet point. A pair of `paths' then corresponds to a pair of field configurations on $\spctm{C}$. For more background on the double path integral see~\cite{hartle2014spacetime}. We also note that the single path integral can always be recovered from the double path integral by simply integrating out one of the paths in the pair.

For simplicity we consider a causet, $\spctm{C}$, consisting of only 4 points labelled $A$, $1$, $2$, and $B$, with the order relations $A\preceq 1$ and $2\preceq B$ (see Fig.~\ref{fig:causet double path integral}). This is the minimal setup in which we can formulate Sorkin's scenario. For brevity, we write $f_x$ for the value of some function $f$ at a point $x\in\lbrace A, 1 , 2, B\rbrace$, instead of as $f(x)$ as was done in the preceding sections. Unlike in the continuum, the discreteness means we can define the field operator at a point $x\in\lbrace A, 1 , 2, B\rbrace$, which we similarly denote as $\phi_x$. Note, $\phi_x$ is equivalent to the smeared field operator, $\phi(f)$, with $f$ taken to be $1$ on $x$ and $0$ everywhere else.
\begin{figure}
 \centering
 \includegraphics[width=0.4\textwidth]{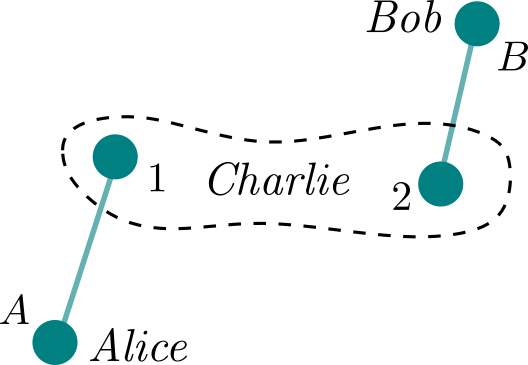}
 \caption{A minimal causet of four elements in which we can consider Sorkin's scenario. Alice unitarily kicks on point $A$, Charlie makes a measurement of a smeared field operator over points $1$ and $2$, and Bob measures his expectation value on point $B$.}
 \label{fig:causet double path integral}
\end{figure}

Moving on to Sorkin's scenario in this setup, we start in the ground state $\gs$ (which one can take to be the Sorkin-Johnston state~\cite{Afshordi_2012_A,Afshordi_2012_B}, or some other candidate vacuum state). Alice then unitarily kicks with $e^{-is \phi_A}$ on point $A$, Charlie performs an ideal measurement (with some resolution $\bor{R}$) of the smeared field operator $\phi(f)$ on points $1$ and $2$, and finally Bob measures their field $\phi_B$ on point $B$ and reconstructs the expectation value of $e^{i t \phi_B}$ from their measurements (see Sections~\ref{sec:Assumptions} for more details). To restrict Charlie's operator $\phi(f)$ to sites $1$ and $2$ we take $f$ to be $0$ on points $A$ and $B$, and to have some non-zero values, $f_1,f_2\neq 0$, on points $1$ and $2$. We can then equivalently write $\phi(f) = f_1 \phi_1 + f_2 \phi_2$.

Following~\eqref{eq:chi s} in Section~\ref{sec:General causality condition for a smeared field operations}, Bob's expectation value, as a function of Alice's kick strength $s$, is then
\begin{align}\label{eq:chi s decoherence}
    \chi (s) & = \gsb \mathcal{U}_{s\phi_A}( \mathcal{E}_{\phi(f),\bor{R}}(e^{it\phi_B}) )\gs \nonumber
    \\
    & =  \gsb e^{is\phi_A} \mathcal{E}_{\phi(f),\bor{R}}(e^{it\phi_B}) e^{-is\phi_A}\gs \nonumber
    \\
    & = \sum_{n\in I} \gsb e^{is\phi_A} \, 1_{\bor{B}_n}(f_1 \phi_1 + f_2 \phi_2) \, e^{it\phi_B} \, 1_{\bor{B}_n}(f_1 \phi_1 + f_2 \phi_2) \, e^{-is\phi_A}\gs \; ,
\end{align}
where in the last line we have substituted in $f_1 \phi_1 + f_2 \phi_2$ for $\phi(f)$ in the arguments of the indicator functions.

In~\cite{deco_scalarfield} Sorkin discusses how to translate expectation values (such as the rhs in~\eqref{eq:chi s decoherence}) from the operator formalism into double path integrals, assuming $\gs$ is the Sorkin-Johnston state. For our simple causet $\spctm{C}$, this is done by integrating all pairs of paths --- pairs of field configurations on $\spctm{C}$ --- against the integral measure, or \textit{decoherence functional},
\begin{align}\label{eq:decoherence functional}
    D(\xi , \overline{\xi}) = \gsb& \delta (\phi_A - \xi_A )\delta (\phi_1 - \xi_1 )\delta (\phi_2- \xi_2)\delta (\phi_B - \xi_B )
    \nonumber
    \\
    & \times \delta (\phi_B - \overline{\xi}_B )\delta (\phi_2 - \overline{\xi}_2 )\delta (\phi_1 - \overline{\xi}_1 )\delta (\phi_A - \overline{\xi}_A )\gs \; ,
\end{align}
where $\xi , \overline{\xi}\in\mathbb{R}^4$ represent a pair of field configurations on the four element causet $\spctm{C}$. Note that going from left to right the delta functions for the $\xi$ terms have been inserted in an order that respects the causal ordering (e.g. the term for point $1$ comes after that for $A$). The $\overline{\xi}$ terms are then in the opposite order. In this way one can encode the spacetime causal structure in the decoherence functional $D(\xi , \overline{\xi})$.

One of the main results in~\cite{deco_scalarfield} is that we can write the decoherence functional in~\eqref{eq:decoherence functional} as
\begin{align}\label{eq:decoherence functional action form}
    D(\xi,\bar{\xi}) = \delta (L(\xi, \bar{\xi})) \, e^{i \Delta S(\xi , \bar{\xi})}, \; 
\end{align} 
where $\delta(\cdot )$ is the usual delta-function, $L$ is some linear function, and $\Delta S$ is a quadratic function. The factor $e^{i \Delta S(\xi , \bar{\xi})}$ is analogous to $e^{i (S[\gamma ] - S[ \gamma' ])}$ in the double path integral for a non-relativistic particle propagating in $d$ dimensions~\cite{hartle2014spacetime,Hilbert_deco}, where $\gamma$ and  $\gamma'$ are different trajectories over which the integration is computed, and $S[\cdot ]$ is the action of the trajectory.

Now, by integrating the measure $D(\xi,\overline{\xi})$ over all pairs of field configurations (i.e. all $\xi , \overline{\xi}\in\mathbb{R}^4$) against some function of $\xi$ and $\overline{\xi}$, one can recover any expectation value from the operator formalism. In particular, for~\eqref{eq:chi s decoherence} we find 
\begin{align}\label{eq:chi s decoherence double path integral}
    \chi (s) = \int_{\mathbb{R}^4}d^4\xi \int_{\mathbb{R}^4}d^4\overline{\xi} \, D(\xi,\overline{\xi}) \, e^{-is(\xi_A - \overline{\xi}_A)}e^{it\xi_B}\sum_{n\in I}1_{\bor{B}_n}(f_1 \xi_1 + f_2 \xi_2) 1_{\bor{B}_n}(f_1 \overline{\xi}_1 + f_2 \overline{\xi}_2) \, .
\end{align}
Note, instead of $e^{i t \xi_B}$ we could have also integrated over $e^{i t \overline{\xi}_B}$ to arrive at the expression on the last line of~\eqref{eq:chi s decoherence}.

Now, without Charlie's ideal measurement of $\phi(f) = f_1\phi_1 + f_2\phi_2$, and hence without the indicator functions in~\eqref{eq:chi s decoherence double path integral}, we find that $\chi(s)$ is independent of $s$; Alice cannot send a signal to Bob. This manifests in the double path integral formalism as an independence of Bob's marginal measure from $s$. That is, if we were to integrate out all the variables not to the past of Bob, i.e. $\xi_A$, $\xi_1$, $\overline{\xi}_A$ and $\overline{\xi}_1$, we would find that the resulting marginal measure does not depend on $s$. This essentially follows from the causal nature of $D(\xi , \overline{\xi})$, which was encoded in the ordering of the delta functions in~\eqref{eq:decoherence functional}.

With Charlie's ideal measurement, however, the indicator functions in~\eqref{eq:chi s decoherence double path integral} `tie' together the field values on points $1$ and $2$ in a way that means that Bob's marginal measure does depend on $s$. To make this clearer, we define the set $S_n = \lbrace (\xi_A , \xi_1 , \xi_2 , \xi_B )\in\mathbb{R}^4 | f_1 \xi_1 + f_2 \xi_2 \in \bor{B}_n \rbrace$, for each $n\in I$. We can then write~\eqref{eq:chi s decoherence double path integral} as
\begin{align}\label{eq:chi s decoherence restricted double path integral}
    \chi (s) = \sum_{n\in I}\int_{S_n}d^4\xi \int_{S_n}d^4\overline{\xi} \, D(\xi,\overline{\xi}) \, e^{-is(\xi_A - \overline{\xi}_A)}e^{it\xi_B}\, .
\end{align}
The integrals over field configurations $\xi$ and $\overline{\xi}$ have now been restricted to the subsets $S_n\subset\mathbb{R}^4$. Importantly, whether or not a given value $\xi_2$ is in a subset $S_n$, and hence whether it is integrated over, depends on the current value of $\xi_1$ under the integral. In this precise sense we say that the values of $\xi_1$ and $\xi_2$ have been `tied' together in the integral over $\xi$. The same is true for $\overline{\xi}_1$ and $\overline{\xi}_2$ in the integral over $\overline{\xi}$. It is this tying together of the field values on the pair of spacelike points $1$ and $2$ that results in Bob's marginal measure depending on Alice's kick strength $s$, despite her being causally disconnected from Bob.

In a future project it would be interesting to study the double path integral picture further, and to determine what sort of functions of the field values on points $1$ and $2$ give rise to this acausal behaviour. This could lead to a more intuitive condition for whether a given operation is consistent with relativistic causality.

\section{Discussion}\label{sec:Discussion}

We have focused on one of the simplest types of operators in real scalar field theory, namely, smeared field operators, which, together with the identity, generate the entire algebra. We formulated and studied the general class of, what we call, Kraus update maps, which includes the case of projective, or ideal, measurements, as well as a host of other typical maps seen in QFT, e.g. unitary kicks. Using the machinery of functional calculus we were able to reduce the question of whether a given Kraus map enables a superluminal signal in Sorkin's scenario to a question of whether the functions involved in the Kraus map have a certain property (see Claim~\ref{claim:causality condition on kappa tilde}). We further illustrated the utility of this property in quickly determining the causal nature of a given Kraus map (Section~\ref{sec:General causality condition for a smeared field operations}). 

We then turned to the specific case of ideal measurements of smeared fields, and in Section~\ref{sec:The acausality of an ideal measurement of a smeared field} we showed our main result that, for any resolution, such a measurement enables a superluminal signal in Sorkin's scenario, and is thus not possible to realise in any experiment. This result casts doubt on the role of the projection postulate in relativistic QFT. If one wishes to retain the projection postulate as an idealised, operationalist description of measurement at the level of QFT, then one must evade one of the assumptions of our calculation (see Section~\ref{sec:Assumptions}). That said, we have only considered ideal measurements of smeared field operators, and it may be the case that ideal measurements of some other local operators in the theory are not acausal. This seems somewhat unlikely, however, since smeared field operators are the generators of all other local operators.

Returning to the point about evading our assumptions, this could be achieved by only considering setups in which Sorkin's scenario is not possible. As in Example~\ref{ex:Cylinder spacetime} this could be due to the spatial topology of the spacetime. This exemption does not cover the case of $1+1$ Minkowski spacetime (and most likely higher dimensions), however, as for any compact subset $K$ in which the measurement takes place Sorkin's scenario is always possible. Another way to avoid Sorkin's scenario is to consider a discrete spacetime and only allow measurements local to a single point (Example~\ref{ex:Single causal set element}). This option is somewhat appealing in its simplicity, and it highlights an interesting relationship between the projection postulate and the continuum or discrete nature of spacetime. On the other hand, one can still consider ideal measurements of smeared field operators covering a pair of spacelike points in a discrete spacetime (Fig.~\ref{fig:causet double path integral}), and this would enable a superluminal signal. Thus, we have to impose, by hand, the restriction that such measurements are not possible, while similar measurements on single points in spacetime are.

All this concern surrounding the projection postulate seems somewhat unnecessary, however, as there already exist perfectly causal and local update maps (in continuum and discrete spacetimes), e.g. the $L^2$-Kraus maps for smeared fields in Example~\ref{ex:L2-Kraus update map for a smeared field}, which we can use to describe measurements in the theory. These $L^2$-Kraus maps are just one family of candidates that could take the place of the projection postulate in QFT as an operationalist description of how a measurement affects the system. The choice of $L^2$ function in the given map could also take the place of the choice of resolution in the ideal measurement update map.

It is important to note, however, that the projection postulate still provides a description of measurement that accords with experimental evidence for non-relativistic quantum systems. Thus, even if we were to drop the projection postulate from QFT, we must still recover it at some effective level when taking a non-relativistic limit of the underlying QFT. An interesting future direction in this regard is then to study the non-relativistic limit of the causal and local $L^2$-Kraus update maps, for example, and determine whether any such maps give rise to the projection postulate in non-relativistic quantum mechanics. In addition to this, it would also be important to connect such update maps to scattering amplitudes (and their resultant probabilities), as this is our primary route to describing physical QFT experiments such as those at colliders.

\section*{Acknowledgements}

EA is supported by the STFC Consolidated Grant ST/W507519/1, and IJ is supported by a DIAS Scholarship.

\section{Appendices}\label{sec:Appendices}

\subsection{Unitary Action on p.v.m}\label{app:unitary_action_on_pvm}

\begin{claim}
Given a self-adjoint operator $X$, with dense domain $Dom(X)\subseteq \alg{F}$, and some unitary operator, $U$, the p.v.m's for $X$ and $U^{-1}X U$ satisfy:
\begin{equation}\label{eq:app_unitary_action_on_pvm}
    P^{U^{-1}XU}(\cdot) = U^{-1}P^X(\cdot ) U \; .
\end{equation}
\end{claim}

\begin{proof}
We first note that $U^{-1}P^X(\cdot ) U$ satisfies all the requirements of a p.v.m. The Spectral Theorem then tells us there exists a corresponding self-adjoint operator, $Y$, with this p.v.m, i.e. $P^{Y}(\cdot) = U^{-1}P^X(\cdot ) U$. 

For any states $\ket{\psi},\ket{\chi}\in\alg{F}$, we get the following relationship for the associated complex measures:
\begin{align}
    \mu^{Y}_{\psi , \chi} (\cdot ) & = \bra{\psi} P^Y (\cdot ) \ket{\chi} \nonumber
    \\
    & = \bra{\psi} U^{-1} P^X (\cdot ) U \ket{\chi} \nonumber
    \\
    & = \mu^{X}_{U\psi , U\chi} (\cdot ) \; .
\end{align}
With this we can determine $Y$ in terms of $X$ and $U$ explicitly. Using these measures we find
\begin{align}\label{eq:determining_Y}
    \bra{\psi}Y \ket{\chi} & = \int_{\mathbb{R}} \lambda \, d\mu^{Y}_{\psi , \chi}(\lambda ) \nonumber
    \\
    & = \int_{\mathbb{R}} \lambda \, d\mu^{X}_{U\psi , U\chi}(\lambda) \nonumber
    \\
    & = \bra{\psi}U^{-1}X U \ket{\chi} \; ,
\end{align}
and hence $Y= U^{-1}X U$, which, as $P^{Y}(\cdot) = U^{-1}P^X(\cdot ) U$, further implies our desired result in~\eqref{eq:app_unitary_action_on_pvm}.

\end{proof}

Note that, for the expression in the last line of~\eqref{eq:determining_Y} to make sense, we need $\ket{\psi},\ket{\chi}\in U^{-1}\text{Dom}(X)$. This sets the domain of $Y$ as $\text{Dom}(Y) = U^{-1}\text{Dom}(X)$.

\subsection{Smeared Field Unitary Action on Functions of Another Smeared Field}\label{app:unitary_kick_on_pvm_phi_f}

\begin{claim}
Given two real-valued test functions $f$ and $g$, the p.v.m for $\phi(f)$ satisfies
\begin{equation}\label{eq:app_unitary_kick_on_pvm_phi_f}
    e^{-i\phi(g)}P^{\phi(f)}(\bor{B} )e^{i\phi(g)} = P^{\phi(f) - \Delta(f,g)}(\bor{B}) = P^{\phi(f)}(\bor{B} + \Delta(f,g)) \; ,
\end{equation}
for any $\bor{B}\in\BorR$.
\end{claim}

\begin{proof}
From the Baker–Campbell–Hausdorff formula, e.g. $e^Xe^Y  = e^{X+Y + \frac{1}{2}[X,Y]+...}$, we find that
\begin{equation}
    e^{is\phi(f)}e^{i\phi(g)} = e^{is\phi(f)+i\phi(g)-\frac{1}{2}s[\phi(f),\phi(g)]}=e^{is\phi(f)+i\phi(g)-\frac{i}{2}s\Delta(f,g)} \; ,
\end{equation}
where $f$ and $g$ are real-valued test functions, and $s\in\mathbb{R}$. With this we can further show that
\begin{equation}
    e^{-i\phi(g)}e^{is\phi(f)}e^{i\phi(g)} = e^{-is\Delta(f,g)}e^{is\phi(f)} \; .
\end{equation}
By differentiating with respect to $s$ we can then show that
\begin{align}
    e^{-i\phi(g)}\phi(f)e^{i\phi(g)} & = - i \frac{\partial}{\partial s}\left( e^{-i\phi(g)}e^{is\phi(f)}e^{i\phi(g)} \right) \Big|_{s=0} \nonumber
    \\
    & = - i \frac{\partial}{\partial s}\left( e^{-is\Delta(f,g)}e^{is\phi(f)} \right) \Big|_{s=0} \nonumber
    \\
    & = -i \left( i\phi(f) - i \Delta(f,g)) \right) \nonumber
    \\
    & = \phi(f) - \Delta(f,g) \; .
\end{align}

Now, using~\eqref{eq:app_unitary_action_on_pvm} with the unitary $U = e^{i\phi(g)}$ and the self-adjoint operator $X=\phi(f)$, we then have
\begin{equation}
    e^{-i\phi(g)}P^{\phi(f)}(\cdot )e^{i\phi(g)} = P^{e^{-i\phi(g)}\phi(f)e^{i\phi(g)}}(\cdot) = P^{\phi(f) - \Delta(f,g)}(\cdot) \; .
\end{equation}
Given some $\bor{B}\in\BorR$, and using~\eqref{eq:pvm_as_indicator_function}, we can write the far rhs as
\begin{equation}
    P^{\phi(f) - \Delta(f,g)}(\bor{B}) = 1_{\bor{B}}(\phi(f)-\Delta(f,g)) \; .
\end{equation}
As a function, $1_{\bor{B}} : \mathbb{R}\rightarrow \mathbb{R}$ satisfies $1_{\bor{B}}(\lambda - c) = 1_{\bor{B}+c}(\lambda)$ for any $c\in\mathbb{R}$, and hence the analogous operator statement also holds through functional calculus. Namely,
\begin{align}
    1_{\bor{B}}(\phi(f)-\Delta(f,g)) & = \int_{\mathbb{R}} 1_{\bor{B}}(\lambda - \Delta(f,g)) dP^{\phi(f)}(\lambda) \nonumber
    \\
    & = \int_{\mathbb{R}} 1_{\bor{B}+\Delta(f,g)}(\lambda) dP^{\phi(f)}(\lambda) \nonumber
    \\
    & = 1_{\bor{B}+\Delta(f,g)}(\phi(f)) \; .
\end{align}
Using~\eqref{eq:pvm_as_indicator_function} we can equivalently write $1_{\bor{B}+\Delta(f,g)}(\phi(f)) = P^{\phi(f)}(\bor{B}+\Delta(f,g))$, and thus we get our desired result in~\eqref{eq:app_unitary_kick_on_pvm_phi_f}.
\end{proof}

\begin{claim}
Given two real-valued test functions $f$ and $g$, and some function $\zeta:\mathbb{R}\rightarrow\mathbb{C}$ for which $\zeta(\phi(f))$ has dense domain $\text{Dom}(\zeta(\phi(f)))\subseteq \alg{F}$, we have the operator equation
\begin{equation}\label{eq:app_unitary_kick_on_zeta_phi_f}
    e^{-i\phi(g)}\zeta(\phi(f))e^{i\phi(g)} = \zeta(\phi(f) - \Delta(f,g)) \; ,
\end{equation}
on the domain $e^{-i\phi(g)}\text{Dom}(\zeta(\phi(f)))\subseteq \alg{F}$.
\end{claim}

\begin{proof}
For any two vectors $\ket{\psi},\ket{\chi}\in e^{-i\phi(g)}\text{Dom}(\zeta(\phi(f)))\subseteq\alg{F}$, the following inner product is well defined:
\begin{equation}
    \bra{\psi}e^{-i\phi(g)}\zeta(\phi(f))e^{i\phi(g)}\ket{\chi} = \bra{e^{i\phi(g)}\psi}\zeta(\phi(f))\ket{e^{i\phi(g)}\chi} \; .
\end{equation}
Through functional calculus we can rewrite this as
\begin{equation}\label{eq:unitary kick on zeta of phi f integral}
    \bra{e^{i\phi(g)}\psi}\zeta(\phi(f))\ket{e^{i\phi(g)}\chi} = \int_{\mathbb{R}}\zeta(\lambda) d\mu^{\phi(f)}_{e^{i\phi(g)}\psi, e^{i\phi(g)}\chi}(\lambda) \; ,
\end{equation}
where, for any $\bor{B}\in\BorR$, the complex measure on the rhs is given by
\begin{align}
    \mu^{\phi(f)}_{e^{i\phi(g)}\psi, e^{i\phi(g)}\chi}(\bor{B}) & = \bra{e^{i\phi(g)}\psi}P^{\phi(f)}(\bor{B})\ket{e^{i\phi(g)}\chi} \nonumber
    \\
    & = \bra{\psi}e^{-i\phi(g)}P^{\phi(f)}(\bor{B})e^{i\phi(g)}\ket{\chi} \nonumber
    \\
    & = \bra{\psi}P^{\phi(f)}(\bor{B}+\Delta(f,g))\ket{\chi} \nonumber
    \\
    & = \mu^{\phi(f)}_{\psi, \chi}(\bor{B}+\Delta(f,g)) \; ,
\end{align}
using~\eqref{eq:app_unitary_kick_on_pvm_phi_f} to get line 3. Thus, we can rewrite~\eqref{eq:unitary kick on zeta of phi f integral} as
\begin{align}
    \bra{e^{i\phi(g)}\psi}\zeta(\phi(f))\ket{e^{i\phi(g)}\chi} & = \int_{\mathbb{R}}\zeta(\lambda) d\mu^{\phi(f)}_{e^{i\phi(g)}\psi, e^{i\phi(g)}\chi}(\lambda) \nonumber
    \\
    & = \int_{\mathbb{R}}\zeta(\lambda) d\mu^{\phi(f)}_{\psi,\chi}(\lambda +\Delta(f,g)) \nonumber
    \\
    & = \int_{\mathbb{R}}\zeta(\lambda' - \Delta(f,g)) d\mu^{\phi(f)}_{\psi,\chi}(\lambda') \nonumber
    \\
    & = \bra{\psi}\zeta(\phi(f) - \Delta(f,g))\ket{\chi} \; ,
\end{align}
where we have changed variables to $\lambda' = \lambda +\Delta(f,g)$ from line 2 to 3. In summary, we have the equality
\begin{equation}
    \bra{\psi}e^{-i\phi(g)}\zeta(\phi(f))e^{i\phi(g)}\ket{\chi} = \bra{\psi}\zeta(\phi(f) - \Delta(f,g))\ket{\chi} \; ,
\end{equation}
for any $\ket{\psi},\ket{\chi}\in e^{-i\phi(g)}\text{Dom}(\zeta(\phi(f)))\subseteq\alg{F}$, and thus the operator equation $e^{-i\phi(g)}\zeta(\phi(f))e^{i\phi(g)} = \zeta(\phi(f) - \Delta(f,g))$ on this domain.

\end{proof}

\subsection{Existence of a Density for the Complex Measure}\label{app:absolute_continuity}

Defining the complex measure $\mu^{\phi(f)}_{\Omega , e^{it\phi(g)}\Omega}(\cdot ) = \gsb P^{\phi(f)}(\cdot ) e^{it \phi(g)}\gs$, where $g$ is a real-valued test function, we show the following:
\begin{claim}
A Lebesgue integrable complex density, $q$, exists for the complex measure $\mu^{\phi(f)}_{\Omega , e^{it\phi(g)}\Omega}$, such that under the integral we can write $d\mu^{\phi(f)}_{\Omega , e^{it\phi(g)}\Omega}(\lambda ) = q(\lambda) d\lambda$.
\end{claim}

\begin{proof}
Through the polarisation identity we can write the complex measure $\mu^{\phi(f)}_{\Omega , e^{it\phi(g)}\Omega}$ in terms of the real (and non-negative) measures
\begin{equation}
    \mu^{\phi(f)}_{a ,b }(\cdot) = \gsb (a+b e^{it\phi(g)})^{\dagger} P^{\phi(f)}(\cdot ) (a+b e^{it\phi(g)}) \gs \; ,
\end{equation}
where $a=\pm 1$ and $b=\pm 1 , \pm i$. Specifically,
\begin{equation}\label{eq:measure_polarisation_identity}
    \mu^{\phi(f)}_{\Omega , e^{it\phi(g)}\Omega} = \frac{1}{4}\left( \mu^{\phi(f)}_{+1 ,+1 } -\mu^{\phi(f)}_{+1 ,-1 } -i \mu^{\phi(f)}_{+1 ,+i } +i \mu^{\phi(f)}_{+1 ,-i } \right) \; .
\end{equation}
The existence of a Lebesgue integrable complex density $q$ for $\mu^{\phi(f)}_{\Omega , e^{it\phi(g)}\Omega}$ is then implied by the existence of a Lebesgue integrable density, $p_{a,b}$, for the measure $\mu^{\phi(f)}_{a ,b }$ (for any $a,b\in\mathbb{C}$). By the Radon–Nikodym theorem, the existence of such a density is implied by the absolute continuity of $\mu^{\phi(f)}_{a ,b }$ with respect to the Lebesgue measure, which we will now prove.

Absolute continuity with respect to the Lebesgue measure means that, if $\bor{B}\in\BorR$ has Lebesgue measure $0$, then $\mu^{\phi(f)}_{a ,b }(\bor{B})=0$ (for any $a,b\in\mathbb{C}$). We first note that
\begin{equation}
    P^{\phi(f)}(\bor{B} ) (a+b e^{it\phi(g)}) \gs = a \, P^{\phi(f)}(\bor{B})\gs + b \, e^{it\phi(g)}P^{\phi(f)}(\bor{B}+t\Delta(f,g))\gs \; ,
\end{equation}
where we have used the fact that $P^{\phi(f)}(\bor{B})e^{it\phi(g)} = e^{it\phi(g)}P^{\phi(f)}(\bor{B}+t\Delta(f,g))$, which follows from~\eqref{eq:app_unitary_kick_on_pvm_phi_f}. Defining $\ket{\psi_a} = a P^{\phi(f)}(\bor{B})\gs$ and $\ket{\psi_b} = b P^{\phi(f)}(\bor{B}+t\Delta(f,g))\gs$ we can then write
\begin{equation}
    P^{\phi(f)}(\bor{B} ) (a+b e^{it\phi(g)}) \gs = \ket{\psi_a} + e^{it\phi(g)}\ket{\psi_b} \; .
\end{equation}
We further note that
\begin{align}
    \norm{\psi_a }^2 & = \braket{\psi_a | \psi_a } \nonumber
    \\
    & = |a|^2 \gsb P^{\phi(f)}(\bor{B})\gs  \nonumber
    \\
    & =|a|^2  \mu^{\phi(f)}_{\Omega ,\Omega}(\bor{B}) \nonumber
    \\
    & =  0 \; , 
\end{align}
which follows from the absolute continuity of $\mu^{\phi(f)}_{\Omega ,\Omega}$ and the fact that $\bor{B}$ has Lebesgue measure $0$. Similarly, 
\begin{align}
    \norm{\psi_b }^2 & = \braket{\psi_b | \psi_b } \nonumber
    \\
    & = |b|^2 \gsb P^{\phi(f)}(\bor{B}+t\Delta(f,g))\gs  \nonumber
    \\
    & =|b|^2  \mu^{\phi(f)}_{\Omega ,\Omega}(\bor{B}+t\Delta(f,g)) \nonumber
    \\
    & =  0 \; , 
\end{align} 
since $\bor{B}+t\Delta(f,g)$ also has Lebesgue measure $0$.

Thus, $\ket{\psi_a}$ and $\ket{\psi_b}$ are both simply the zero vector, and so $P^{\phi(f)}(\bor{B} ) (a+b e^{it\phi(g)}) \gs = \ket{\psi_a} + e^{it\phi(g)}\ket{\psi_b}$ is the zero vector too. Thus, 
\begin{equation}
    \mu^{\phi(f)}_{a ,b }(\bor{B}) = \gsb (a+b e^{it\phi(g)})^{\dagger} P^{\phi(f)}(\bor{B} ) (a+b e^{it\phi(g)}) \gs = 0 \; ,
\end{equation}
as it is simply the inner product of $(a+b e^{it\phi(g)}) \gs$ and the zero vector. Since we now have our desired absolute continuity of $\mu^{\phi(f)}_{a ,b }$ (for any $a,b\in\mathbb{C}$), we know that via the Radon–Nikodym theorem there exists an associated Lebesgue integrable density, $p_{a,b}$, and by~\eqref{eq:measure_polarisation_identity} there therefore exists a Lebesgue integrable complex density, $q$, for the complex measure $\mu^{\phi(f)}_{\Omega , e^{it\phi(g)}\Omega}$. Specifically,
\begin{equation}
    q = \frac{1}{4}\left( p_{+1 ,+1 } -p_{+1 ,-1 } -i p_{+1 ,+i } +i p_{+1 ,-i } \right) \; .
\end{equation}
\end{proof}

\subsection{Sorkin Scenario Results}\label{app:Sorkin scenario results}

Consider a globally hyperbolic (continuum) spacetime $(\spctm{M},\metric)$ and some classical solution, or mode, $\varphi\in \text{Sol}(\spctm{M})$ (a real-valued smooth solutions of compact spatial support).

For clarity in what follows we call a compact subset $K\subset \spctm{M}$ a \textit{lab} for $\varphi$ if there exists a test function $f\in C^{\infty}_0(\spctm{M})$ with $\supp f \subseteq K$, and such that $\varphi = \Delta f$. We call such an $f$ a \textit{lab source}.

It is useful to note that, given some $g\in C^{\infty}_0(\spctm{M})$, 
\begin{equation}
    \Delta(g,f) = \int_{\spctm{M}\times \spctm{M}} dx\, dy \; g(x) \Delta(x,y) f(y) = \int_{\spctm{M}} dx\; g(x) \varphi(x) \; ,
\end{equation}
and so the value $\Delta(g,f) = - \Delta(f,g)$ can also be computed by integrating the test function $g$ against the solution $\varphi$ generated by $f$.

\begin{definition}[Sorkin scenario]\label{def:Sorkin scenario}
    Given a mode $\varphi \in \text{Sol}(\spctm{M})$ and a corresponding lab $K$, we say a \textit{Sorkin scenario} exists for the mode-lab pair, $(\varphi , K)$, if there exists two test functions $h,g\in C^{\infty}_0(\spctm{M} ; \mathbb{R})$ of mutually spacelike support, and such that $\supp h\subset K^-$ and $\supp g\subset K^+$, and such that $\Delta(f,h)\neq 0$ and $\Delta(f,g)\neq 0$ for any lab source $f$ for $\varphi$.
\end{definition}

\begin{claim}\label{claim:Sorkin scenario equivalence 1}
For a given mode $\varphi\in \text{Sol}(\spctm{M})$ and corresponding lab $K$, a Sorkin scenario exists iff there are two mutually spacelike points $x_{\pm}\in K^{\pm}\cap \supp \varphi$. 
\end{claim}

\begin{proof}
\textbf{if}. By assumption there exist the two mutually spacelike points $x_{\pm}\in K^{\pm}\cap \supp\varphi$. Now, given $x_{\pm}\in \supp \varphi$, we know that for any neighbourhoods $N_{\pm}\ni x_{\pm}$, there are points $y_{\pm}\in N_{\pm}$ at which $\varphi(y_{\pm}) \neq 0$. More specifically, if $\varphi(x_{\pm})\neq 0$ we can simply take $y_{\pm} = x_{\pm}$. It may be the case, however, that $\varphi(x_{\pm}) = 0$ (even though $x_{\pm}\in \supp\varphi$, as $x_{\pm}$ could lie on the boundary of $\supp\varphi$), in which case there will be points $y_{\pm}$ arbitrarily close to $x_{\pm}$ at which $\varphi(y_{\pm})\neq 0$. The open-ness of the set of points spacelike to $x_+$ (which includes $x_-$ by assumption) ensures we can also pick $y_-$ spacelike to $x_+$. We can similarly pick $y_+$ spacelike to $y_-$. Next, by the smoothness of $\varphi$ we know there are two mutually spacelike neighbourhoods, $N_{\pm}\ni y_{\pm}$ where $\varphi|_{N_{\pm}}\neq 0$ . We can then pick any test function $h$ supported in $N_-$, and $g$ supported in $N_+$, and both non-negative. This ensures that integrating $\varphi$ against $h$ or $g$ will give a non-zero result. These integrals are the same as $\Delta(f,h)$ and $\Delta(f,g)$ respectively.

\textbf{only if}. By assumption there exist two such test functions, $h$ and $g$, as described above. As $\Delta(f,h)\neq 0$ we know that the integral of $\varphi$ against $h$ is non-zero, and thus there must be a point $x_-\in \supp h$ at which $\varphi(x_-)\neq 0$. Similarly, there must be a point $x_+\in \supp g$ at which $\varphi(x_+)\neq 0$. Since the supports of $h$ and $g$ are spacelike, we know that $x_-$ and $x_+$ are spacelike.
\end{proof}

The previous claim turns the statement of whether a Sorkin scenario exists for a given mode-lab pair to a statement about the modes functional form outside of the lab $K$. Next, we find a further equivalent statement in terms of the mode's Cauchy data. Given a Cauchy surface $\Sigma$, we let $C(\varphi,\Sigma) = (\varphi , \dot{\varphi})|_{\Sigma}$ denote the Cauchy data on $\Sigma$. Here the dot above denotes the derivative in the direction normal to $\Sigma$.

\begin{claim}\label{claim:Sorkin scenario equivalence 2}
There exist two ordered Cauchy surfaces $\Sigma_{\pm}$ (i.e. $I^{+}(\Sigma_-)\supset \Sigma_+$) such that the two subsets $\supp C(\varphi,\Sigma_-)$ and $\supp C(\varphi,\Sigma_+)$ are totally ordered (no pairs of mutually spacelike points $x_{\pm}\in \supp C(\varphi,\Sigma_{\pm})$ ) iff there exists a lab $K$ for $\varphi$ for which no Sorkin scenario is possible.
\end{claim}

\begin{proof}
\textbf{if}. Consider any lab $K$ for $\varphi$. As $K$ is compact we know there exists a pair of ordered Cauchy surfaces $\Sigma_{\pm}$ such that $K\subset J^+(\Sigma_-)\cap J^-(\Sigma_+)$. By assumption, there are no ordered Cauchy surfaces $\tilde{\Sigma}_{\pm}$ in which the supports of the respective Cauchy data are totally ordered. Thus, this also applies to our surfaces $\Sigma_{\pm}$. That is, there exist two points $x_{\pm}\in \supp C(\varphi,\Sigma_{\pm})$ which are mutually spacelike. By the open-ness of the set of points spacelike to a given point, and the smoothness of $\varphi$, one can then find two mutually spacelike points $y_{\pm}\in I^{\pm}(x_{\pm})$ for which $\varphi(y_{\pm})\neq 0$. By construction we also have that $y_{\pm}\in K^{\pm}\cap\supp\varphi$. By Claim~\ref{claim:Sorkin scenario equivalence 1} we know a Sorkin scenario then exists for the lab $K$. Thus, we have shown that the assumption that there are no ordered Cauchy surfaces in which the supports of the respective data are totally ordered implies that for \textit{every} lab $K$ there exists a Sorkin scenario.

\textbf{only if}. By assumption there exist two ordered Cauchy surfaces $\Sigma_{\pm}$ for which the supports of the corresponding Cauchy data are totally ordered, so that each point $x_- \in \supp C(\varphi,\Sigma_-)$ is timelike or null related to each point in $x_+ \in \supp C(\varphi,\Sigma_+)$. As $\supp\varphi \cap J^+(\Sigma_+)\subseteq J^+( \supp C(\varphi,\Sigma_+) )$, we know that $\supp\varphi \cap J^+(\Sigma_+)$ is contained in the future of each point $x_- \in \supp C(\varphi,\Sigma_-)$. Similarly, we know that $\supp\varphi \cap J^-(\Sigma_-)$ is contained in the past of each point $x_+ \in \supp C(\varphi,\Sigma_+)$. This further implies that $\supp\varphi \cap J^{\pm}(\Sigma_{\pm})$ is contained in the future/past of any point $y_{\mp}\in \supp\varphi \cap J^{\mp}(\Sigma_{\mp})$, and hence if we can find a pair of mutually spacelike points $z_{\pm}\in J^{\pm}(\Sigma_{\pm})$ (there may not be any such pair, e.g. a cylinder spacetime with $\Sigma_{\pm}$ sufficiently separated in time), we know that at least one of the points is not in $\supp\varphi$. Taking the lab $K \supset J^+(\Sigma_-)\cap J^-(\Sigma_+)\cap\supp\varphi$, we then know that there are no pairs of mutually spacelike points $z_{\pm}\in K^{\pm}$ which are both in $\supp\varphi$, and hence by Claim~\ref{claim:Sorkin scenario equivalence 1} we know that no Sorkin scenario exists for the lab $K$.

\end{proof}

\begin{claim}\label{claim:sorkin lab K exists for any mode}
    For any mode $\varphi\in\text{Sol}(\spctm{M})$, and any spacelike Cauchy surface $\Sigma$, there is a lab $K$ for $\varphi$, containing $\supp C(\varphi,\Sigma)$, for which a Sorkin scenario exists.
\end{claim}

\begin{proof}
    While not important to the proof, we first note that by Theorem 3.78 in~\cite{minguzzi2008causal} it is always possible to foliate a globally hyperbolic spacetime with smooth spacelike Cauchy hypersurfaces. Now, since $\Sigma$ is spacelike, any pair of distinct points $x,y\in \supp C(\varphi,\Sigma)$ are mutually spacelike. By smoothness of $\varphi$ we can then find a further pair of spacelike points, $x_+$ to the future of $x$, and $x_-$ to the past of $y$, such that $x_{\pm}\in\supp \varphi$. We can then find some pair of Cauchy surfaces, $\Sigma_{\pm}$, after and before $\Sigma$ respectively, and close enough to $\Sigma$ such that $x_{\pm}\in I^{\pm}(\Sigma_{\pm})$. One can further find some $f\in C^{\infty}_0(\spctm{M})$ supported in $I^+(\Sigma_-)\cap I^-(\Sigma_+)$ such that $\varphi = \Delta f$. Finally, we can then take any compact $K\subset J^+(\Sigma_-)\cap J^-(\Sigma_+)$ and containing $\supp f$. Such a $K$ is a lab for $\varphi$, and, by construction, there exists a pair of spacelike points $x_{\pm}\in K^{\pm}\cap\supp \varphi$, and so by Claim~\ref{claim:Sorkin scenario equivalence 1} we know there exists a Sorkin scenario for this mode-lab pair.
\end{proof}

As we show in Section~\ref{sec:The acausality of an ideal measurement of a smeared field}, the existence of a Sorkin scenario leads to the fact that an ideal measurement of the associated smeared field, $\phi(f)$, violates causality, and is thus impossible. The above claim, then, is essentially saying that for any given `reference frame' (defined by a foliation of the spacetime into spacelike Cauchy surfaces), and some point of (coordinate) time in that reference frame (defined by a given Cauchy surface in that foliation), there exists a designated lab subset which is spatially `wide' enough to measure the mode $\varphi$ (in the sense that is contains $\supp C(\varphi , \Sigma)$), but too `brief' in time (in the sense that the proof involves making $K$ shorter in time than some time scale related to the mode $\varphi$) for one to make an ideal measurement of the associated smeared field.

This seems physically reasonable. Another question that comes to mind is whether there exists, for the given mode $\varphi$, a lab $K$ for which no Sorkin scenario exists, and hence a lab in which an ideal measurement is in principle possible. Furthermore, one may wonder whether such a lab exists for all modes in a given spacetime. By Claim~\ref{claim:Sorkin scenario equivalence 2} such a lab exists for a mode $\varphi$ if the supports of $\varphi$'s Cauchy data on some pair of ordered Cauchy surfaces are totally ordered. Any mode not satisfying this property for any pair of ordered Cauchy surfaces is therefore impossible to measure (in an ideal measurement sense) anywhere in spacetime.

\begin{claim}\label{claim:massless KG field sorkin scenario always}
    For a massless real scalar field in $1+1$-Minkowski spacetime, denoted 
 by $\mathbb{R}^{1,1}$, a Sorkin scenario exists for every mode-lab pair.
\end{claim}

\begin{proof}
    As previously stated, to prove this claim we need to show that Sorkin scenarios are possible for every mode $\varphi\in\text{Sol}(\mathbb{R}^{1,1})$ satisfying the massless Klein-Gordon equation $\Box\varphi = 0$, and every lab $K$ in which $\varphi$ can be sourced.

    Changing to lightcone coordinates $v = t+x$ and $u = t-x$ we find that any smooth solution $\varphi(u,v)$ must satisfy $\partial_u \partial_v \varphi = 0$. By integrating this equation over $[0,u]\times [0,v]$ we can rearrange to get
    \begin{equation}
        \varphi(u,v) = \varphi(0,v) + \varphi(u,0) - \varphi(0,0) \; .
    \end{equation}
    By defining $V(v) := \varphi(0,v) - \varphi(0,0)$ and $U(u) := \varphi(u,0)$ we see that any solution $\varphi$ can be decomposed as $\varphi(u,v) = U(u) + V(v)$, i.e. as a sum of a function of $u$ and a function of $v$. Note that $U,V\in C^{\infty}(\mathbb{R})$, but they are not necessarily compactly supported. Since $\varphi$ is necessarily not a constant, we can always ensure (by swapping $u$ and $v$ coordinates if needed) that $V$ is necessarily not a constant function.

    As we are interested in spatially compact solutions we can further constrain $U$ and $V$. In $(t,x)$ coordinates a general solution is of the form $\varphi(t,x) = U(t-x) + V(t+x)$. If the intersection of $\supp \varphi$ with any spatial slice (e.g. any $t=\text{const}$ surface) is a compact subset of that spatial slice, then the Cauchy data for $\varphi$ on any $t=\text{const}$ slice must be compactly supported on that slice. Consider the slice $t=0$ and denote the Cauchy data as $\phi(x) := \varphi(0,x)$ and $\pi(x) := (\partial_t\varphi)(0,x)$. We then have $\phi(x) = V(x) + U(-x)$ and $\pi(x) = V'(x) + U'(-x)$.
    
    The Cauchy data on this slice is compactly supported, i.e. $\phi , \pi\in C^{\infty}_0 (\mathbb{R})$. Therefore, there exists some $R>0$ for which, for all $|x|\geq R$, $\phi(x) = \pi(x) = 0$. Therefore, for all $|x|\geq R$, we get i) $V(x) = - U(-x)$ and ii) $V'(x) = - U'(-x)$. Taking a derivative of i) at any point $|x| > R$ we get $V'(x) = U'(-x)$, which by ii) implies $U'(-x)=0$ for all $|x| > R$. That is, $U$ is constant for a large enough $|x|$, and thus, by i), so is $V$. Furthermore, from i) we see that the constants have to match (at opposite ends) between $U$ and $V$.

    Let $v_+$ be the smallest value such that $V(v)$ is constant for any $v>v_+$. Similarly, let $v_-$ be the largest value such $V(v)$ is constant for any $v<v_-$. Recall that $V$ is necessarily not a constant function. Thus, $v_- < v_+$. If $U$ is not a constant function we can similarly define $u_{\pm}$ for $U$, otherwise we set $u_{\pm} = 0$.
    
    Without changing $\varphi$ we can redefine $V$ and $U$ by adding some constant to $V$ and subtracting it from $U$, and thus, without loss of generality we can always choose $U$ and $V$ such that $U(u)=0$ for any $u>u_+$. i) then implies that $V(v) = 0$ for any $v<v_-$. From the smoothness of $V$, the fact that $V$ is necessarily not a constant function, and the definition of $v_-$, we know that there is then some interval $(v_- , v_{-+})$ (with $v_{-+} < v_+$) in which $V\neq 0$. Therefore, at the points $(u,v)$ for which $u>u_+$ and any $v\in (v_{-} , v_{-+})$, we know that $U(u)=0$ and $V(v)\neq 0$, and hence $\varphi(u,v)\neq 0$. Let $S_+$ denote the set of these points.

    We can repeat this argument again but with the alternative choice that $U(u) = 0$ for any $u<u_-$ (by adding some constant to $V$ and subtracting it from $U$). Thus $V(v) = 0$ for any $v>v_+$. We then know that there is some interval $(v_{+-},v_+)$ (with $v_{+-} > v_{-+}$) in which $V\neq 0$, and so at any point $(u,v)$ for which $u<u_-$ and $v\in (v_{+-},v_+)$ we have $\varphi(u,v)\neq 0$. Let $S_-$ denote the set of these points.

    By construction, the sets $S_+$ and $S_-$ are mutually spacelike, and have non-zero intersection with any constant time slice for large enough $|t|$. Thus, for any lab $K$ for $\varphi$, we know by the compactness of $K$ that we can find two time coordinates $t_{\pm}$ such that the slab $[t_- , t_+]\times \mathbb{R}$ strictly contains $K$. By the above argument we can then find a pair of mutually spacelike points $s_{\pm}\in S_{\pm}$ with time coordinates greater/smaller than $t_{\pm}$, and so $s_{\pm}\in K^{\pm}$. We also have that $s_{\pm}\in\supp\varphi$. By Claim~\ref{claim:Sorkin scenario equivalence 1} we then know a Sorkin scenario exists for any such lab $K$. Finally, we note that this whole argument applies to any mode $\varphi$.
\end{proof}

\begin{claim}\label{claim:massive KG field sorkin scenario always}
    For a real scalar field of mass $m > 0$ in $1+1$-Minkowski spacetime, denoted 
 by $\mathbb{R}^{1,1}$, a Sorkin scenario exists for every mode-lab pair.
\end{claim}

\begin{proof}
    Consider any mode $\varphi\in \text{Sol}(\mathbb{R}^{1,1})$ satisfying $P\varphi = ( \Box + m^2 )\varphi = 0$, where in lightcone coordinates, $u = t-x$ and $v = t+x$, we have $\Box = 2\partial_u \partial_v$. The retarded Green function in these coordinates is
    \begin{equation}\label{eq:massive Green function}
        G(u,v ; u' , v') = \frac{1}{2}\theta (u-u')\theta(v-v')J_0 (m \sqrt{2(u-u')(v-v')}) \; ,
    \end{equation}
    where $J_0$ is the 0-index Bessel function of the first kind, and $\theta$ is the Heaviside step function. One can verify that $(PG)(u,v;u',v') = \delta(u-u')\delta(v-v')$.

    To show that a Sorkin scenario exists for every lab $K$ for $\varphi$, we can, by Claim~\ref{claim:Sorkin scenario equivalence 2}, show that there do not exist two ordered Cauchy surfaces $\Sigma_{\pm}$ for which the supports of the corresponding Cauchy data are totally ordered. We can do this by way of contradiction.

    Assuming there are two such surfaces $\Sigma_{\pm}$ for which the supports $C_{\pm} = \supp C(\varphi , \Sigma_{\pm})$ are totally ordered (no pair of mutually spacelike points), we can then see that there exists a $v = \text{const}$ surface for which $v' \geq v$ for any $(u',v')\in C_+$, and $v' \leq v$ for any $(u',v')\in C_-$. There similarly exists a $u = \text{const}$ surface which is less/greater than or equal to the $u$ coordinate of any point in $C_{+/-}$. We can always shift the origin of our $(u,v)$ coordinates such that these two surfaces are $v=0$ and $u=0$ respectively. We then have that $C_{\pm}\subset J^{\pm}((0,0))$, where $(0,0)$ is the new origin of our $(u,v)$ coordinates.

    We next note that we can source the solution $\varphi$ with a test function $f$ entirely supported in $J^-((0,0))$, i.e. $u,v \leq 0$. To see this we first note that there exists some $s_-< 0$ for which $C_-$ is in the future of the $t=s_-$ surface (as $C_-$ is compact), which we denote here by $\tilde{\Sigma}_-$. We can then take any smooth function $\chi$ which is $1$ on and to the future of $\tilde{\Sigma}_-$, and is 0 to the past of some other $t=s'_-$ surface for $s'_- < s_-$. Setting $f = P( \chi \varphi)$ then gives us a smooth test function which satisfies $\varphi = \Delta f$, and which is compactly supported in $J^-(C_-)\cap J^-(\tilde{\Sigma}_-)$ (as the support of $\varphi$ to the past of $\tilde{\Sigma}_-$ is contained in the past of $C_-$). Since $C_-\subset J^-((0,0))$, then $J^-(C_-)\subseteq J^-((0,0))$, and hence so is $\supp f$ as desired.

    Now consider some $s_+ > 0$ for which the $t=s_+$ surface, denoted by $\tilde{\Sigma}_+$, has $C_+$ in its past. Since the support of $\varphi$ to the future of $\tilde{\Sigma}_+$ is contained in $J^+(C_+)$, which is itself contained in $J^+((0,0))$, i.e. $u,v\geq 0$, we know that $\varphi (u,v) = 0$ for any $(u,v)$ for which $u\leq 0$, and $v > 2s_+ - u$, or $v\leq 0$ and $u > 2s_+ - v$, i.e. any point outside the future lightcone of $C_+$ and later in time than $s_+$. Consider the case where $u\leq 0$ and $v > 2s_+ - u$, as the other will be identical. At any such point $(u,v)$ we get $0 = \varphi(u,v) = (\Delta f)(u,v) = (G f)(u,v)$, where $G$ is the retarded Green function in~\eqref{eq:massive Green function}. Writing this out as an integral we get
    \begin{align}
        0 & = \varphi(u,v)\\ \nonumber
        & = \frac{1}{2} \int_{-\infty}^u du' \int_{-\infty}^v dv' \, J_0 (m \sqrt{2(u-u')(v-v')}) f(u',v') 
        \\ \nonumber
        & = \frac{1}{2} \int_{-\infty}^u du' \int_{-\infty}^0 dv' \,  J_0 (m \sqrt{2(u-u')(v-v')}) f(u',v')  \; ,
    \end{align}
    where line 3 follows from the fact that $v>0$ and that $f$ is supported in the portion for which $u,v \leq 0$. We can then see that $\varphi(u,v)$ depends only on $v$ through the Bessel function under the integral, at least for $u\leq 0$ and $v > 2s_+ - u$.

    In fact, the rhs in the final line is analytic in $v$ for $v>0$, and hence by the \textit{identity theorem} for analytic function (i.e. any analytic function that vanishes in an interval must vanish everywhere), the fact that $\varphi(u,v) = 0$ for $v > 2s_+ - u$ (with $u\leq 0$) implies $\varphi(u,v) = 0$ for all $v>0$ (with $u\leq 0$).

    We can repeat the same argument for the case where $v\leq 0$ and $u > 2s_+ - v$, and similarly show that $\varphi(u,v) = 0$ for all $u>0$ and $v\leq 0$. We now have that $\varphi$ is zero at any point spacelike to the origin $(0,0)$. Thus, on the Cauchy surface $t=0$ the Cauchy data for $\varphi$ must be zero everywhere apart from on the origin. This contradicts the fact that the Cauchy data on any Cauchy surface must be a smooth compactly supported (and non-trivial) function.

\end{proof}

\subsection{Recovering $\text{tr}(\rho e^{it\phi(g)})$ from an $L^2$-Kraus Measurement}\label{app:recovering tr rho exp from L2 kraus}

Consider some normalised $L^2$ function $k:\mathbb{R}\rightarrow\mathbb{C}$, a smeared field $\phi(g)$ for some $g\in C_0^{\infty}(\spctm{S})$, and the associated $L^2$-Kraus update map $\mathcal{E}_{\phi(g),k}$ given in Example~\ref{ex:L2 kraus} as
\begin{equation}
    \mathcal{E}_{\phi(g),k}(X) = \int_{\mathbb{R}}d\gamma \, k(\phi(g)-\gamma)X k(\phi(g)-\gamma)^{\dagger} \; ,
\end{equation}
for some bounded operator $X\in\alg{B}$.

The map $\mathcal{E}_{\phi(g),k}$ can be interpreted as the update map arising from a non-selective measurement of the observable corresponding to the operator $\phi(g)$, with the particular function $k$ encoding the precise effect of the experiment on future statistics (we assume we know, or have control over, the form of $k$). Given some state $\rho$, the probability that the outcome of this measurement lands in some Borel set $\bor{B}\subset \mathbb{R}$ is given by
\begin{equation}
    \text{prob}(\bor{B}) = \text{tr}(\rho \mathcal{E}_{\phi(g),k,\bor{B}}(\mathds{1})) \; ,
\end{equation}
where we have defined the map
\begin{equation}
    \mathcal{E}_{\phi(g),k,\bor{B}}(X) = \int_{\bor{B}}d\gamma \, k(\phi(g)-\gamma)X k(\phi(g)-\gamma)^{\dagger} \; ,
\end{equation}
for any $X\in\alg{B}$. For any $\bor{B}\in\BorR$ we can write $\text{prob}(\bor{B})$ in terms of the probability density $\mathrm{p}$ as
\begin{equation}
    \text{prob}(\bor{B}) = \int_{\bor{B}}d\gamma \, \mathrm{p}(\gamma) \; .
\end{equation}
Here the explicit form of $\mathrm{p}$ is
\begin{equation}\label{eq:prob density k}
    \mathrm{p}(\gamma) = \text{tr}\left( \rho \, k(\phi(g)-\gamma) k(\phi(g)-\gamma)^{\dagger} \right) = \text{tr}\left( \rho \, \tilde{k}(\phi(g)-\gamma) \right) \; ,
\end{equation}
where we have defined $\tilde{k} = |k|^2$.

The value one measures in a given experiment can then be thought of as the realisation of a random variable $\mathtt{X}$ with probability density $\mathrm{p}$. Given some function $\eta:\mathbb{R}\rightarrow\mathbb{C}$, the expected value of the random variable $\eta(\mathtt{X})$ with respect to this probability density is then
\begin{equation}\label{eq:exp eta of X}
    \mu = \mathbb{E}(\eta(\mathtt{X})) = \int_{\mathbb{R}}d\gamma \, \eta(\gamma )\mathrm{p}(\gamma ) \; .
\end{equation}
This true expected value, $\mu$, can be approximated arbitrarily closely in the following sense.

Consider $N$ independent runs of the experiment with outcomes $(x_1,x_2, ... ,x_N)$. These outcomes are realisations of the independent and identically distributed random variables $\mathtt{X}_n$ for $n=1,...,N$. For each run ($n=1,...,N$) we can compute $y_n = \eta (x_n)$. Let $m_N = \frac{1}{N}\sum_{n=1}^N \eta(x_n)$ denote the sample mean of these values. This sample mean is the realisation of the random variable $\mathtt{M}_N = \frac{1}{N}\sum_{n=1}^N \eta(\mathtt{X}_n)$. Note that $\mathbb{E}(\mathtt{M}_N) = \mu$, and that $\text{Var}(\mathtt{M}_N) = \frac{\sigma^2}{N}$, where $\sigma^2 = \text{Var}(\eta(\mathtt{X})) = \mathbb{E}(\eta(\mathtt{X})^2) - \mu^2$. This expression for the variance of $\mathtt{M}_N$ follows from Bienaymé's identity.

Now, we would like $m_N$ to be arbitrarily close to the true value $\mu$, say $|m_N - \mu| \leq \epsilon$ for some $\epsilon >0$. Since $m_N$ is the realisation of the random variable $\mathtt{M}_N$, we cannot say this for certain. However, through Chebyshev's inequality we can ensure that, with a large enough number of trials $N$, we will only ever have $|m_N - \mu| \geq \epsilon$ with some arbitrarily small probability $\delta >0$. Specifically, Chebyshev's inequality says that
\begin{equation}\label{eq:Chebyshevs inequality}
    \text{prob}\left( |\mathtt{M}_N - \mu | \geq \epsilon \right) \leq \frac{\sigma^2}{\epsilon^2 N} \; .
\end{equation}
Given some $\delta>0$, if we want $\text{prob}\left( |\mathtt{M}_N - \mu | \geq \epsilon \right) \leq \delta$, then we simply need to do $N\geq \frac{\sigma^2}{\epsilon^2 \delta}$ trials.

Now that we have established that one can approximate $\mu = \mathbb{E}(\eta(\mathtt{X}))$ arbitrarily closely with enough trials (in a probabilistic sense), for some function $\eta$, let us determine how this can be used to recover $\text{tr}(\rho e^{it\phi(g)})$.

Recalling~\eqref{eq:prob density k} and~\eqref{eq:exp eta of X} we compute
\begin{align}\label{eq:relationship between eta and k}
    \mathbb{E}(\eta(\mathtt{X})) & = \int_{\mathbb{R}}d\gamma \, \eta(\gamma )\mathrm{p}(\gamma ) \nonumber
    \\
    & = \int_{\mathbb{R}}d\gamma \, \eta(\gamma ) \text{tr}\left( \rho \, \tilde{k}(\phi(g)-\gamma) \right) \nonumber
    \\
    & = \text{tr}\left( \rho \int_{\mathbb{R}}d\gamma \, \eta(\gamma ) \, \tilde{k}(\phi(g)-\gamma) \right) \nonumber
    \\
    & = \text{tr}\left( \rho \left( \eta \ast \tilde{k} \right) (\phi(g)) \right) \; ,
\end{align}
where $(\eta \ast \tilde{k})(\cdot )$ denotes the convolution of the functions $\eta$ and $\tilde{k}$, and we have used functional calculus to define the operator $\left( \eta \ast \tilde{k} \right) (\phi(g))$. Given the $L^2$-Kraus measurement of $\phi(g)$ with the function $k$,~\eqref{eq:relationship between eta and k} provides the explicit relationship between the expected value of some function ($\eta$ in this case) of the experimental outcomes and the `quantum' expectation value.

Given we want to recover $\text{tr}(\rho e^{it\phi(g)})$, we thus need to find the function $\eta$ such that $(\eta\ast \tilde{k})(\cdot) = e^{it(\cdot)}$, for some $t\in\mathbb{R}$. One can check that
\begin{equation}\label{eq:eta that does the job}
    \eta(\cdot) =\frac{e^{it(\cdot)}}{\sqrt{2\pi}\fwt{F}\lbrace \tilde{k} \rbrace (-t)} \;,
\end{equation}
does the job. Recall that $\fwt{F}\lbrace f\rbrace$ denotes the Fourier transform of $f$ (see~\eqref{eq:def FT}). 

In summary, to recover $\text{tr}(\rho e^{it\phi(g)})$ via $L^2$-Kraus measurements of $\phi(g)$ with the function $k$, we need to do $N$ trials, and for each outcome $x_n$ we compute $\eta(x_n)$, where $\eta$ is given explicitly in~\eqref{eq:eta that does the job}. We then compute the sample mean of the values $\eta(x_n)$, and, given $N$ is large enough (c.f.~\eqref{eq:Chebyshevs inequality}), we then know that our sample mean will be within $\pm\epsilon$ of the value $\text{tr}(\rho e^{it\phi(g)})$ with some probability $1-\delta$.

\subsection{Constant Weierstrass Transform Implies Constant Input}\label{app:Constant Weierstrass transform implies constant input}

\begin{claim}
Consider some bounded (measurable) function $\zeta:\mathbb{R}\times\mathbb{R}\rightarrow\mathbb{C}$. If $\fwt{W}\lbrace \zeta(\cdot , t) \rbrace (s+it)$ (for $t\in\mathbb{R}$) is constant in $s\in\mathbb{R}$, then $\zeta(\cdot , t)$ is an almost everywhere constant function.
\end{claim}

\begin{proof}
We write
\begin{equation}
    \fwt{W}\lbrace \zeta(\cdot , t) \rbrace (s+it) = \int_{\mathbb{R}}dx \, f(s,t,x) \; ,
\end{equation}
where $f(s,t,x) = \frac{1}{\sqrt{4\pi}}e^{-\frac{(x-s-it)^2}{4}}\zeta(x , t)$. Next, we note that, for any $z\in\mathbb{C}$, 
\begin{align}\label{eq:bounded absolute integral of f}
    \int_{\mathbb{R}}dx | f(z,t,x)| & = \int_{\mathbb{R}}dx \frac{1}{\sqrt{4\pi}}e^{\frac{1}{4}(t+\text{Im}(z))^2}  e^{-\frac{(x-\text{Re}(z))^2}{4}}|\zeta(x,t)| \nonumber
    \\
    & \leq e^{\frac{1}{4}(t+\text{Im}(z))^2}\norm{\zeta}_{\infty}\int_{\mathbb{R}}dx \frac{1}{\sqrt{4\pi}}  e^{-\frac{(x-\text{Re}(z))^2}{4}} \nonumber
    \\
    & = e^{\frac{1}{4}(t+\text{Im}(z))^2}\norm{\zeta}_{\infty} \nonumber
    \\
    & < \infty \; ,
\end{align}
where in line 2 we have used the fact that $\zeta$ is bounded. Thus, 
\begin{equation}
    \fwt{W}\lbrace \zeta(\cdot , t) \rbrace (z+it) = \int_{\mathbb{R}}dx \, f(z,t,x) \; ,
\end{equation}
is defined for any $z\in\mathbb{C}$. We write $F(z)= \fwt{W}\lbrace \zeta(\cdot , t) \rbrace (z+it)$ for convenience. By assumption, $F$ is constant along the real axis of the complex plane ($\text{Im}(z) = 0$).

We next show that $F$ is an entire function, i.e. it is holomorphic on the entire complex plane. To do this we need to show that $\partial_z F(z)$ is defined for all $z\in\mathbb{C}$. We have
\begin{equation}\label{eq:complex derive of F}
    \partial_z F(z) = \partial_z \left( \int_{\mathbb{R}}dx \, f(z,t,x) \right) \; .
\end{equation}
To move the derivative inside the integral we need to show that: \textbf{i)} $f(z,t,x)$ is integrable in $x$ for all $z\in\mathbb{C}$, \textbf{ii)} for almost every $x\in\mathbb{R}$, $\partial_z(f(z,t,x))$ exists for all $z\in\mathbb{C}$, and \textbf{iii)} $|\partial_z(f(z,t,x))|$ is integrable in $x$ for all $z\in\mathbb{C}$. We already have \textbf{i)} from~\eqref{eq:bounded absolute integral of f}, and \textbf{ii)} follows as $f(z,t,x)$ is an entire function in $z$. 

For \textbf{iii)} we need to show that
\begin{equation}
    \int_{\mathbb{R}}dx |\partial_z(f(z,t,x))| <\infty \; ,
\end{equation}
for all $z\in\mathbb{C}$. We first note that
\begin{equation}
    \partial_z(f(z,t,x)) = \frac{1}{2\sqrt{4\pi}}(x-z-it)e^{-\frac{1}{4}(x-z-it)^2}\zeta(x,t) \; ,
\end{equation}
and hence 
\begin{align}
     |\partial_z(f(z,t,x))| & \leq \frac{1}{2\sqrt{4\pi}}\norm{\zeta}_{\infty} |x-z-i t|e^{\frac{1}{4}(t+\text{Im}(z))^2}e^{-\frac{1}{4}(x-\text{Re}(z))^2} \nonumber
     \\
     & \leq \frac{1}{2\sqrt{4\pi}}\norm{\zeta}_{\infty} (|x-\text{Re}(z)|+|t+\text{Im}(z)|)e^{\frac{1}{4}(t+\text{Im}(z))^2}e^{-\frac{1}{4}(x-\text{Re}(z))^2} \; .
\end{align}
Thus,
\begin{align}
    \int_{\mathbb{R}}dx |\partial_z(f(z,t,x))| & \leq \frac{1}{2\sqrt{4\pi}}\norm{\zeta}_{\infty}e^{\frac{1}{4}(t+\text{Im}(z))^2} \int_{\mathbb{R}}dx \, e^{-\frac{1}{4}(x-\text{Re}(z))^2}(|x-\text{Re}(z)|+|t+\text{Im}(z)|) \nonumber
    \\
    & = \frac{1}{2\sqrt{4\pi}}\norm{\zeta}_{\infty}e^{\frac{1}{4}(t+\text{Im}(z))^2} \int_{\mathbb{R}}dy (|y| +|t+\text{Im}(z)|)e^{-\frac{y^2}{4}}  \nonumber
    \\
    & = \frac{1}{\sqrt{4\pi}}\norm{\zeta}_{\infty}e^{\frac{1}{4}(t+\text{Im}(z))^2} \int_0^{\infty}dy (|y| +|t+\text{Im}(z)|)e^{-\frac{y^2}{4}}  \nonumber
    \\
    & < \infty \; ,
\end{align}
where we have changed variables to $y=x-\text{Re}(z)$ in line 2, and the last line follows as the integral in line 3 is finite.

Now that we have established \textbf{i)}, \textbf{ii)}, and \textbf{iii)}, we can bring the derivative inside the integral in~\eqref{eq:complex derive of F}. We then have
\begin{align}
    \partial_z F(z) & =  \int_{\mathbb{R}}dx \, \frac{1}{\sqrt{4\pi}}\partial_z ( e^{-\frac{(x-z-it)^2}{4}} ) \zeta(x,t) \nonumber
    \\
    & = \int_{\mathbb{R}}dx \, \frac{1}{2\sqrt{4\pi}}(x-z-it)e^{-\frac{1}{4}(x-z-it)^2}\zeta(x,t) \; ,
\end{align}
and by \textbf{iii)} we already know that the integral in the last line exists for any $z\in\mathbb{C}$. Thus, $F(z)$ is holomorphic for all $z\in\mathbb{C}$ --- it is an entire function. 

We can thus write $F(z)$ as the power series
\begin{equation}
    F(z) = \sum_{n=0}^{\infty}a_n z^n \; ,
\end{equation}
which converges everywhere in the complex plane. As $F(z) = \fwt{W}\lbrace \zeta(\cdot , t) \rbrace(z+it)$, we should really write $a_n\equiv a_n(t)$ to account for the $t$ dependence.

Concentrating on the real axis, i.e. $z=s\in\mathbb{R}$, we note that $F(s)$ is constant by assumption, and hence $\partial_s^n F(s) =0$ for all $n\in\mathbb{N}$. Evaluating these derivatives at $s=0$, and using the fact that for any power series we can move the derivatives inside the summation (when evaluated within the interior of the domain of convergence), we get
\begin{equation}
    0 = (\partial_s^n F(s))_{s=0} = n!\, a_n(t) \; .
\end{equation}
Thus, $a_n(t) = 0$ for all $n\in\mathbb{N}$, and hence $F(z) = a_0(t)$, i.e. it is constant throughout the whole complex plane. 

Evaluating $F(z)$ at $z = i(\tau - t)$, for some $\tau\in\mathbb{R}$, we get
\begin{equation}
    a_0(t) = F(z)|_{z=i(\tau - t)} = \fwt{W}\lbrace \zeta(\cdot , t) \rbrace(z+it)|_{z=i(\tau-t)} = \fwt{W}\lbrace \zeta(\cdot , t) \rbrace(i\tau) \; ,
\end{equation}
Defining $\tilde{\zeta}(\cdot ,t) = \zeta(\cdot , t) - a_0(t)$, we then have $\fwt{W}\lbrace \tilde{\zeta}(\cdot , t) \rbrace(i\tau) = 0$ by the linearity of the Weierstrass transform. This is true for all $\tau\in\mathbb{R}$. Thus,
\begin{align}
    0 & = \fwt{W}\lbrace \tilde{\zeta}(\cdot , t) \rbrace (i \, 2\tau ) \nonumber
    \\
    & = \int_{\mathbb{R}}dx \frac{1}{\sqrt{4\pi}} e^{-\frac{(x-i\, 2\tau )^2}{4}} \tilde{\zeta}(x , t) \nonumber
    \\
    & = \frac{e^{\tau^2}}{\sqrt{2}}\frac{1}{\sqrt{2\pi}} \int_{\mathbb{R}}dx e^{i\tau x}e^{-\frac{x^2}{4}}\tilde{\zeta}(x , t) \nonumber
    \\
    & = \frac{e^{\tau^2}}{\sqrt{2}} \fwt{F}\lbrace e^{-\frac{(\cdot )^2}{4}}\tilde{\zeta}(\cdot , t) \rbrace (\tau) \; ,
\end{align}
for all $\tau\in\mathbb{R}$. Recall that $\fwt{F}\lbrace f\rbrace(\tau)$ denotes the Fourier transform defined in~\eqref{eq:def FT}. We then have
\begin{equation}\label{eq:vanishing FT}
    \fwt{F}\lbrace e^{-\frac{(\cdot )^2}{4}}\tilde{\zeta}(\cdot , t) \rbrace (\tau) = 0 \; ,
\end{equation}
for all $\tau\in\mathbb{R}$. Note that the function $e^{-\frac{x^2}{4}}\tilde{\zeta}(x , t)$ is square-integrable in $x$, as
\begin{align}
    \int_{\mathbb{R}}dx \, e^{-\frac{x^2}{2}}|\tilde{\zeta}(x , t)|^2 & \leq \norm{\tilde{\zeta}}_{\infty}\int_{\mathbb{R}}dx \, e^{-\frac{x^2}{2}} \nonumber
    \\
    & = \norm{\tilde{\zeta}}_{\infty}\sqrt{2\pi} \nonumber
    \\
    & \leq (\norm{\zeta}_{\infty} + |a_0(t)| )\sqrt{2\pi} \nonumber
    \\
    & < \infty \; ,
\end{align}
Now, since $e^{-\frac{x^2}{4}}\tilde{\zeta}(x,t)$ is in $L^2(\mathbb{R};\mathbb{C})$, and since the Fourier transform and its inverse are unitary operators on $L^2(\mathbb{R};\mathbb{C})$, we can take an inverse Fourier transform of~\eqref{eq:vanishing FT} to get that $e^{-\frac{x^2}{4}}\tilde{\zeta}(x,t) = 0$ for almost every $x\in\mathbb{R}$. This implies that $\tilde{\zeta}(\cdot , t) = \zeta(\cdot , t) - a_0(t) = 0$ almost everywhere, and thus $\zeta (\cdot , t) = a_0(t)$ almost everywhere, i.e. $\zeta(\cdot , t)$ is an almost everywhere constant function as desired.

\end{proof}

\subsection{Non-Triviality of Borel Set Self-Intersection}\label{app:Non-triviality of Borel set self-intersection}

Here $\lambda(S)$ denotes the \textit{Lebesgue measure} of some subset $S\subseteq \mathbb{R}$, and $S\triangle S' = (S\setminus S')\cup (S'\setminus S)$ denotes the \textit{symmetric difference} of two subsets $S,S'\subseteq\mathbb{R}$. We write $S\approx S'$ if $\lambda(S\triangle S')=0$, i.e. the two subsets $S,S'\subseteq\mathbb{R}$ differ on a set of Lebesgue measure zero.

We consider non-trivial resolutions of the form $\bor{R} = \lbrace \bor{B}_n \rbrace_{n\in I}$, where $I$ is some countable indexing set. That is, we consider resolutions with at least two bins, i.e. the cardinality $|I|\geq 2$, and where $\lambda(\bor{B}_n)\neq 0$ for each $n\in I$. For a given resolution $\bor{R}$ we define the set
\begin{equation}
    \bor{R}_t = \cup_{n\in I}\bor{B}_n \cap ( \bor{B}_n +t ) \; ,
\end{equation}
for each $t\in\mathbb{R}$. We then make the following:
\begin{claim}[Non-triviality of $\bor{R}_t$ for some $t\in\mathbb{R}$]
Given some non-trivial resolution $\bor{R} = \lbrace \bor{B}_n \rbrace_{n\in I}$, there exists a $t\in\mathbb{R}$ such that $\bor{R}_t \not\approx \mathbb{R}$ and $\bor{R}_t \not\approx\varnothing$.
\end{claim}

\begin{proof}
We show this with two sub-claims, Claim~\ref{claim:existence of T at which measure is not full} and~\ref{claim:continuity of measure function}.

In Claim~\ref{claim:existence of T at which measure is not full} we establish that there is some $T\neq 0$ for which $\bor{R}_T$ is not of full measure, i.e. its complement is of non-zero measure. In particular, we show there exists some closed interval $D\subset \mathbb{R}$ such that $D\cap \bor{R}_T$ is not of full measure, and hence of measure less than $\lambda (D)$. This means that the function $t\mapsto \lambda (D\cap \bor{R}_t)$, which starts at $\lambda(D\cap\bor{R}_0) = \lambda(D\cap\mathbb{R})=\lambda(D)$ at $t=0$, must decrease to $\lambda(D\cap \bor{R}_T)<\lambda(D)$ at $t=T$. It may even be zero at $t=T$, which could be due to the fact that $\bor{R}_T \approx \varnothing$ (this was the case in Example~\ref{ex:Integer bin resolution} for $t=1$). Thus, Claim~\ref{claim:existence of T at which measure is not full} is not enough on its own to prove the existence of some $t\in\mathbb{R}$ for which $\bor{R}_t \not\approx \mathbb{R}$ \textit{and} $\bor{R}_t \not\approx\varnothing$.

For that we need to combine Claim~\ref{claim:existence of T at which measure is not full} with Claim~\ref{claim:continuity of measure function}. In the latter we show that the function $t\mapsto \lambda(D \cap \bor{R}_t)$ is continuous. Thus, even if it is zero at $t=T$, it must meet every value in between $\lambda(D)$ and $0$ as $t$ goes from $t=0$ to $t=T$ (this follows from the Intermediate Value Theorem). This implies that there will be some $\tau$ between $0$ and $T$ for which $\lambda(D \cap \bor{R}_{\tau}) = \lambda(D)/2$ say, and hence $\bor{R}_{\tau}$ cannot be equivalent to $\varnothing$ or $\mathbb{R}$, thus proving the main claim.

\end{proof}

\begin{claim}\label{claim:existence of T at which measure is not full}
There exists a closed interval $D\subset\mathbb{R}$, and some non-zero $T\in\mathbb{R}$, such that $\lambda(D\cap \bor{R}_T)<\lambda(D)$.
\end{claim}

\begin{proof}
By assumption, there are at least two Borel sets in the resolution $\bor{R}$ of non-zero measure. This means that, for any $n\in I$ (where $I$ is the countable indexing set for the resolution), we know that both $\bor{B}_n$ and its complement have non-zero measure, i.e. $\lambda(\bor{B}_n),\lambda( \mathbb{R}\setminus\bor{B}_n )\neq 0$.

As $\lambda(\bor{B}_n)\neq 0$, by the Lebesgue density theorem we know there exists a point $x\in\bor{B}_n$ such that, for all $\alpha <1$, there exists a $d'>0$ such that
\begin{equation}
    \lambda([x-d , x+ d ]\cap \bor{B}_n ) > \alpha \lambda([x-d,x+d]) = \alpha \, 2 d \; ,
\end{equation}
for all $d<d'$. Similarly, since $\lambda(\mathbb{R}\setminus \bor{B}_n)\neq 0$, we know (by the Lebesgue density theorem again) that there exists a point $y\in\mathbb{R}\setminus \bor{B}_n$, such that for all $\alpha <1$, there exists a $d'' >0$ such that
\begin{equation}
    \lambda([y-d , y+ d ]\cap \bor{B}_n ) < (1-\alpha ) \lambda([y-d,y+d]) = (1 - \alpha ) \, 2 d \; ,
\end{equation}
for all $d< d''$.

For some $\alpha < 1$, say $\alpha = 1/2$, we can then find such a $d'$ and $d''$ as above and set $\tilde{d} = \min \lbrace d' , d'' \rbrace$, so that for all $d<\tilde{d}$ we have the following chain of inequalities:
\begin{equation}\label{eq:chain of inequalities for closed intervals}
    \lambda([y-d , y+ d ]\cap \bor{B}_n ) < d < \lambda([x-d , x+ d ]\cap \bor{B}_n ) \; .
\end{equation}
Picking some $d<\tilde{d}$ we define the closed interval $D=[x-d,x+d]$. Setting $T=x-y$ (note that since $x\in \bor{B}_n$ and $y\in\mathbb{R}\setminus\bor{B}_n$, it must be the case that $x\neq y$, and hence $T\neq 0$) we then have that $[y-d,y+d] = [x-d,x+d]-T = D - T$, and so $\lambda(D \cap (\bor{B}_n + T)) = \lambda((D-T)\cap \bor{B}_n) < \lambda (D \cap B_n )$ by~\eqref{eq:chain of inequalities for closed intervals}. This further implies that
\begin{equation}
    \lambda( D\cap \bor{B}_n \cap (\bor{B}_n +T) ) \leq \lambda( D \cap (\bor{B}_n +T) ) < \lambda(D \cap \bor{B}_n) \; ,
\end{equation}
and hence
\begin{equation}\label{eq:useful inequality for closed interval}
    \lambda( D\cap \bor{B}_n \cap (\bor{B}_n +T) ) - \lambda(D \cap \bor{B}_n) < 0 \; .
\end{equation}
Finally, this means that
\begin{align}
    \lambda( D \cap \bor{R}_T) & = \lambda(D \cap ( \cup_{m\in I} \bor{B}_m \cap (\bor{B}_m + T) ) )
    \nonumber
    \\
    & = \sum_{m\in I} \lambda( D\cap \bor{B}_m \cap (\bor{B}_m +T) )
    \nonumber
    \\
    & = \lambda( D\cap \bor{B}_n \cap (\bor{B}_n +T) ) + \sum_{m\neq n} \lambda( D\cap \bor{B}_m \cap (\bor{B}_m +T) )
    \nonumber
    \\
    & \leq \lambda( D\cap \bor{B}_n \cap (\bor{B}_n +T) ) + \sum_{m\neq n} \lambda( D\cap \bor{B}_m ) 
    \nonumber
    \\
    & = \lambda( D\cap \bor{B}_n \cap (\bor{B}_n +T)  + \lambda(D \cap (\cup_{m\neq n} \bor{B}_m))
    \nonumber
    \\
    & = \lambda( D\cap \bor{B}_n \cap (\bor{B}_n +T) + \lambda(D \cap ( \mathbb{R}\setminus \bor{B}_n ) )
    \nonumber
    \\
    & = \lambda( D\cap \bor{B}_n \cap (\bor{B}_n +T) + \lambda(D) - \lambda(D \cap \bor{B}_n )
    \nonumber
    \\
    & < \lambda(D) \; ,
\end{align}
as desired. Note, we have used~\eqref{eq:useful inequality for closed interval} to get the last line.

\end{proof}

\begin{claim}\label{claim:continuity of measure function}
For some closed interval $D\subset\mathbb{R}$, the function $t\in\mathbb{R}\mapsto \lambda(D \cap \bor{R}_t)$ is continuous.
\end{claim}

\begin{proof}
For brevity, let $L(t)$ denote the function $t\in\mathbb{R}\mapsto \lambda(D \cap \bor{R}_t)$, and let $L_n(t) =\lambda(D\cap \bor{B}_n \cap (\bor{B}_n + t))$. We will show continuity of $L$ in three steps. \textbf{1)} We first show continuity of the simpler function $t\in\mathbb{R}\mapsto \lambda(D\cap \bor{B}\cap(U+t))$, for any Borel set $\bor{B}\in\BorR$, and any finite union of disjoint open intervals $U = \cup_{i=1}^M O_i$. \textbf{2)} Next, we use this result to show continuity of $L_n$. \textbf{3)} We then show that $L$ can be approximated by a finite sum of the functions $L_n$, for some finite set of indices $n\in I$. This, together with the fact that each $L_n$ is continuous, allows us to finally show that $L$ is continuous.

\vspace{5mm}
\noindent\textbf{1)} Define $\tilde{L}(t)= \lambda(D\cap \bor{B} \cap (U+ t) )$ for $t\in\mathbb{R}$, where $\bor{B}$ is some non-zero measure Borel subset, and $U=\cup_{i=1}^M O_i$ is some finite union of disjoint open intervals.

Consider the characteristic function $1_{D\cap \bor{B} \cap (U+ t)}(x)$ for $x\in\mathbb{R}$. If $x\in D \cap \bor{B} \cap (U+t )$ then $x\in U+t$. Since $U$ is a union of intervals it is open, and so there exists a small open interval $(x-r,x+r)$ for some $r>0$, for which $(x-r,x+r)\subseteq U+t$. This means that for any $|\delta t|<r$, $x-\delta t\in U+t$ also. Equivalently, $x \in U+t + \delta t$ for all $|\delta t|< r$. Thus,
\begin{equation}
    \lim_{\delta t \rightarrow 0}1_{D\cap \bor{B} \cap (U+ t + \delta t)}(x) = 1_{D\cap \bor{B} \cap (U+ t)}(x) = 1 \; ,
\end{equation}
for all $x\in D \cap \bor{B} \cap (U+t )$.

For $x\notin U+t$, then $1_{D\cap \bor{B} \cap (U+ t)}(x) = 0$. If we further assume that $x$ is not a boundary point of $U+t$ then there is again an open interval $(x-r,x+r)$, for some $r>0$, for which $(x-r,x+r)\cap (U+t) = \varnothing$, and hence $x\notin U+t+\delta t$ for all $|\delta t|< r$. For such points we then get
\begin{equation}
    \lim_{\delta t \rightarrow 0}1_{D\cap \bor{B} \cap (U+ t + \delta t)}(x) = 1_{D\cap \bor{B}\cap (U+ t)}(x) = 0 \; .
\end{equation}

If, on the other hand, $x$ is a boundary point of $U+t$, then $1_{D\cap \bor{B} \cap (U+ t)}(x) = 0$ (as $U+t$ is open and so $x$ --- being a boundary point of $U+t$ --- is not in $U+t$), but the above limit may give $1$ if it is approached from the side in the interior of $U+t$. Fortunately, the boundary points of $U+t$ are of measure zero, since $U+t$ is a finite union of disjoint open intervals, and so the number of boundary points is finite and hence countable, and any countable set has measure zero.

We will now apply the Dominated Convergence Theorem to show that $\tilde{L}(t)$ is continuous. First, we note that the sequence of functions $1_{D\cap \bor{B} \cap (U+ t + \delta t_n)}(x)$, for any sequence $(\delta t_n)_{n}$ that converges to $0$ as $n\rightarrow \infty$, is dominated by the characteristic function $1_{D\cap \bor{B}}(x)$ (which is integrable as $D\cap \bor{B}$ has finite Lebesgue measure). Next, since $1_{D\cap \bor{B} \cap (U+ t + \delta t_n)}(x)$ converges pointwise almost everywhere to $1_{D\cap \bor{B} \cap (U+ t)}(x)$, we get by the Dominated Convergence Theorem that
\begin{equation}
    \lim_{n\rightarrow \infty}\int_{\mathbb{R}}dx \, 1_{D\cap \bor{B} \cap (U+ t + \delta t_n)}(x) = \int_{\mathbb{R}}dx \, 1_{D\cap \bor{B} \cap (U+ t )}(x) \; ,
\end{equation}
which, by definition of the Lebesgue integral, is equivalent to the statement that
\begin{equation}
    \lim_{n\rightarrow \infty}\lambda(D \cap \bor{B}\cap (U+t + \delta t_n)) = \lambda(D \cap \bor{B}\cap (U+t)) \; .
\end{equation}
Since this is true for any such sequence $(\delta t_n)_{n}$, we have 
\begin{equation}
    \lim_{\delta t\rightarrow 0}\lambda(D \cap \bor{B}\cap (U+t + \delta t)) = \lambda(D \cap \bor{B}\cap (U+t)) \; ,
\end{equation}
and so $\tilde{L}(t) = \lambda(D \cap \bor{B}\cap (U+t))$ is continuous.

\vspace{5mm}
\noindent\textbf{2)} Fix $n\in I$. To show continuity of $L_n$ at $t\in\mathbb{R}$ we need to show that, for any $\varepsilon >0$, there exists a $\delta >0$ such that, for all $|\delta t|<\delta$, $|L_n(t+\delta t) - L_n(t)|<\varepsilon$.

Pick any $t\in\mathbb{R}$ and any $\varepsilon >0$. Next, pick some finite $\delta' >0$. For all $|\delta t|<\delta'$ we can write
\begin{equation}\label{eq:L tilde with tilde Bn}
    L_n(t+\delta t) = \lambda (D\cap\bor{B}_n\cap(\bor{B}_n + t+\delta t) )= \lambda (D\cap\bor{B}_n\cap(\tilde{\bor{B}}_n + t+\delta t) )\; ,
\end{equation}
for some bounded $\tilde{\bor{B}}_n$ which does not depend on the precise value of $|\delta t|<\delta'$. One can check that
\begin{equation}\label{eq:Bn tilde}
    \tilde{\bor{B}}_n = \bor{B}_n \cap \left( \cup_{\delta t'\in[-\delta' , \delta']} D-t-\delta t' \right) \; ,
\end{equation}
does the job. Intuitively, this satisfies~\eqref{eq:L tilde with tilde Bn} as $\tilde{\bor{B}}_n$ is essentially a restriction of (the possibly unbounded) $\bor{B}_n$ to the (bounded) set of points in $\bor{B}_n$ which overlap with $D$ after translation by $t+\delta t$ for some $|\delta t|<\delta'$.

As $\tilde{\bor{B}}_n$ is a bounded Borel set it can be approximated arbitrarily closely by a finite disjoint union of open intervals, denoted here by $U_n$. In particular, we can pick $U_n$ such that $\tilde{\bor{B}}_n$ and $U_n$ differ on a set of measure less than $\varepsilon/3$, i.e.
\begin{equation}\label{eq:tilde Bn approximated by Un}
    \lambda(\tilde{\bor{B}}_n \triangle U_n) < \frac{\varepsilon}{3} \; ,
\end{equation}
where we recall that $\triangle$ denotes the symmetric difference of two sets. The precise set $U_n$ can depend on the value of $\delta'$, as the definition of $\tilde{B}_n$ in~\eqref{eq:Bn tilde} depends on $\delta'$. $U_n$ does not, however, depend on the precise value of $|\delta t|<\delta'$, since $\tilde{B}_n$ does not.

\eqref{eq:tilde Bn approximated by Un} further implies that the sets $D\cap\bor{B}_n\cap(\tilde{\bor{B}}_n + t+\delta t) $ and $D\cap\bor{B}_n\cap(U_n+t+\delta t)$ differ on a set of measure less than $\varepsilon/3$, and hence using~\eqref{eq:L tilde with tilde Bn} we have
\begin{equation}
    |L_n(t+\delta t) - \lambda(D\cap\bor{B}_n\cap(U_n+t+\delta t))| < \frac{\varepsilon}{3} \; ,
\end{equation}
which is true for all $|\delta t|< \delta'$ as $U_n$ does not depend on the precise value of $|\delta t|<\delta'$. Thus,
\begin{equation}\label{eq:almost done for Ln}
    |L_n(t+\delta t) - L_n(t) | < \frac{2\varepsilon}{3} + \big| \lambda(D\cap\bor{B}_n\cap(U_n+t+\delta t)) - \lambda(D\cap\bor{B}_n\cap(U_n+t)) \big|  \; ,
\end{equation}
for $|\delta t|< \delta'$.

We can then appeal to the continuity of the function $t\mapsto \lambda(D\cap\bor{B}_n\cap(U_n+t))$, which follows by replacing $\bor{B}$ with $\bor{B}_n$ and $U$ with $U_n$ in \textbf{1)}. Therefore, we know there exists a $\delta'' >0$ such that, for all $|\delta t|< \delta''$, we have 
\begin{equation}\label{eq:Un less than eps over three}
    \big| \lambda(D\cap\bor{B}_n\cap(U_n+t+\delta t)) - \lambda(D\cap\bor{B}_n\cap(U_n+t)) \big| < \frac{\varepsilon}{3} \; .
\end{equation}
Setting $\delta = \text{min}\lbrace \delta' , \delta''\rbrace$, the inequalities in~\eqref{eq:almost done for Ln} and~\eqref{eq:Un less than eps over three} then imply that, for all $|\delta t|< \delta$,
\begin{equation}
    |L_n(t+\delta t) - L_n(t) | < \varepsilon \; ,
\end{equation}
as desired.

\vspace{5mm}
\noindent\textbf{3)} To show continuity of $L$ at $t\in\mathbb{R}$, we similarly need to show that, for any $\varepsilon >0$, there exists a $\delta >0$ such that, for all $|\delta t|<\delta$, $|L(t+\delta t) - L(t)|<\varepsilon$.

Pick some $t\in\mathbb{R}$ and some $\varepsilon >0$. Note that $0\leq L(t) \leq \lambda (D) <\infty$, and that $L_n(t)\leq L(t)$ for all $n\in I$. Furthermore, as $\sum_{n\in I}L_n(t)$ converges to $L(t)$, and as $L_n(t)\geq 0$ for each $n\in I$, we know that $\sum_{n\in I}L_n(t)$ actually converges absolutely to $L(t)$. Thus, we know there is some finite set $J\subseteq I$ (with $J=I$ only if $I$ is finite) of cardinality $N=|J|$ such that 
\begin{equation}
    | L(t) - \sum_{n\in J} L_n(t) | < \frac{\varepsilon}{3} \; ,
\end{equation}
Therefore, for any $\delta t$,
\begin{align}\label{eq:almost there for L}
    | L(t+\delta t) - L(t) | & <  \frac{2\varepsilon}{3} + \Big|\sum_{n\in J} L_n(t+ \delta t) - L_n (t) \Big|
    \nonumber
    \\
    & \leq \frac{2\varepsilon}{3} + \sum_{n\in J} | L_n(t+ \delta t) - L_n (t) | \; .
\end{align}
By \textbf{2)} we know there exists a $\delta_n >0$ (for each $n\in I$), such that $| L_n(t+ \delta t) - L_n (t) |<\frac{\varepsilon}{3N}$ for all $|\delta t|<\delta_n$. As $J$ is a finite set, $\delta = \text{min}_{n\in J}\lbrace \delta_n \rbrace$ is necessarily non-zero, and hence $| L_n(t+ \delta t) - L_n (t) |<\frac{\varepsilon}{3N}$ for all $|\delta t|<\delta$ and all $n\in J$.~\eqref{eq:almost there for L} then gives
\begin{align}
    | L(t+\delta t) - L(t) | & < \frac{2\varepsilon}{3} + \sum_{n\in J} \frac{\varepsilon}{3 N} 
    \nonumber
    \\
    & = \frac{2\varepsilon}{3} + N \frac{\varepsilon}{3 N}
    \nonumber
    \\
    & = \varepsilon \; ,
\end{align}
as desired.
\end{proof}

\section{References}

\bibliographystyle{alpha}
\bibliography{bibs/refs}

\end{document}